\documentclass[11pt]{article}

\usepackage{amsthm}
\usepackage{graphicx} % support the \includegraphics command and options
\usepackage{array} % for better arrays (eg matrices) in maths

\usepackage{amsmath, amssymb, amsfonts, verbatim}
\usepackage{hyphenat,epsfig,subcaption,multirow}
\usepackage[font=small,labelfont=bf]{caption}
\usepackage{nicefrac}
\usepackage{paralist}

\usepackage[usenames,dvipsnames]{xcolor}
\usepackage[ruled]{algorithm2e}

\DeclareFontFamily{U}{mathx}{\hyphenchar\font45}
\DeclareFontShape{U}{mathx}{m}{n}{
      <5> <6> <7> <8> <9> <10>
      <10.95> <12> <14.4> <17.28> <20.74> <24.88>
      mathx10
      }{}
\DeclareSymbolFont{mathx}{U}{mathx}{m}{n}
\DeclareMathSymbol{\bigtimes}{1}{mathx}{"91}

\usepackage{tcolorbox}
\tcbuselibrary{skins,breakable}
\tcbset{enhanced jigsaw}

\usepackage[normalem]{ulem}
\usepackage[compact]{titlesec}

\definecolor{DarkRed}{rgb}{0.5,0.1,0.1}
\definecolor{DarkBlue}{rgb}{0.1,0.1,0.5}

\usepackage{nameref}
\definecolor{ForestGreen}{rgb}{0.1333,0.5451,0.1333}
\definecolor{Red}{rgb}{0.9,0,0}
\usepackage[linktocpage=true,
	pagebackref=true,colorlinks,
	linkcolor=DarkRed,citecolor=ForestGreen,
	bookmarks,bookmarksopen,bookmarksnumbered]
	{hyperref}
\usepackage[noabbrev,nameinlink]{cleveref}
\crefname{property}{property}{Property}
\creflabelformat{property}{(#1)#2#3}
\crefname{equation}{eq}{Eq}
\creflabelformat{equation}{(#1)#2#3}
\usepackage{bm}
\usepackage{url}
\usepackage{xspace}
\usepackage[mathscr]{euscript}

\usepackage{tikz}
\usetikzlibrary{arrows}
\usetikzlibrary{arrows.meta}
\usetikzlibrary{shapes}
\usetikzlibrary{backgrounds}
\usetikzlibrary{positioning}
\usetikzlibrary{decorations.markings}
\usetikzlibrary{decorations.pathreplacing} % needed for drawing braces
\usetikzlibrary{patterns}
\usetikzlibrary{calc}
\usetikzlibrary{fit}
\usetikzlibrary{decorations}

\usepackage[framemethod=TikZ]{mdframed}

\usepackage[noend]{algpseudocode}
\makeatletter
\def\BState{\State\hskip-\ALG@thistlm}
\makeatother

\usepackage{cite}
\usepackage{enumitem}
\setlist[itemize]{leftmargin=20pt}
\setlist[enumerate]{leftmargin=20pt}

\usepackage[margin=1in]{geometry}

\usepackage{thmtools}
\usepackage{thm-restate}

\newtheorem{theorem}{Theorem}
\newtheorem{lemma}{Lemma}[section]
\newtheorem{proposition}[lemma]{Proposition}
\newtheorem{corollary}[lemma]{Corollary}
\newtheorem{claim}[lemma]{Claim}
\newtheorem{fact}[lemma]{Fact}
\newtheorem{assumption}[lemma]{Assumption}

\newtheorem{definition}[lemma]{Definition}
\newtheorem{problem}{Problem}

\newtheorem*{claim*}{Claim}
\newtheorem*{assumption*}{Assumption}
\newtheorem*{proposition*}{Proposition}
\newtheorem*{lemma*}{Lemma}

\newtheorem{observation}[lemma]{Observation}
\newtheorem{property}{Property}

\newtheorem*{theorem*}{Theorem}

\crefname{lemma}{Lemma}{Lemmas}
\crefname{claim}{claim}{claims}
\crefname{property}{Property}{Properties}
\crefname{invariant}{Invariant}{Invariants}

\newtheorem{mdresult}{Result}

\theoremstyle{definition}

\newtheorem*{mdproblem*}{Problem}
\newenvironment{Problem*}{\begin{mdframed}[hidealllines=false,innerleftmargin=10pt,backgroundcolor=gray!10,innertopmargin=5pt,innerbottommargin=5pt,roundcorner=10pt]\begin{mdproblem*}}{\end{mdproblem*}\end{mdframed}}
\newtheorem{mddefinition}[lemma]{Definition}

\newtheorem*{mddefinition*}{Definition}
\newenvironment{Definition*}{\begin{mdframed}[hidealllines=false,innerleftmargin=10pt,backgroundcolor=white!10,innertopmargin=5pt,innerbottommargin=5pt,roundcorner=10pt]\begin{mddefinition*}}{\end{mddefinition*}\end{mdframed}}
\newtheorem{mdremark}{Remark}

\newenvironment{ourbox}{\begin{mdframed}[hidealllines=false,innerleftmargin=10pt,backgroundcolor=white!10,innertopmargin=2pt,innerbottommargin=5pt,roundcorner=10pt]}{\end{mdframed}}

\newtheorem{mdalgorithm}{Algorithm}

\allowdisplaybreaks

\renewcommand{\qed}{\nobreak \ifvmode \relax \else
      \ifdim\lastskip<1.5em \hskip-\lastskip
      \hskip1.5em plus0em minus0.5em \fi \nobreak
      \vrule height0.75em width0.5em depth0.25em\fi}

\renewcommand{\leq}{\leqslant}
\renewcommand{\geq}{\geqslant}
\renewcommand{\le}{\leq}
\renewcommand{\ge}{\geq}

\newcommand{\Leq}[1]{\ensuremath{\underset{\textnormal{#1}}\leq}}

\newcommand{\Eq}[1]{\ensuremath{\underset{\textnormal{#1}}=}}

\newcommand{\tvd}[2]{\ensuremath{\norm{#1 - #2}_{\mathrm{tvd}}}}
\newcommand{\Ot}{\ensuremath{\widetilde{O}}}
\newcommand{\eps}{\ensuremath{\varepsilon}}
\newcommand{\Paren}[1]{\Big(#1\Big)}
\newcommand{\Bracket}[1]{\Big[#1\Big]}
\newcommand{\bracket}[1]{\left[#1\right]}
\newcommand{\paren}[1]{\ensuremath{\left(#1\right)}\xspace}
\newcommand{\card}[1]{\left\vert{#1}\right\vert}

\newcommand{\IN}{\ensuremath{\mathbb{N}}}

\newcommand{\norm}[1]{\ensuremath{\|#1\|}}

\newcommand{\set}[1]{\ensuremath{\left\{ #1 \right\}}}

\newcommand{\alg}{\ensuremath{\mathcal{A}}\xspace}

\DeclareMathOperator*{\Exp}{\ensuremath{{\mathbb{E}}}}
\DeclareMathOperator*{\Prob}{\ensuremath{\textnormal{Pr}}}
\renewcommand{\Pr}{\Prob}

\newcommand{\Ex}{\Exp}

\newenvironment{tbox}{\begin{tcolorbox}[
		enlarge top by=5pt,
		enlarge bottom by=5pt,
		 breakable,
		 boxsep=0pt,
                  left=4pt,
                  right=4pt,
                  top=10pt,
                  arc=0pt,
                  boxrule=1pt,toprule=1pt,
                  colback=white
                  ]%%
	}
{\end{tcolorbox}}

\newcommand{\event}{\ensuremath{\mathcal{E}}}

\newcommand{\rv}[1]{\ensuremath{{\mathsf{#1}}}\xspace}
\newcommand{\rA}{\rv{A}}
\newcommand{\rB}{\rv{B}}
\newcommand{\rC}{\rv{C}}
\newcommand{\rD}{\rv{D}}

\newcommand{\rX}{\rv{X}}

\newcommand{\rG}{\rv{G}}
\newcommand{\rProt}{\rv{\Prot}}
\newcommand{\rR}{\rv{R}}

\newcommand{\supp}[1]{\ensuremath{\textnormal{\text{supp}}(#1)}}
\newcommand{\distribution}[1]{\ensuremath{\textnormal{dist}(#1)}\xspace}

\newcommand{\kl}[2]{\ensuremath{\mathbb{D}(#1~||~#2)}}
\newcommand{\II}{\ensuremath{\mathbb{I}}}
\newcommand{\HH}{\ensuremath{\mathbb{H}}}
\newcommand{\mi}[2]{\ensuremath{\def\mione{#1}\def\mitwo{#2}\mireal}}
\newcommand{\mireal}[1][]{
  \ifx\relax#1\relax%
    \II(\mione \,; \mitwo)%
  \else%
    \II(\mione \,; \mitwo\mid #1)%
  \fi
}
\newcommand{\en}[1]{\ensuremath{\HH(#1)}}

\newcommand{\itfacts}[1]{\Cref{fact:it-facts}-(\ref{part:#1})\xspace}

\newcommand{\Grs}{\ensuremath{G_{\textnormal{rs}}}}
\newcommand{\Lrs}{\ensuremath{L_{\textnormal{rs}}}}
\newcommand{\Rrs}{\ensuremath{R_{\textnormal{rs}}}}

\newcommand{\Mrs}[1]{\ensuremath{M_{\textnormal{rs}}^{#1}}}

\newcommand{\drs}{\ensuremath{\delta_{\text{rs}}}}
\newcommand{\crs}{\ensuremath{c_{\text{rs}}}}

\newcommand{\prot}{\ensuremath{\pi}}
\newcommand{\Prot}{\ensuremath{\Pi}}

\newcommand{\jstar}{\ensuremath{j^{\star}}}
\newcommand{\istar}{\ensuremath{i^{\star}}}
\newcommand{\raoconst}{\ensuremath{c_{\textnormal{comp}}}}

\newcommand{\reald}{\ensuremath{\GG^{\textnormal{\textsc{real}}}}}

\newcommand{\cc}[1]{\ensuremath{\textsc{CC}(#1)}}

\newcommand{\ic}[2]{\ensuremath{\textsc{IC}(#1,#2)}}

\newcommand{\PP}{\ensuremath{\mathbb{P}}}
\newcommand{\cP}{\ensuremath{\mathcal{P}}}
\newcommand{\GG}{\ensuremath{\mathcal{G}}}
\newcommand{\cH}{\ensuremath{\mathcal{H}}}
\newcommand{\cA}{\ensuremath{\mathcal{A}}}
\newcommand{\cB}{\ensuremath{\mathcal{B}}}
\newcommand{\cHstar}{\ensuremath{\mathcal{H}^*}}

\newcommand{\embed}{\ensuremath{\textnormal{\textsc{embed}}}}

\newcommand{\RS}{Ruzsa–Szemerédi }

\providecommand{\HS}{\textnormal{\textsc{HiddenRestrictions}}}
\providecommand{\cD}{\ensuremath{\mathcal{D}}}
\providecommand{\DHR}{\ensuremath{\cD_{\mathrm{HR}}}}
\providecommand{\Unif}{\ensuremath{\mathcal{U}}}
\providecommand{\pre}{\mathrm{pre}}
\providecommand{\post}{\mathrm{post}}
\renewcommand{\RS}{Ruzsa-Szemer\'{e}di~}

\title{An Unconditional Space-Approximation Lower Bound for Semi-Streaming Matching}

\author{
Yang P. Liu\footnote{(yangl7@andrew.cmu.edu) Carnegie Mellon University. Supported in part by NSF Career Award CCF2542272.\smallskip} 
\and 
Kheeran Naidu\footnote{(kheeran.naidu@gmail.com)\smallskip}
\and
Hoai-An Nguyen\footnote{(hnnguyen@andrew.cmu.edu) Carnegie Mellon University. Supported in part by
an NSF GRFP fellowship grant number DGE2140739/DGE2631988, NSF Career Award CCF2542272, Office
of Naval Research award number N000142112647, and a Simons Investigator Award. 
\smallskip} 
\and 
Janani Sundaresan\footnote{(jsundare@uwaterloo.ca) University of Waterloo. \smallskip} 
}

\date{}

\begin{document}
\pagenumbering{gobble}
\maketitle

% !TeX root = main.tex 
%!TEX root = main.tex

\begin{abstract}

\medskip

We prove the first unconditional lower bound tradeoff between the number of passes and the approximation factor of maximum matchings in the semi-streaming model. Specifically, we show that any algorithm that outputs an $(1-\eps)$-approximation to maximum matching requires $\Omega(\log(1/\eps))$ passes for constant $\eps \in (0,1)$.

Prior to our result, such tradeoffs were proven only conditioned on a certain combinatorial assumption about the density of \RS graphs. The state-of-the-art unconditional lower bound rules out any approximation factor better than $(8/9)$ in two passes but leaves open the possibility of a $(1-\eps)$-approximation for any $\eps > 0$ even in three passes. 

We provide two different proofs of our result: the first proof uses the recent hierarchical-embedding framework of Assadi, Konrad, Naidu, and Sundaresan (STOC 2024) along with combinatorial constructions using \RS graphs. The second extends the two-level \RS hierarchy of Konrad and Naidu (SODA 2024) beyond two levels and proves the corresponding communication lower bound directly by induction, rather than through a reduction to \textsc{Index}.

\end{abstract}

\newpage

\tableofcontents

\newpage

\pagenumbering{arabic}

% !TeX root = main.tex 
%!TEX root = main.tex

\section{Introduction}\label{sec:intro}
In this paper, we study maximum bipartite matching in insertion-only streams.
The input is a bipartite graph \(G\) whose edges arrive in a stream, and an algorithm with limited memory may take several passes over this stream. The relevant memory regime for graph streaming is often \emph{semi-streaming}, introduced by \cite{FKMSZ05}, where the space usage is near-linear in the number of vertices, i.e., $\Ot(n)$. 
We focus on the matching-output version of the problem: the algorithm must return a set of edges forming a large matching in the input graph, rather than estimating the size of a maximum matching. 

Similarly to many previous works, we focus on approximate maximum matchings. In particular, we study the following natural question: for a constant $\eps>0$, how many passes are needed to obtain a semi-streaming
algorithm for $(1-\eps)$-approximate maximum bipartite matching?

\paragraph{Upper bounds.}
There have been numerous works improving the number of passes for semi-streaming algorithms for $(1-\eps)$-approximate maximum matching \cite{M05,EKS09,AG11,AG13,EKMS12,Kapralov13,BDETY17,AG18,T18,GKMS19,ALT21,FMU22,AJJST22,HS23,A25}, where the current best bounds are $O(1/\eps^2)$ \cite{AG13,ALT21}, which is constant for fixed $\eps$, and $O((\log n)/\eps)$ \cite{AG18,A25}, i.e., a better $\eps$ dependence at the cost of a mild dependence on $n$.

\paragraph{Lower bounds.}
Progress on the side of lower bounds is much more limited. While there are several lower bounds against one-pass streaming algorithms (see e.g. \cite{GoelKK12,Kapralov13,Kapralov21,AJX26}), multipass lower bounds are much rarer. In the context of approximate matching, \cite{A22,AS23} proved that semi-streaming $(1-\eps)$-approximate matching requires $\Omega(\log(1/\eps)/\log(1/\beta_n))$ passes, where $\beta_n$ is a parameter which is related to the density of \emph{\RS graphs}; these graphs have edge sets that decompose into many large induced matchings, and are defined formally in \Cref{def:rs-graph}. (In fact, this lower bound even holds for algorithms that approximate the matching \emph{size}.) Currently, it is known that
\[
    \frac{1}{\log \log n} \underset{\text{\cite{FLNRRS02}}}{\le} \beta_n \underset{\text{\cite{FHS17}}}{\le} 1-o(1).
\]
Thus, this bound does not imply anything for constant $\eps$ for the current bound of $\beta_n \ge 1/\log\log n$. However, under the plausible \emph{conditional} hypothesis that $\beta_n=\Omega(1)$ for all $n$, these results imply that semi-streaming $(1-\eps)$-approximate matching size requires $\Omega(\log(1/\eps))$ passes.

Recently, \cite{KonradN24} gave the first unconditional lower bound against two-pass approximate maximum matching. Precisely, they prove that any two-pass algorithm for $(8/9+\eps)$-approximate maximum matching that returns the edge set of the approximate matching requires at least $n^{1+\Omega(1/(\log \log n)^2)}$ space. The proof was based on a two-level \RS graph hierarchy, and analyzing a corresponding communication complexity question. Building on this work, we prove the following unconditional pass lower bound.

\begin{theorem}
\label{thm:main}
\label{thm:main-lb}
For any constant $p\geq 1$ and constant $\delta > 0$, any $p$-pass semi-streaming algorithm cannot output a better than $(1-3^{-p}+\delta)$-approximate maximum bipartite matching. Put differently, for any sufficiently small constant $\eps > 0$, a $(1-\eps)$-approximation requires $\Omega(\log(1/\eps))$ passes.
\end{theorem}

This gives an unconditional result matching the conditional result of \cite{AS23}. Relative to \cite{AS23}, whose proof goes through harder intermediate problems such as hidden permutations and set-hiding, both of our proofs are based on \RS graph hierarchies and the communication/search problems they induce, giving a stronger unconditional lower bound with simpler arguments. We note, however, that \cite{AS23} also holds for size estimation whereas here we study the problem of outputting the edges. 

\paragraph{Related work.}
There are several algorithms and lower bounds for semi-streaming maximum matching in slightly different settings than those studied in this paper. For exact maximum bipartite matching, \cite{GO13} proved multipass lower bounds for testing whether a graph has a perfect matching, while semi-streaming algorithms with sublinear pass complexity in $n$ are known, based on optimization and parallel reachability algorithms \cite{AJJST22}. In dynamic streams, where the stream consists of edge insertions and deletions, recent work has established that $\Theta(\log \log n)$ passes is necessary and sufficient for achieving a constant-factor approximation in semi-streaming space \cite{ABKNS25}; identifying the exact dependence on the accuracy parameter $\eps$ remains open. We refer the reader to the references within \cite{ABKNS25} for further discussion related to matching in dynamic streams.

Finally, beyond matching, there has been recent progress on techniques for proving lower bounds against multi-pass semi-streaming algorithms for a number of graph problems; see \cite{CKPSRSY21,AKNS24,Assadi23} and the references therein. 

\paragraph{Road map.}
In Section~\ref{sec:prelims} we give the preliminaries with additional preliminaries in Appendix~\ref{app:additional-prelims}. We give a technical overview of previous work and both proofs of our result in Section~\ref{sec:technical-overview}. The first proof appears in Section~\ref{sec:he-proof} and Section~\ref{app:roundelim} and the second proof appears in Section~\ref{sec:communication} and Section~\ref{sec:reduction}.

% !TeX root = main.tex 
%!TEX root = main.tex

\section{Preliminaries}\label{sec:prelims}

\paragraph{Notation.}
All logarithms are base \(2\) unless otherwise stated. For a positive integer
\(m\), write \([m]:=\{1,\ldots,m\}\). For a tuple
\((x_1,\ldots,x_t)\) and \(i\in[t]\), we write
\(x_{<i}:=(x_1,\ldots,x_{i-1})\). For a set of tuples
\(\{(x,y):x\in X,y\in Y\}\), and any \(x\in X\), we write
\((x,*)\) for \(\{(x,y):y\in Y\}\).

For any graph \(G=(V,E)\) and any \(U\subseteq V\), we use \(G[U]\) to denote
the subgraph of \(G\) induced by \(U\). For a graph \(G\), let \(\mu(G)\)
denote the size of a maximum matching in \(G\).

When distinguishing a random variable from a realization, we often write the
random variable in sans serif, e.g. \(\rX\), and its realization as \(X\).
For a random variable \(\rX\), \(\distribution{\rX}\) denotes its law and
\(\supp{\rX}\) its support. For a finite set \(S\), let \(\Unif(S)\) denote
the uniform distribution on \(S\). For \(x\in[0,1]\), let
\[
    h(x):=-x\log x-(1-x)\log(1-x)
\]
be the binary entropy function, with the convention \(0\log 0=0\). For two
strings \(x,y\) indexed by the same finite set, let \(d_{\rm H}(x,y)\) denote
their Hamming distance. Subscripts on asymptotic notation indicate which
parameters the hidden constants may depend on, and a tilde suppresses
logarithmic factors in the relevant size parameters.

\paragraph{\texorpdfstring{\RS Graphs and Embeddings}{RS Graphs and Embeddings}.}

We rely on an extremal family of graphs known as \RS graphs~\cite{RuzsaS78}.
We use the construction originally given by~\cite{FLNRRS02} and
subsequently improved by~\cite{GoelKK12}.

\begin{definition}[$(r,t)$-RS graph]\label{def:rs-graph}
A bipartite graph \(Q=(L\cup R,E)\) is an \((r,t)\)-RS graph if its edge set
can be partitioned into \(t\) pairwise edge-disjoint induced matchings
\[
    M_1,\ldots,M_t,
    \qquad
    M_i=\{e_{i,1},\ldots,e_{i,r}\},
\]
each of size \(r\). Here ``induced'' is with respect to the whole graph \(Q\):
the subgraph of \(Q\) induced by the endpoints of \(M_i\) contains exactly the
edges of \(M_i\).
\end{definition}

\begin{proposition}[cf.~\cite{GoelKK12}]\label{prop:RS-gkk12}
For every sufficiently small constant \(\drs>0\), there is a constant
\(\crs=\crs(\drs)>0\) such that, for every sufficiently large integer \(n\),
there is a bipartite \((r,t)\)-RS graph \(Q=(L\cup R,E)\) satisfying
\(|L|=|R|=n\) and
\[
    r=\left\lfloor\left(\frac12-\drs\right)n\right\rfloor,
    \qquad
    t\geq n^{\crs/\log\log n}.
\]
\end{proposition}

In the first proof, the
displayed values of \(r\) and \(t\) omit floors for readability; if the
construction supplies more than the required number of induced matchings, we
discard the extras.  Section~\ref{sec:reduction} writes the rounding explicitly.

\begin{definition}[Embedding product]\label{def:embed}
Let \(Q=(L_{\rm out}\cup R_{\rm out},E_{\rm out})\) be an \((r,t)\)-RS graph
with induced matchings
\[
    M_i=\{e_{i,1},\ldots,e_{i,r}\},
    \qquad
    e_{i,j}=(u_{i,j},v_{i,j}),
\]
where \(u_{i,j}\in L_{\rm out}\) and \(v_{i,j}\in R_{\rm out}\). Let
\(\mathcal H=\{H_{i,j}:i\in[t],j\in[r]\}\) be a family of bipartite graphs on
a common vertex set \(L_{\rm in}\cup R_{\rm in}\). The embedding of
\(\mathcal H\) into \(Q\), denoted by \(\embed(\mathcal H\to Q)\), has left
side \(L_{\rm out}\times L_{\rm in}\) and right side
\(R_{\rm out}\times R_{\rm in}\). For every \(i\in[t]\), \(j\in[r]\), and edge
\((x,y)\in E(H_{i,j})\), add the edge
\[
    ((u_{i,j},x),(v_{i,j},y)).
\]
We say that this edge is added with respect to the outer edge \(e_{i,j}\).
\end{definition}

% !TeX root = main.tex 
%!TEX root = main.tex
\section{Technical Overview}\label{sec:technical-overview}

\newcommand{\MBM}[3]{\ensuremath{\textnormal{\ensuremath{#3}-\textbf{MBM}}_{#1,#2}}}

\newcommand{\player}[1]{\ensuremath{{P}_{#1}}}

\newcommand{\tG}{\ensuremath{\widetilde{G}}}

\newcommand{\tE}{\ensuremath{\widetilde{E}}}

\newcommand{\Gstar}{\ensuremath{G^{\star}}}
\newcommand{\Gstarb}{\ensuremath{\overline{G^{\star}}}}

\newcommand{\Mstar}{\ensuremath{M^{\star}}}

\newcommand{\cMstar}{\ensuremath{\mathcal{M}^{\star}}}

\newcommand{\Mstarb}{\ensuremath{\overline{M^{\star}}}}

\newcommand{\dmark}{\ensuremath{\delta_\text{markov}}}

\newcommand{\ristar}{\ensuremath{\rv{i}^\star}}

\newcommand{\rjstar}{\ensuremath{\rv{j}^\star}}

\newcommand{\rH}{\ensuremath{\rv{H}}}

\newcommand{\sigmae}{\ensuremath{\overline{\sigma}}}

\newcommand{\Sigmae}{\ensuremath{\overline{\Sigma}}}

\newcommand{\rSigma}{\ensuremath{\rv{\Sigma}}}

\newcommand{\rPi}{\ensuremath{\rv{\Pi}}}

\newcommand{\dist}{\ensuremath{\textnormal{dist}}}

\newcommand{\Gbar}{\ensuremath{\overline{G}}}

\newcommand{\Lin}{\ensuremath{L_{\textnormal{in}}}}
\newcommand{\Rin}{\ensuremath{R_{\textnormal{in}}}}

\newcommand{\Lem}{\ensuremath{L_{\textnormal{em}}}}
\newcommand{\Rem}{\ensuremath{R_{\textnormal{em}}}}

\newcommand{\Laux}{\ensuremath{L_{\textnormal{aux}}}}
\newcommand{\Raux}{\ensuremath{R_{\textnormal{aux}}}}

\newcommand{\Maux}{\ensuremath{M^{\textnormal{aux}}}}

\newcommand{\Min}{\ensuremath{M^{\textnormal{in}}}}

\newcommand{\Vin}{\ensuremath{V^{\textnormal{in}}}}

\newcommand{\kappars}{\ensuremath{\kappa_{\textnormal{rs}}}}

\newcommand{\Vstar}{\ensuremath{V^{\star}}}
% !TeX root = main.tex 
%!TEX root = main.tex

\subsection{Prior Work}

Our work builds on the prior lower bounds for approximate matchings in \cite{GoelKK12} and \cite{KonradN24}, and our first proof also builds on the recent maximal independent set lower bound of \cite{AKNS24}. We begin with a short description of the prior work.

\paragraph{The one-pass lower bound of \cite{GoelKK12}.}

\cite{GoelKK12} pioneered the use of RS-graphs in proving graph streaming lower bounds. They proved that any streaming algorithm that obtains a $(2/3 + \eps)$-approximation requires $n^{1+\Omega(1/\log \log n)}$ space. 

Formally, RS-graphs are defined in \Cref{def:rs-graph}. Informally, their edges are a union of $t$ disjoint induced matchings, each of size $r$.\footnote{An induced matching is a matching where there are no other edges of the graph between the endpoints of the matching.} They are globally dense, but the vertex subsets corresponding to the induced matchings are sparse. See \Cref{fig:rs} for an illustration.

\begin{figure}[h!]
	\centering
	\includegraphics[scale=0.07,angle=90]{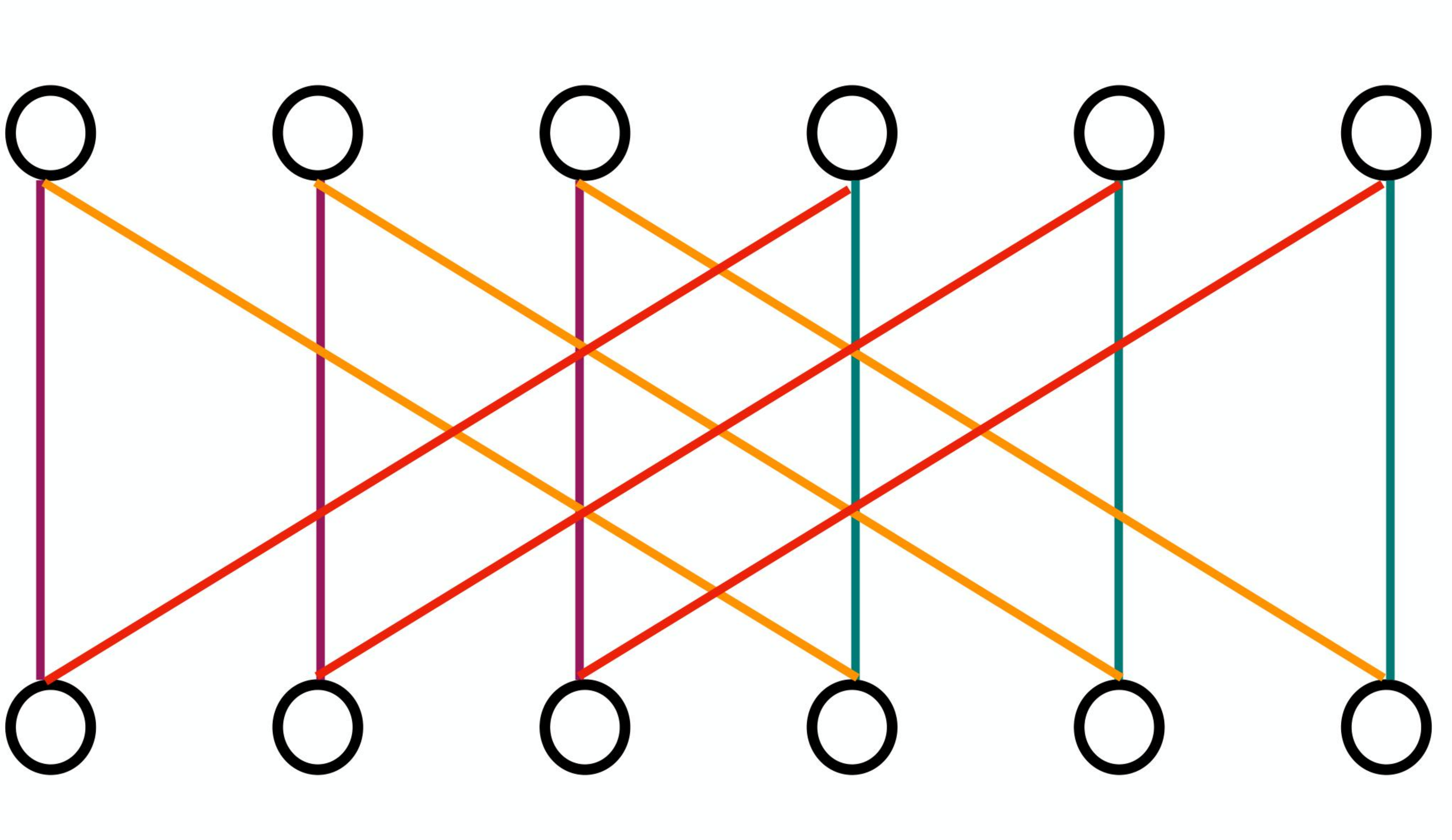}
	\caption{ An RS-graph with parameters $t = 4$, and $r = 3$. There are four induced matchings, each of size three. Each induced matching is shown in a different colour. }\label{fig:rs}
\end{figure}

\cite{GoelKK12} constructs an instance $G = (V, E)$ as follows with $n$ vertices. The base is an $(r, t)$-RS graph, and each edge of the RS-graph is removed independently with probability $1/2$ with $r \cdot t = n^{1+1/\log \log n}$. One induced matching is picked to be special with auxiliary edges and vertices. These edges are revealed only after all the edges of the RS-graph are sent, so the identity of the special induced matching is hidden. The instance is constructed so that the following statements are true:
\begin{enumerate}[label=$(\roman*)$]
	\item The graph $G$ has a large matching of size $\mu(G) = \Omega(n)$.
	\item Let $\Vstar$ be the vertex set corresponding to the special induced matching. The set $V \setminus \Vstar$ forms a small vertex cover of size at most $2\mu(G)/3$.
	\item Any matching of size $(2/3 + \eps) \cdot \mu(G)$ must find at least $\Omega(\eps)$-fraction of the edges of the induced matching.
\end{enumerate}
Refer to \Cref{fig:gkk12} for an illustration. Without knowing the identity of the induced matching among the $t$-many choices, an algorithm is forced to remember almost \emph{all} the edges of the RS-graph.\footnote{Remembering the edge corresponds to whether the edge was removed or not.} This requires $\Omega(rt) = \Omega(n^{1+1/\log \log n})$ space,\footnote{The specific parameters are given in \Cref{prop:RS-gkk12}.} corresponding to the dense RS-graphs used in \cite{GoelKK12}. This is more than semi-streaming space, hence a better approximation than $2/3$ is not possible with one pass.\footnote{Stronger lower bounds are now known for one pass algorithms \cite{Kapralov13,Kapralov21,AJX26}, and we know that any algorithm in one pass in semi-streaming cannot do better than a $0.55$ approximation. The best algorithm however still finds a greedy matching with a $0.5$-approximation. }

\begin{figure}[ht!]
	\centering
	\begin{subfigure}[b]{0.45\textwidth}
		\centering
		\includegraphics[scale=0.25]{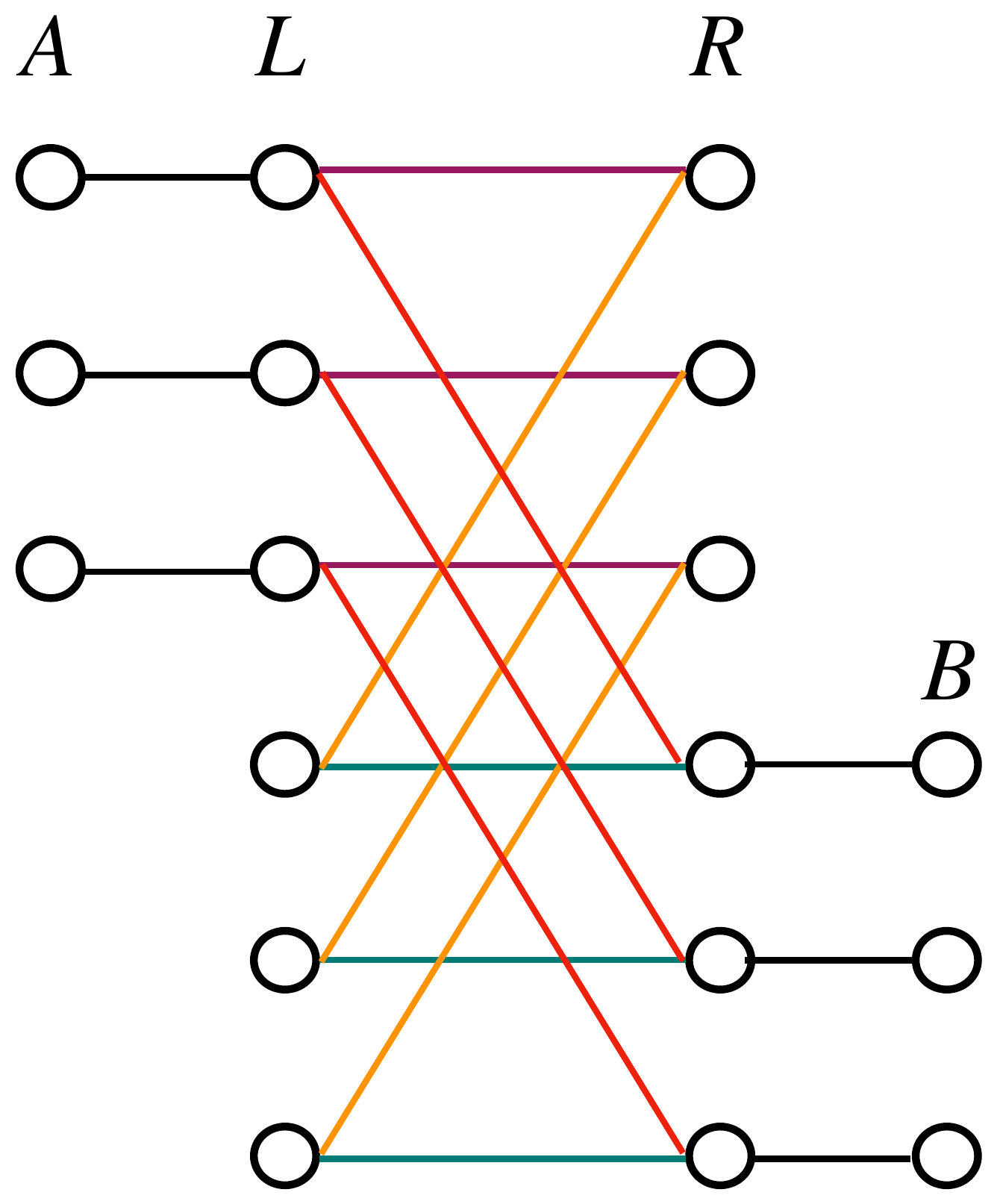}
		\caption{Given an RS graph on vertices $L \cup R$ with $4$ matchings of size $3$ each, auxiliary vertices $A \cup B$ are added, which connect to all vertices except the end points of the special orange matching. }\label{fig:RSsubgkk1}
	\end{subfigure}
	\hfill
	\begin{subfigure}[b]{0.45\textwidth}
		\centering
		\includegraphics[scale=0.25]{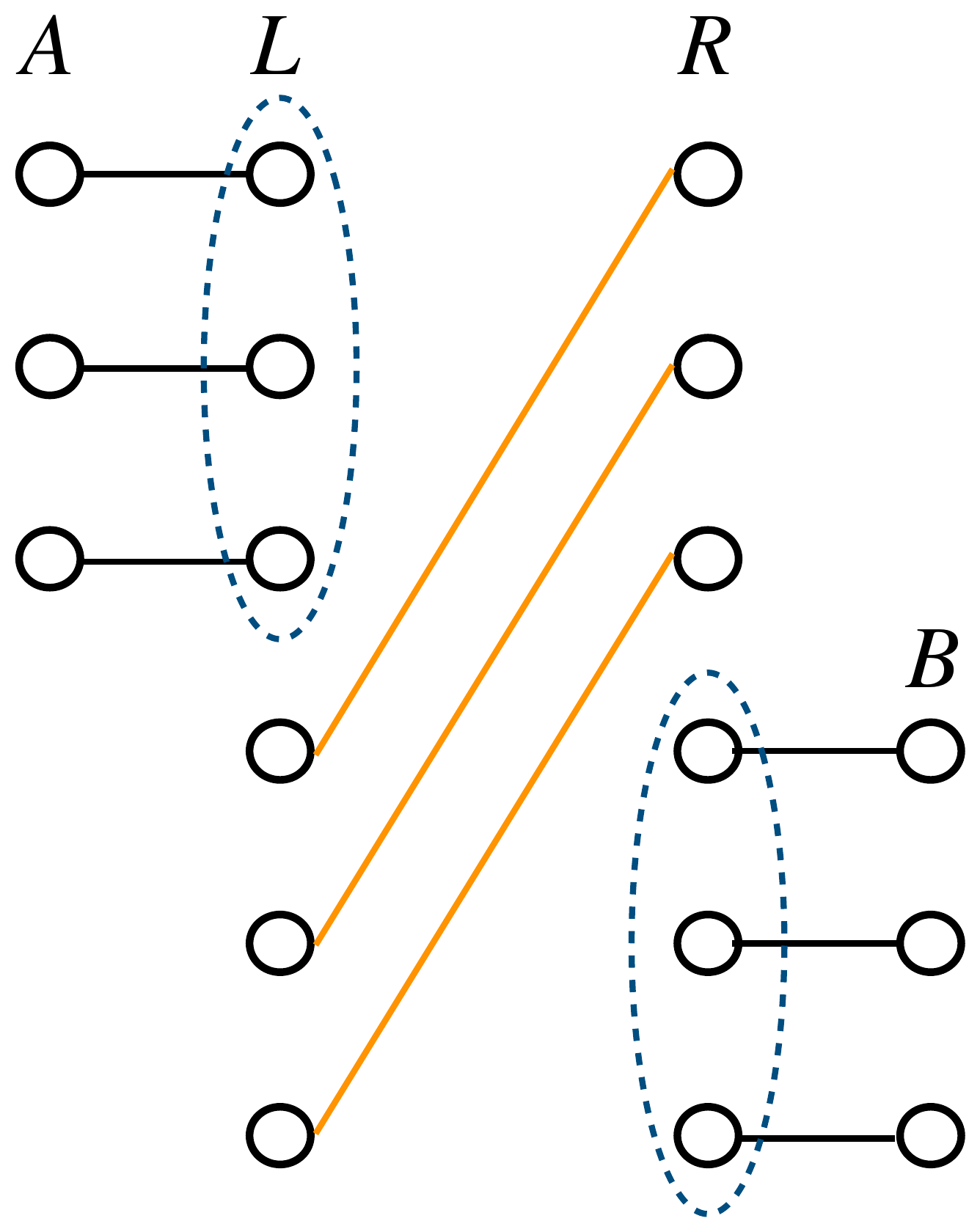}
		\caption{The orange edges of the RS graph are present with probability $(1/2)$ independently. The vertices in the blue circle (the set $V \setminus \Vstar$) form a small vertex cover of the other edges in the graph.}
	\end{subfigure}
	
	\caption{One round hard instance for $(2/3)$ approximate matching in \cite{GoelKK12}. Any $(2/3+\eps)$-approximate matching must find at least an $\Omega(\eps)$-fraction of the yellow edges.}
	\label{fig:gkk12}
\end{figure}

\paragraph{The two-pass lower bound of \cite{KonradN24}.}

To extend \cite{GoelKK12} into a two-pass lower bound, each edge of the RS-graph is replaced with a one-pass hard instance. More concretely, there are two levels of RS-graphs: an outer one with $n_1$ vertices and an inner one with $n_2$ vertices. The special matching of the outer graph is revealed at the end of the first pass. Each edge of this graph is \emph{another independent hard instance} for one pass algorithms. 

This gives a hard instance for two-pass algorithms: in the first pass, the special matching of the outer graph will be unknown, and in the second pass, all the special matchings of the smaller independent instances will be unknown. See \Cref{fig:2roundGKK} for an illustration.

\begin{figure}[ht]
	\centering
	\includegraphics[scale=0.07]{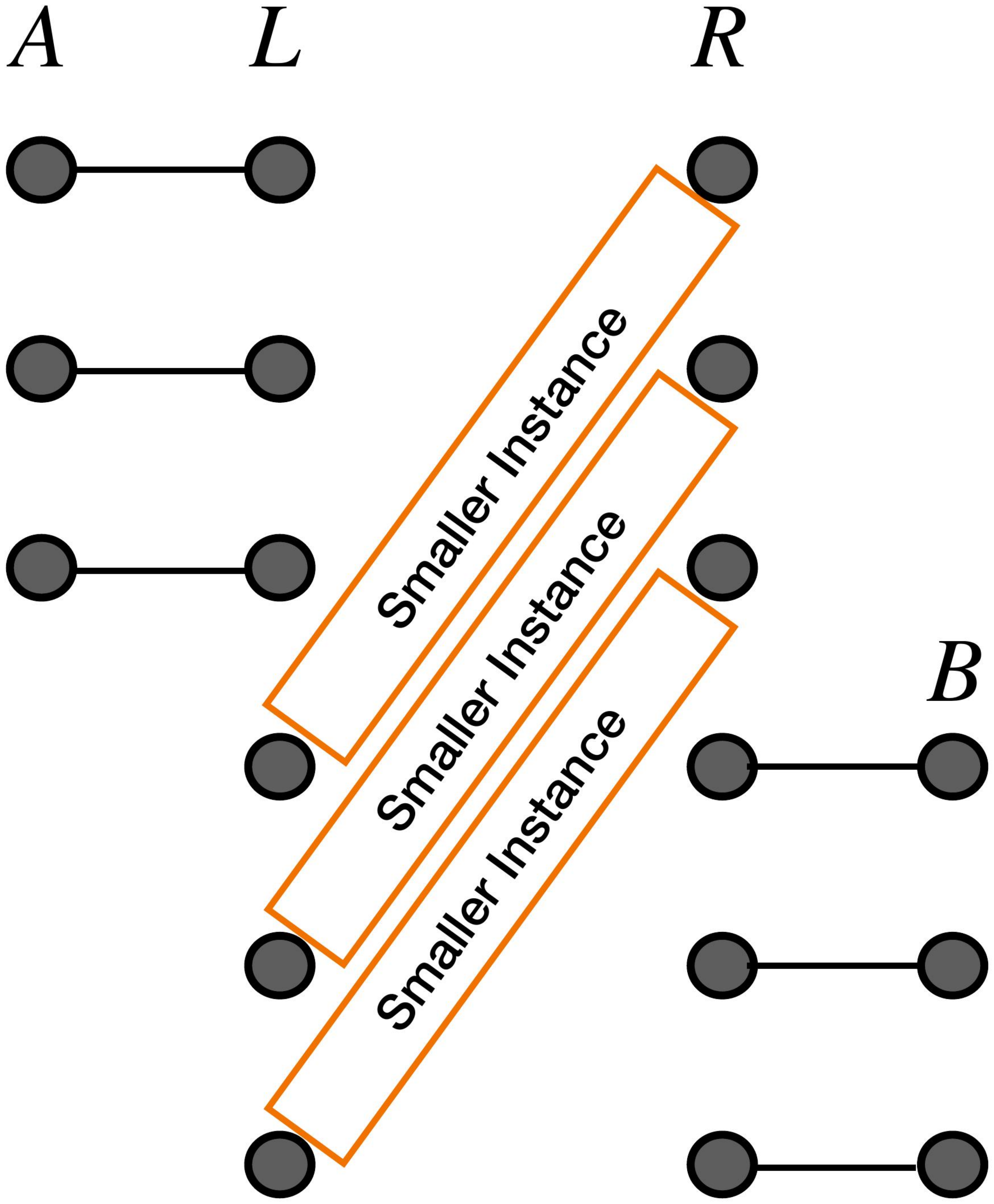}
	\caption{A hard instance constructed by \cite{KonradN24}, any $(8/9+\eps)$-approximate matching finds at least an $\Omega(\eps)$-fraction of the edges of the special matchings in the smaller instances. The smaller instances look exactly as in \Cref{fig:RSsubgkk1}. }\label{fig:2roundGKK}
\end{figure}

\cite{KonradN24} implements this approach to prove that any algorithm with an approximation ratio better than $(1-(1/3)^2)=(8/9)$ requires more than semi-streaming space (see \Cref{fig:2roundGKK}). In their instance:
\begin{enumerate}
	\item The auxiliary edges of the outer instance cover a $(2/3)$ fraction of the vertices, and the auxiliary edges of the smaller instances in the special matching cover another $(2/3)$ fraction of the vertices. In total, these auxiliary edges cover:
	\[
	\frac23 + \frac13 \cdot \frac23 = \frac89
	\]
	fraction of the vertices. Thus, the lower bound holds when an approximation ratio larger than $(8/9)$ is found. 
	\item The outer graph is chosen with $n_1 \approx n_2^{2\log \log n_2}$ vertices. The total number of vertices will be $n_2 \cdot n_1$.
	Any streaming algorithm has two choices:
	\begin{enumerate}
		\item Use $\Omega(n_1^{1+1/\log \log n_1})$ space to gain information about the $\approx n_1$ many smaller special instances in the first pass;
		\item Or spend $\Omega(n_1 \cdot n_2^{1+1/\log \log n_2})$ space to solve a majority of the smaller instances in the second pass without help from the first pass.
	\end{enumerate}
	\cite{KonradN24} proves that either one of these cases must happen, leading to a lower bound of
	\[
	n_1 \cdot \min\{{n_1^{1/\log \log n_1},  n_2^{1+1/\log \log n_2}}\} = (n_1 \cdot n_2)^{1+\Omega(1/\log \log (n_1 \cdot n_2))^2},
	\]
	which is not possible in semi-streaming.
\end{enumerate}

The techniques of \cite{KonradN24} however, do not extend beyond two passes, since in their information theoretic component they rely on the two-way communication complexity tradeoffs of the Index problem \cite{JainRS03}. We extend this framework further, and show that any $p$-pass algorithm cannot get an approximation ratio better than $(1-1/3^p)$ in semi-streaming space, which in turn proves that any $(1-\eps)$-approximation algorithm requires $\Omega(\log (1/\eps))$-passes.  We give two distinct proofs of our result.

\subsection{Our Approach} 

It is fairly straightforward to construct hard instances for multiple passes by extending the construction of \cite{KonradN24}. To construct a hard instance for $p$-passes we simply take an outer RS-graph with $n_p$ vertices, and replace the edges with hard instances for $(p-1)$-passes. 
Similar to the proof of \cite{KonradN24}, finding any approximation larger than $(1-\Omega(1/3^{p})+\eps)$-factor will require an algorithm to find many special edges in the last layer of these RS-graphs, and one can expect this to require $n^{1+1/(\log \log n)^{\Omega(p)}}$ space, which is larger than semi-streaming space. 

While extending the construction of \cite{KonradN24} may be simple, it is unclear how to prove the lower bound via information-theoretic arguments. We give two distinct approaches to do this. 

\paragraph{First Approach: Hierarchical Embedding of \cite{AKNS24}.}

This proof relies on the information-theoretic framework of \cite{AKNS24}, which proves an $\Omega(\log \log n)$ lower bound in the number of passes for finding a maximal independent set. While \cite{AKNS24} also involves combinatorial constructions pertaining to embedding non-bipartite graphs, the information-theoretic round elimination arguments are more relevant for our purpose. 

A hard instance for $p$-passes is constructed as follows. 
There are $t$-groups of smaller $(p-1)$-pass instances, and each group contains $r$ of these instances (in total, there are $r \cdot t$ many instances). To solve the $p$-pass hard instance, \emph{all} the smaller instances in a specially chosen group have to be solved. The identity of the specially chosen group is only revealed after all the $t$-groups are seen in the first pass.

An algorithm that solves the $p$-pass hard instance has two choices:
\begin{enumerate}[label=$(\roman*)$]
	\item Gain some advantage to solve all the $r\cdot t$ many $(p-1)$-pass instances in the first pass using $O(r \cdot t)$ space. 
	\item After the special group is known, spend $r \cdot s_{p-1}$ space to solve all the instances in $p-1$ passes, where $s_{p-1}$ is the space required to solve a smaller hard instance if only $p-1$ passes are given. 
\end{enumerate}

\cite{AKNS24} proves that this algorithm is optimal by using direct sum arguments \cite{BarakBCR10} and message compression arguments \cite{BravermanG14,RaoY20}. We apply this round elimination argument with these changes:

\begin{enumerate}[label=$(\roman*)$]
	\item Each group of smaller instances corresponds to an induced matching, and picking a group corresponds to identifying a special induced matching. 
	\item In place of solving \emph{all} the smaller instances in the special group, we show the statement holds even if only a large constant fraction of them have to be solved. This corresponds to a large matching in our graph only identifying many of the special edges and not all of them.
\end{enumerate}

\paragraph{Second Approach: Direct Proof via Induction.}

The second proof uses the same hierarchical graph construction, but proves the
communication lower bound from scratch rather than reducing to
\(\mathsf{Index}\) as \cite{KonradN24} does.  The communication problem has \(p+1\) players
\(P_1,\ldots,P_{p+1}\), \(p\) passes, and parameters
\((r_1,t_1),\ldots,(r_p,t_p)\).  Think of \(P_1\) as holding a large table of
bits, while \(P_2,\ldots,P_{p+1}\) hold instructions for which entries of this
table are relevant.  The final player must output one bit for every tuple
\((y_1,\ldots,y_p)\in[r_1]\times\cdots\times[r_p]\); here \(r_k\) is the
number of possibilities for the \(k\)-th output coordinate.  To locate the bit
for a fixed tuple, the players' inputs also define hidden indices
\(i_1,\ldots,i_p\), where \(i_k\in[t_k]\).  Namely, \(P_{p+1}\)'s input gives
\(i_1\), \(P_p\)'s input gives \(i_2\) as a function of \((i_1,y_1)\), and so
on, until \(P_2\)'s input gives \(i_p\).  The desired bit is then the entry of
\(P_1\)'s table indexed by \((i_1,y_1,\ldots,i_p,y_p)\).

There are simple protocols for this problem.  If the players had one more pass,
they could reveal the hidden indices one level at a time and then ask \(P_1\)
for the corresponding table entries.  With only \(p\) passes, one level has to
be handled before the relevant hidden index is known.  The natural workaround
is to choose a level and pay for all hidden possibilities at that level.  For
example, \(P_1\) can send all bottom bits for every possible final hidden index,
costing about \(r_1\cdots r_p t_p\) bits.  Or \(P_2\) can first send its whole
table of final hidden indices; then \(P_1\) only needs to send the bottom bits
selected by that table, so the leading cost moves to about
\(r_1\cdots r_{p-1}t_{p-1}\).  In general, choosing level \(k\) costs roughly
$r_1\cdots r_k t_k$, 
and the direct lower bound says that the best of these natural protocols is
optimal up to logarithmic factors.

The proof formalizes this intuition by induction on the levels.  At level
\(k\), after fixing the hidden indices from earlier levels, one player holds a
large table \(F_k\), and another player holds a selector \(G_k\) that chooses
one column of \(F_k\) in each relevant row.  The transcript can learn the
selected table \(F_{k,G_k}\) in two ways: by learning the whole table \(F_k\),
which is diluted by the \(t_k\) possible columns, or by learning the selector
\(G_k\), whose information is controlled by the previous level.  Iterating this
recurrence bounds the information in the final selected table, which is exactly
the output vector.  Since constant Hamming advantage requires
\(\Omega(r_1\cdots r_p)\) information about this vector, the communication must
be
\[
    \widetilde{\Omega}\!\left(\min_{k\in[p]} r_1\cdots r_k t_k\right).
\]

% !TeX root = main.tex 
%!TEX root = main.tex
\section{First Proof: Hierarchical Embeddings}\label{sec:he-proof}

Our result is proved by a $p$-round communication lower bound for approximate maximum matching (see \Cref{prop:cc-stream}). We state our result formally in terms of communication complexity first.

\begin{theorem} \label{thm:cc-multi-MBM}
	For any $p \in \IN$, $0 < \eps \leq 1$, and sufficiently large integer $n \geq 1$, any $p$-round $(p+1)$-player protocol $\prot$ that outputs an $(1- \eps)$-approximate maximum matching with
\[
\eps < \paren{\frac{1-2\drs}{3+2\drs}}^p,
\] 
where $\drs$ is the constant from \Cref{prop:RS-gkk12}, and has at least $2/3$ probability of success has communication cost
$$\cc{\prot} = n^{1+\Omega\paren{1/(\log\log n)^{3p}}}.$$
\end{theorem}

Proof of \Cref{thm:main-lb} follows using \Cref{prop:cc-stream} and \Cref{thm:cc-multi-MBM}.

\begin{proof}[First Proof of \Cref{thm:main-lb}.]
		Let $\alg$ be a semi-streaming algorithm that gets a $(1-\eps)$-approximate maximum matching in $p$-passes and $\Ot(n)$ space with probability of success at least $2/3$. By the connection between streaming algorithms and communication protocols \Cref{prop:cc-stream}, $\alg$ gives rise to a protocol for approximate matching with the same guarantees. 
	
	Using \Cref{thm:cc-multi-MBM}, we know that the space required by such a protocol is at least
	\[
	n^{1+\Omega\paren{1/(\log\log n)^{3p}}} = \omega(n \log^c n),
	\]
	for any constant $c$ when $p = o(\log \log n/\log \log \log n)$, which is more than the space used by $\alg$. Therefore, the approximation ratio of such an algorithm should be at most:
	\[
	1-\eps >1-  \paren{\frac{1-2\drs}{3+2\drs}}^p,
	\]
	where $\drs$ is the choice of constant from \Cref{prop:RS-gkk12}. For sufficiently small $\drs$, this approximation ratio  is larger than $(1-1/3^p + \delta)$, proving the statement.
\end{proof}

The rest of this section is dedicated to the proof of \Cref{thm:cc-multi-MBM}.
The first part of this proof is the construction of a hard distribution for matchings. 
 The distributions are defined recursively as a family $\set{\GG_p}_{p \geq 0}$ where each $\GG_p$ is a hard distribution for $p$-round $(p+1)$-player protocols.

Before describing the hard distributions, we recall that the embedding product
from \Cref{def:embed} replaces every edge of an ``outer'' RS graph with one of
the ``inner'' graphs of a family. See \Cref{fig:embedRSmatch} for an
illustration. For ease of notation, we use $\Lem$ to denote the vertex set $L_{\rm out}\times L_{\rm in}$ and $\Rem$ to denote the set $R_{\rm out}\times R_{\rm in}$, corresponding to the vertices of the embedded graph. For matchings in RS-graphs, we use $\Mrs{i}$ to denote the $i^{\textnormal{th}}$ induced matching from the $t$ induced matchings.

\begin{figure}
	\centering
	\includegraphics[scale=0.1]{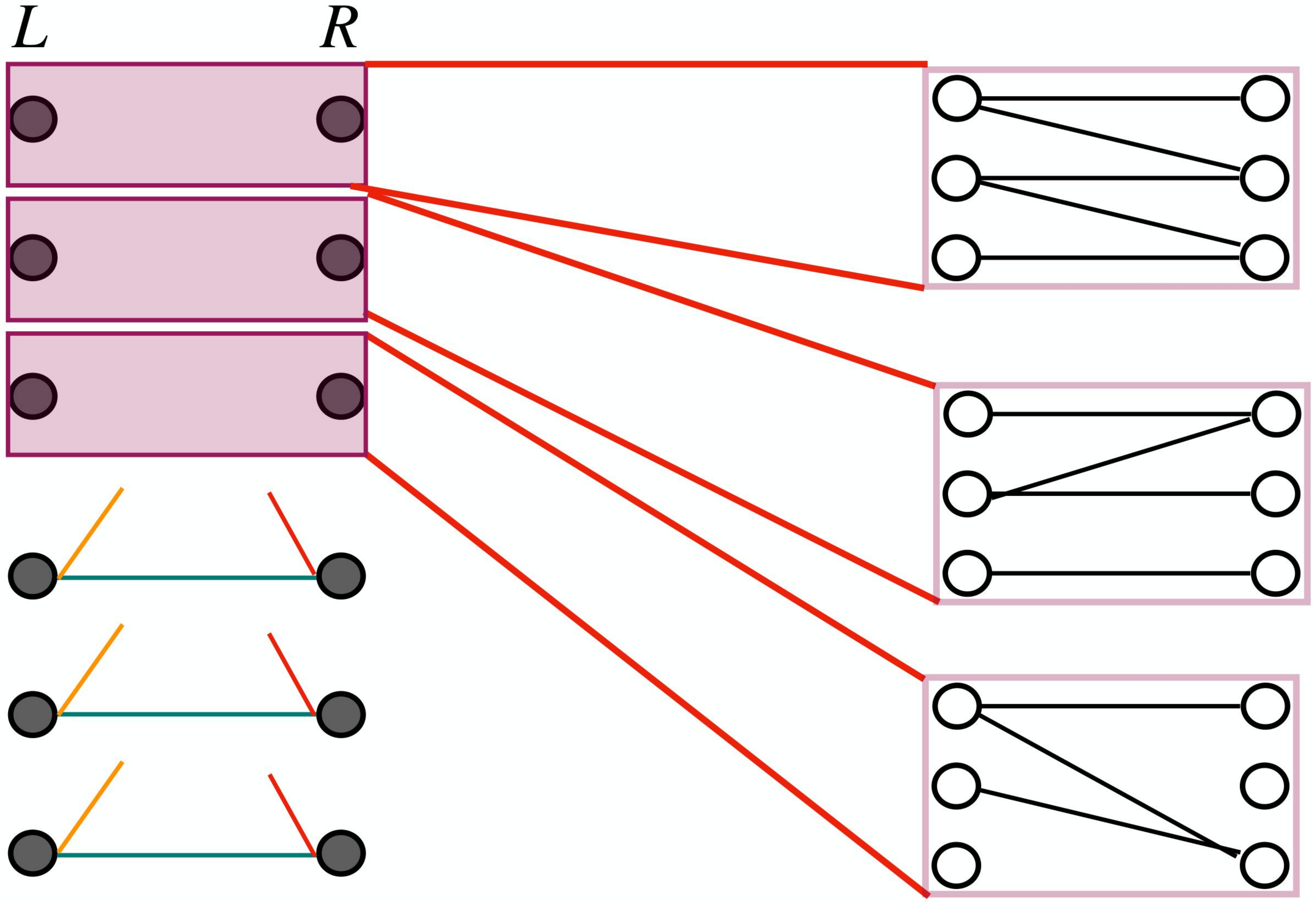}
	\caption{An example of the embedding product on RS graphs. An RS graph with parameters $r = 3, t = 4$ is used and there are totally 12 many smaller graphs. We show three such graphs on the right. The shaded circles indicate vertices which are groups of three. Each edge is replaced by a specific bipartite graph from the given family.}\label{fig:embedRSmatch}
\end{figure}

The following lemma captures the main property of these embeddings, which crucially relies on the matchings in a RS graph being induced and edge disjoint. 

\begin{lemma}\label{lem:match-embed-induced}
	Let $\Grs$ be a balanced bipartite $(r,t)$-RS graph for some $r,t \geq 1$, let $$\cH := \{H_{i,j} \mid i \in [t] , j \in [r]\}$$ be a family of balanced bipartite graphs and let $G := \embed(\cH \rightarrow \Grs)$. 
	For any $i \in [t]$, the \textbf{\emph{induced subgraph}} of $G$ on vertices corresponding to $\Mrs{i}$, i.e., $\{(v,*) \mid v \in e~\text{for some $e \in \Mrs{i}$}\}$, is a \emph{\textbf{vertex-disjoint union}} of graphs $H_{i,j}$ for $j \in [r]$ in $\cH$. 
\end{lemma}

\begin{proof}
	For any $i \in [t]$, let $V_i := \{(v,*) \mid v \in e_{i,j}~\text{for some $e_{i,j} \in \Mrs{i}$}\}$ denote the vertices of $G$ corresponding to the induced matching $\Mrs{i}$. 
	Let $G_i := G[V_i]$ be the induced subgraph of $G$ on $V_i$. We prove the result in two parts: (i) $G_i$ contains the disjoint union of $H_{i,j}$ for $j \in [r]$, and (ii) it does not contain any other edges.
	
	\paragraph{Part (i).} 
	Fix $j \in [r]$ and observe that all edges of $H_{i,j}$ are added to $G_i$ by the embedding product with respect to the edge $e_{i,j} = (u, v)$ in a way that ``preserves its structure''. 
	Formally, the bipartite edges of $H_{i,j}$ between vertices $\Lin$ and $\Rin$ are added, with a one-to-one correspondence, between the vertices $\{ u\} \times \Lin \subseteq \Lem$ and $\{ v \} \times \Rin \subseteq \Rem$ of $G_i$.
	Furthermore, since the edges of the matching $\Mrs{i}$ are vertex-disjoint, the edges added to $G$ with respect to the edges of $\Mrs{i}$ form a vertex-disjoint union of $H_{i,j}$ for all $j \in [r]$ in $G_i$.
	
	\paragraph{Part (ii).}
	It remains to argue that no other edges are added to $G_i$ by the embedding product. 
	Towards a contradiction, assume that $G_i$ contains an edge $f$ that is \emph{not} added with respect to an edge of $\Mrs{i}$. 
	This implies that $f$ is added with respect to an edge $e_{i',j'} = (u',v')$ of $\Grs$ for some $i' \neq i \in [t]$ and $j' \in [r]$, and that $u',v'$ must each be an endpoint of edges in $\Mrs{i}$.
	This, however, is a contradiction since $e_{i',j'}$ is an edge of $\Mrs{i'}$ that is edge disjoint from $\Mrs{i}$, and $\Mrs{i}$ is an induced matching.
\end{proof}

As previously mentioned, the distributions are defined recursively as a family $\set{\GG_p}_{p \geq 0}$ where each $\GG_p$ is a hard distribution for $p$-round $(p+1)$-player protocols. 
We use $\player{1},\player{2}$,$\ldots,\player{p+1}$ to denote the players in $\GG_p$. 
In addition, for any input graph $G$ and any $a \in [p+1]$, the input to player $\player{a}$ is denoted by $G(a)$. 

\paragraph{Base case for $p=0$.} The base case is defined for $p=0$ such that protocols with $0$ rounds of communication are required to solve a simple yet non-trivial matching problem (and is hence trivially impossible). 
Formally, the distribution is as follows and is illustrated in \Cref{fig:harddistG0}:

\begin{tbox}\textbf{The distribution $\GG_0$ over $n$-vertex graphs for $0$-round ``protocols''.}
	\begin{enumerate}
		\item Let $G$ be a balanced bipartite graph on $2n$ vertices. $G$ is made of $r_0 = (n/2)$ vertex disjoint gadgets, each defined as follows:
		\begin{enumerate}
			\item Let $H = (L ,R,E)$ be a gadget on vertices $L= \set{u_1, u_2}$ and $R= \set{v_1, v_{2}}$
			\item With probability $1/2$, the edges $(u_1, v_1)$ and $(u_2, v_2)$ are included in $E$ (a parallel gadget). Otherwise, the edges $(u_1, v_2)$ and $(u_2, v_1)$ are included in $E$ (a cross gadget).
		\end{enumerate}
		\item Player $\player{1}$ receives $G(1) := G$ as the input (this is the only player when $p=0$). 
	\end{enumerate}
\end{tbox}

\begin{figure}[h!]
	\centering
	\begin{subfigure}[b]{0.45\textwidth}
		\centering
		\includegraphics[scale=0.05]{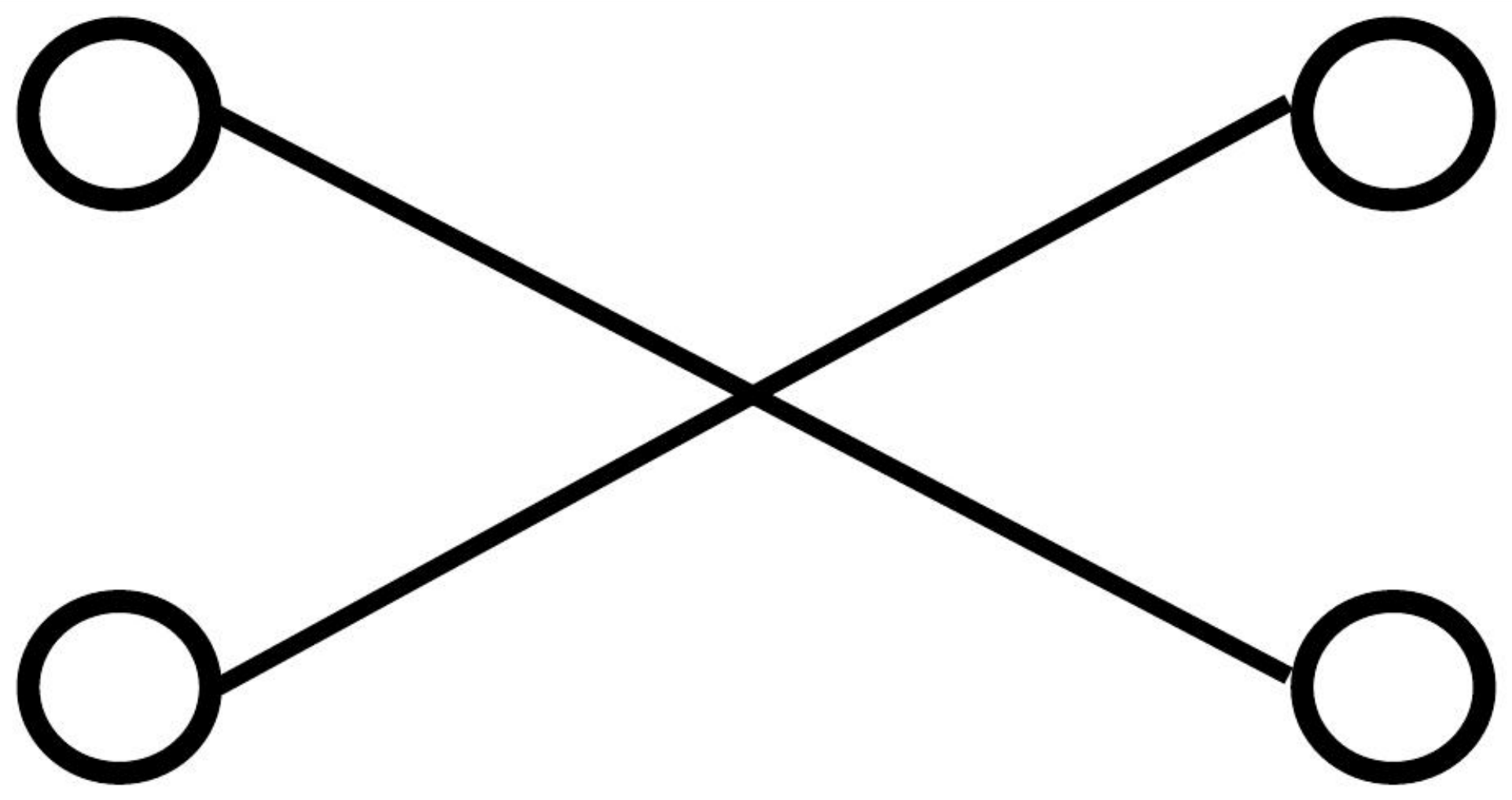}
		\label{fig:cross-gadget}
		\caption{A cross gadget.}
	\end{subfigure} 
	\hspace{25pt}
	\begin{subfigure}[b]{0.45\textwidth}
		\centering
		\includegraphics[scale=0.05]{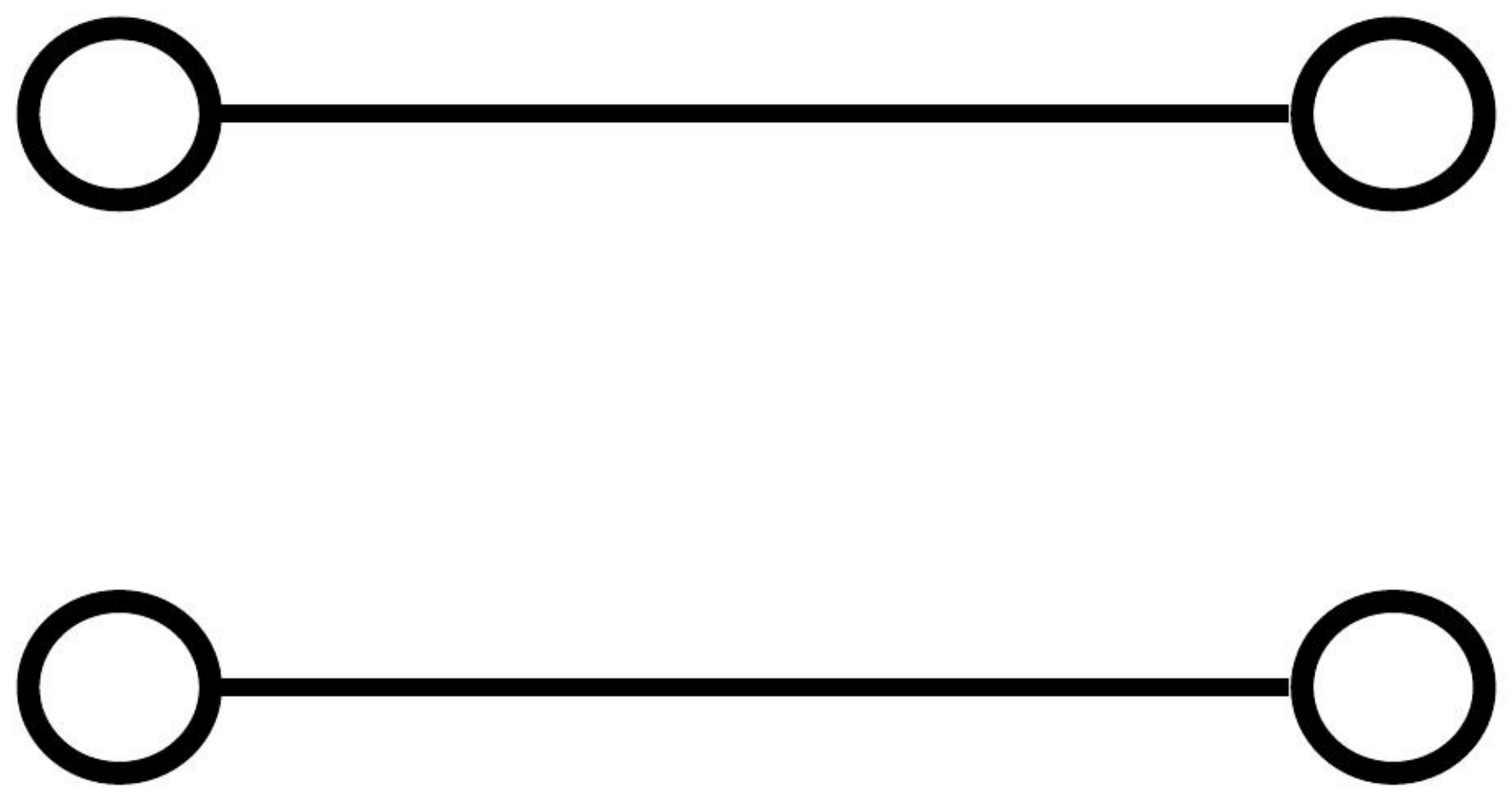}
		\label{fig:parallel-gadget}
		\caption{A parallel gadget.}
	\end{subfigure} \\	
	\begin{subfigure}[b]{0.8\textwidth}
		\centering
		\includegraphics[angle=90,scale=0.07]{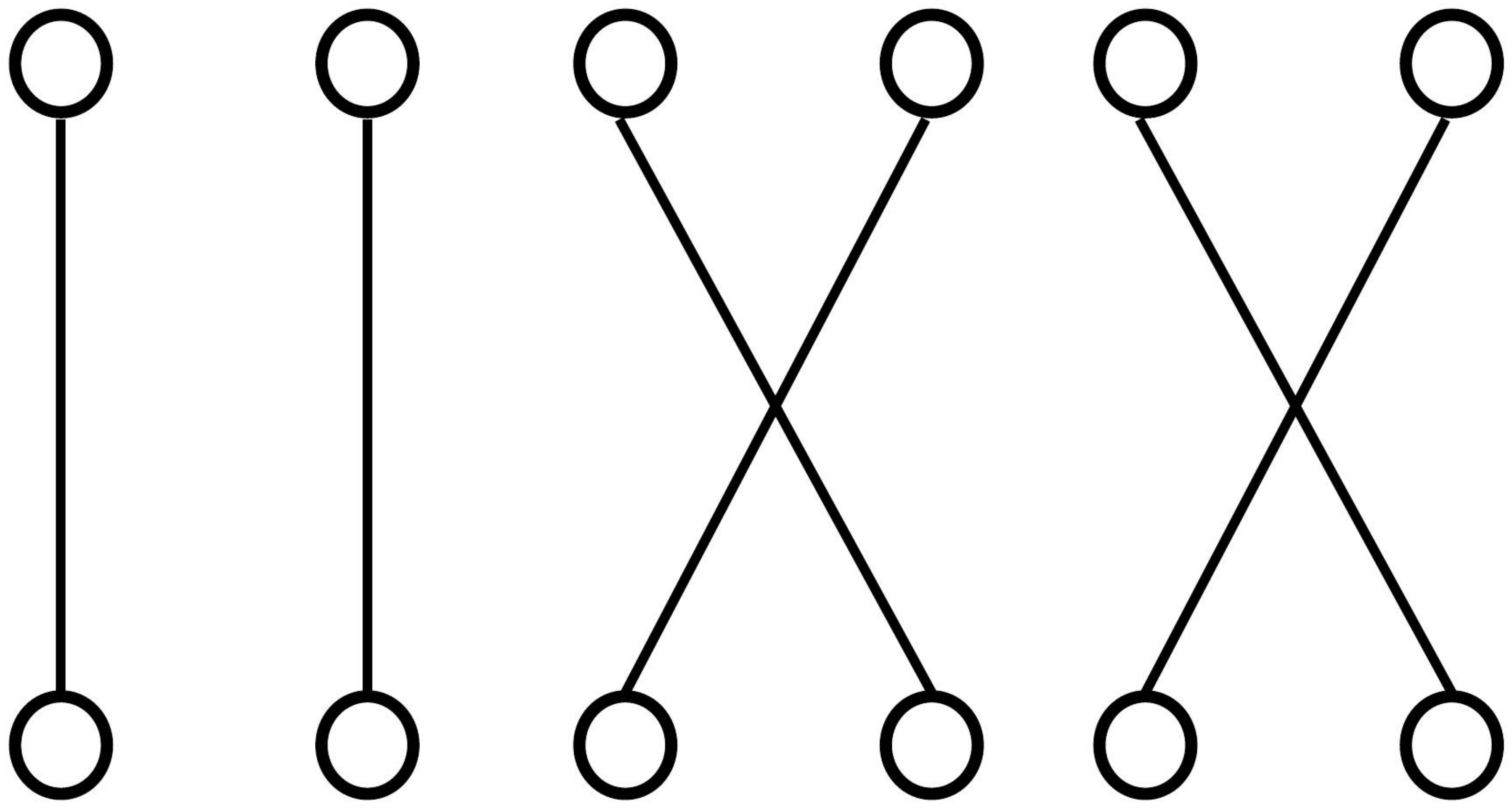}
		
		\caption{An example of a graph sampled from $\GG_0$ with $n = 6$, two cross gadgets and one parallel gadget.}
	\end{subfigure}
	\caption{The hard distribution $\GG_0$ and the gadgets.}\label{fig:harddistG0}
\end{figure}

We can immediately see that any graph sampled from $\GG_0$ is made up only of a single perfect matching of size $2 \cdot r_0 = n$. 
Learning a single edge of a gadget reveals the type of gadget, i.e., whether it is a parallel gadget or cross gadget, and thus also reveals the other edge. 
Furthermore, to learn the type of a gadget, any $0$-round ``protocol'', which has to commit to an answer without knowing the input, can {\em not} do better that random guessing.

\paragraph{Distributions for $p \geq 1$.}
The hard distribution $\GG_p$ for $p$-rounds will be a $(2n)$-vertex balanced bipartite graph constructed by {\em embedding} ``many'' hard instances for $(p-1)$-rounds sampled from $\GG_{p-1}(n_{p-1})$ into a balanced bipartite $(r_p, t_p)$-RS graph with $2b_p$ vertices (see \Cref{eq:params-match1} and \Cref{eq:parameters-match2} for the parameters).
The distribution is as follows and is illustrated in \Cref{fig:harddistGp}:
\begin{tbox}\textbf{The distribution $\GG_p$ on $n$ vertices for $p$-round protocols for $p \geq 1$ and sufficiently large $n \in \IN$.}
	\begin{enumerate}
		\item Pick a \underline{fixed} balanced bipartite $(r_p,t_p)$-RS graph $\Grs$ on $2b_p$ vertices $\Lrs \cup \Rrs$, from \Cref{prop:RS-gkk12}.  
		\item Sample $r_p \cdot t_p$ \underline{independent} instances from $\GG_{p-1}(n_{p-1})$ denoted by \[
		\cH := \set{H_{i,j} \mid i \in [t_p], j \in [r_p]}.\]  
		\vspace{-9mm}
		\item Create the graph $\tG := \embed(\cH \rightarrow \Grs)$ on vertices $\Lem \cup \Rem$ (as in \Cref{def:embed}). 
		\item For $a \in [p]$, player $\player{a}$ receives the edges in $\tG$ that
		correspond to subgraphs $H_{i,j}(a)$ for all $i \in [t_p]$ and $j \in [r_p]$ where $H_{i,j}(a)$ is the subgraph of $H_{i,j}$ given to the $a^\text{th}$ player in $\GG_{p-1}(n_{p-1})$.
		\item Pick $\istar$ \underline{uniformly at random} from $[t_p]$ and let $\Mrs{\istar}$ in $\Grs$ be the \emph{special matching}. Let $V^{\star} \subseteq \Lrs \cup \Rrs$ be the set of vertices in $\Grs$ incident to $\Mrs{\istar}$. 
		\item Player $\player{p+1}$ receives an arbitrary auxiliary \underline{perfect matching} $\Maux$ from $\set{(u,*) \mid u \in \paren{\Lem \cup \Rem} \setminus V^{\star}}$ to new sets of vertices $\Laux$ and $\Raux$ (each of size $\paren{b_p - r_p} \cdot n_{p-1})$. 
	\end{enumerate}
\end{tbox}

\begin{figure}[ht]
	\centering
	\includegraphics[scale=0.07]{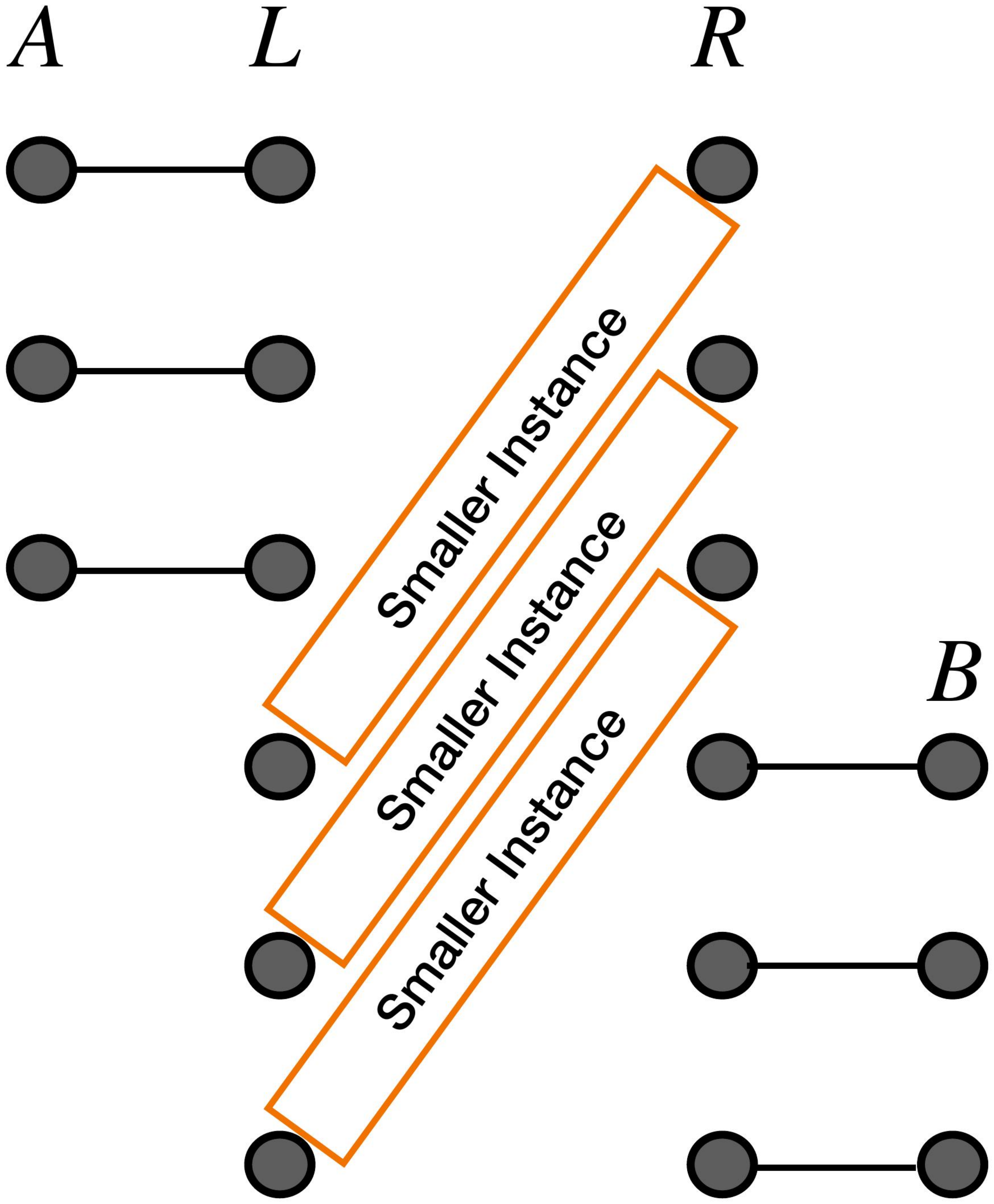}
	\caption{An illustration of the distribution $\GG_p$ for $p \geq 1$. When $p = 1$, the smaller instances are as in \Cref{fig:harddistG0}. Only the smaller instances in the special matching are shown. All edges of the outer RS graph are replaced also. The shaded circles represent groups of vertices.}
	\label{fig:harddistGp}
\end{figure}
The parameters $n_{p-1}$ and $b_p$ are chosen such that:
\begin{equation}\label{eq:params-match1}
	n = n_{p-1} \cdot (2b_p-r_p), \qquad
	b_p = (n_{p-1})^{(\log \log n_{p-1})^3}.
\end{equation}

Using the family of RS graph in \Cref{prop:RS-gkk12} for the construction of the distribution $\GG_p$ on $n$ vertices, we get the following parameters for $p \geq 1$:
\begin{equation}
	n= \paren{\frac{3}{2} + \drs} \cdot b_p \cdot n_{p-1}, \qquad r_p:= \paren{\frac{1}{2} - \drs} \cdot b_p, \qquad t_p:= b_p^{\crs/\log \log b_p}. \label{eq:parameters-match2}
\end{equation}

Lastly we define a new parameter $\kappars$ that will be used repeatedly in the calculations. 
\begin{equation}\label{eq:kap-del-def}
	\kappars = \frac{1-2\drs}{3+2\drs}.
\end{equation}

\paragraph{Notation.} 
{For any $a \in [p+1]$, let $\GG_{p}(a)$ denote 
	the distribution of edges in the input graph that are given to player $\player{a}$. The distribution $\GG_p$ is a joint distribution of $(\GG_{p}(1),\ldots,\GG_{p}(p+1))$. 
	
	An instance sampled from $\GG_p$ is a graph $G$ together with its edge partitioning $(G(1)$, $\ldots$, $G(p+1))$ between the $(p+1)$ players. 
	With a slight abuse of notation, we sometimes refer to $G$ sampled from $\GG_p$, i.e., $G \sim \GG_p$, as the entire instance where the edge partitioning between the players is implicit. 
	Any $G \sim \GG_p$ contains $r_p \cdot t_p$ instances $H_{1,1},\ldots,H_{t_p,r_p}$ sampled from $\GG_{p-1}(n_{p-1})$ that are embedded into $\Grs$ to give $\tG$. 
	We refer to these $(p-1)$-round instances as the sub-instances of $G$. For any $i \in [t_p], j \in [r_p]$, we further write $\tE_{i,j}$ to denote the edges inserted to $\tG$ as part of the embedding of $H_{i,j}$ into $\Grs$.
	
	Among the sub-instances of any $G \sim \GG_p$, there are 
	$r_p$ sub-instances that correspond to the edges in the special matching $\Mrs{\istar}$; we denote them by $\cHstar$ and refer to them as the \textbf{special sub-instances or subgraphs} of $G$.

	Any $G \in \GG_p$ has $r_p$ many special subgraphs. 
	Each special subgraphs is an instance sampled from $\GG_{p-1}(n_{p-1})$, which has its own corresponding $r_{p-1}$ many special subgraphs that are instances sampled from $\GG_{p-2}(n_{p-2})$. 
	Continuing this gives us $r_p \cdot r_{p-1} \cdot \ldots \cdot r_1$ many special subgraphs sampled from $\GG_0$, which is each made up of $r_0$ many (either parallel or cross) gadgets with $2$ vertex disjoint edges each, i.e., a matching of size $2$. 
	We define these gadgets/matchings as the \textbf{ultimate gadgets} $$\cMstar = \{ \Mstar_{j_p,j_{p-1}, \dots j_0} :
	j_p \in [r_p], j_{p-1} \in [r_{p-1}], \dots, j_0 \in [r_0]
	\}$$ whose union we denote as $\Mstar$ and call the \textbf{ultimate induced matching} (see \Cref{prop:s-ind-match} below).
	
	We need some key properties about any input graph $G$ sampled from $\GG_p$ for $p \geq 1$.}

\begin{property}\label[property]{prop:product-dist-match}
	For any $p \geq 0$, the distribution $\GG_p$ is a product distribution over the $(p+1)$-players. 
\end{property}
\begin{proof}
	When $p = 0$, there is only one player, and $\GG_0$ is a product distribution trivially. Let us assume that $\GG_{p-1}$ on $n_{p-1}$ vertices is a product distribution for the induction hypothesis. 
	
	In distribution $\GG_{p}$ with $n$-vertex graphs, we know that the input of $\player{p+1}$ is independent of the input of all the other players, as they receive one induced matching $\istar$ from $[t_p]$ uniformly at random. The $r_p \cdot t_p$ instances from $\GG_{p-1}$ are given to the other players, and they are sampled independently of each other. As distribution $\GG_{p-1}$ is a product distribution, the input of these players is independent of each other. 
\end{proof}

\begin{property} \label[property]{prop:ind-subG}
	The special subgraphs of $G$ are all vertex-disjoint. 
	Furthermore, the subgraph of $G$ induced by the vertices corresponding to the special subgraphs consists only of the edges of these subgraphs, i.e., the edges of $\Gstar_j$ for all $j \in [r_p]$. 
\end{property}
\begin{proof}
	Recall that, by definition, the vertices corresponding to the special subgraphs are the vertices of $G$ that correspond to the matching $\Mrs{\istar}$, i.e., $\{(u,*) \mid u \in V^{\star}\}$. 
	Let $\Gstar$ be the subgraph induced by these vertices.
	
	We have by \Cref{lem:match-embed-induced} that the edges of $\Gstar \cap \tG$ are exactly the edges of a vertex-disjoint union of the special subgraphs $H_{\istar,j}$ for $j \in [r_p]$. 
	Finally, $\Gstar \cap \tG = \Gstar$ since the only edges not in $\tG$ are the edges of the auxiliary matching $\Maux$, each of which is incident to at most a single vertex of $\{(u,*) \mid u \in V^{\star}\}$ and thus can not be in $\Gstar$.
\end{proof}

We can extend this property to all the ultimate small gadgets also.

\begin{property} \label[property]{prop:s-ind-match}
	The ultimate gadgets $\cMstar$ in $G$ are vertex-disjoint and their union forms an \underline{induced} matching $\Mstar$, namely, the ultimate induced matching, of size $(\kappars)^p \cdot n$.
\end{property}
\begin{proof}
	We first show by induction that the special gadgets $\cMstar$ are vertex-disjoint and that $\Mstar$ is an induced matching of size $2 \cdot r_p \cdot r_{p-1} \cdot \ldots \cdot r_0$.
	
	\subparagraph*{Base Case ($\GG_0$).} 
	By construction, the only edges of any graph $G \sim \GG_0$ are exactly the $r_0$ many vertex-disjoint gadgets that make up $G$, which are the ultimate gadgets of $\cMstar$. Furthermore, each gadget is made up of 2 vertex-disjoint edges and thus their union is an induced matching $\Mstar$ of size $2 \cdot r_0$.
	
	\subparagraph*{Inductive Step ($\GG_p$).}
	For the induction hypothesis, 
	assume that the statement holds for graphs sampled from $\GG_{p-1}$. By definition, any graph $G \sim \GG_p$ on $n$-vertex graphs has special subgraphs $H_{\istar, j}\sim \GG_{p-1}$ with $n_{p-1}$ vertices
	 for $j \in [r_p]$. By the induction hypothesis, the union of the ultimate gadgets $\cMstar_j$ of each $H_{\istar, j}$ form an induced matching $\Mstar_j$ of size $2 \cdot r_{p-1} \cdot \ldots \cdot r_0$.
	We also have that the union of the special gadgets $\cMstar$ in $G$ is $\Mstar = \bigcup_{j \in [r_p]} \Mstar_j$. 
	
	The vertices of $\Mstar_{j} $ and $\Mstar_{j'}$ are disjoint by \Cref{prop:ind-subG} for $j \neq j'$, $j , j' \in [r_p]$. 
	As each $\Mstar_j$ is an induced matching, their union $\Mstar$ is an induced matching in $G$ of size $r_p \cdot (2 \cdot r_{p-1} \cdot \ldots 
	\cdot r_0)$.
	
	Finally, using the parameters in \Cref{eq:params-match1} and \Cref{eq:parameters-match2}, we have that
	\begin{align*}
		|\Mstar| &= 2 \cdot r_p \cdot r_{p-1} \cdot \ldots \cdot r_0 \\
		&= (1/2 - \drs)^p \cdot (b_p \cdot \ldots \cdot b_1) \cdot n_0 \tag{since $2 \cdot r_0 = n_0$}\\
		&= (1/2 - \drs)^p (b_p \cdot \ldots \cdot b_2) \cdot (3/2 + \drs)^{-1} \cdot n_1  \tag{as $n_1 = (3/2 + \drs) \cdot n_0 \cdot b_1$}\\
		&= (1/2 - \drs)^p (3/2 + \drs)^{-p} \cdot n \tag{repeating the above two steps $(p-1)$ more times}\\
		&= \paren{\frac{1/2 - \drs}{3/2 + \drs}}^p \cdot n = (\kappars)^p \cdot n,
	\end{align*}
	by definition of $\kappars$ in \Cref{eq:kap-del-def} as required.
\end{proof}

\begin{property}\label[property]{prop:perf-matching}
	Any graph $G \sim \GG_p$ with $n$-vertices for $ p \geq 0$ has a perfect matching of size $n$. 
\end{property}

\begin{proof}
	The proof is again by induction. 
	
	\subparagraph*{Base Case ($\GG_0$).} 
	Graph $G \sim \GG_0$ is made of $(n/2)$ gadgets on four vertices each. Each gadget has a perfect matching by construction, and the union of these matchings gives a perfect matching of $G$. 
	
	\subparagraph*{Inductive Step ($\GG_p$).}
	For the induction hypothesis, 
	assume that the statement holds for graphs sampled from $\GG_{p-1}$. By definition, any graph $G \sim \GG_p$ has special subgraphs $H_{\istar, j}\sim \GG_{p-1}$ for $j \in [r_p]$.  Let $\Min_j$ for $j \in [r_p]$ be the perfect matching of the inner instance $H_{\istar, j}$.
	We pick the following edges in our matching. 
	\[
	\Maux \cup \bigcup_{j \in [r_p]} \Min_j
	\]
	The vertices of $\Maux$ cover all vertices of $G$ not incident to any special subgraph $H_{\istar, j}$ by construction. We also have that the subgraphs $H_{\istar, j}$ are vertex-disjoint by \Cref{prop:ind-subG}. Thus, taking the union of these matchings is sufficient to give a perfect matching of $G$.
\end{proof}

\begin{property} \label[property]{prop:perfGnotM}
	Let $\Gbar$ be the subgraph of $G$ induced by the vertices which are \emph{not} incident to the ultimate matching $\Mstar$. Then, $\Gbar$ has a vertex cover of size
	$$ (1- (\kappars)^p) \cdot n.$$
\end{property}
\begin{proof}
	For the base case of the induction on $p$, we see that \emph{all} the vertices of $G\sim \GG_0$ are incident to some edge in $\Mstar$. This is because $\Mstar$ contains all the edges of all the $(n/2)$ gadgets of $G$. There is a vertex cover of size $0$ trivially.  
	
	\subparagraph*{Inductive Step ($\GG_p$).}
	Now, suppose that the statement holds for graphs sampled from $\GG_{p-1}(n_{p-1})$.
	Again, we look at the special subgraphs $H_{\istar, j}\sim \GG_{p-1}(n_{p-1})$ for $j \in [r_p]$.  These graphs are vertex disjoint by \Cref{prop:ind-subG}, and they have a vertex  cover $\Vin_j$ that covers all the edges not in $\Mstar_j$ (the ultimate matching of $H_{\istar, j}$).  The vertex cover we pick contains:
	\[
	\textnormal{Vertices of $\tG$ which matching $\Mrs{\istar}$ does not touch and } \bigcup_{j \in [r_p]} \Vin_j.
	\]
	
	First, we argue that this subset of vertices covers the required edges. Any edge in $\Maux$ is covered by the first part, as $\Maux$ only incident to vertices not corresponding to $\Mrs{\istar}$. Any other edge in $\tG$ which is not added because of $\Mrs{\istar}$ is also covered by the first part. We know $\Vin_j$ is a vertex cover of the subgraph of $H_{\istar, j}$ with vertices not touched by $\Mstar_j$ by the induction hypothesis. As $\Mstar = \cup_{j \in [r_p]} \Mstar_j$, we are done. 
	
	It remains to analyze the size of our vertex cover. By the induction hypothesis, the size is exactly:
	\begin{align*}
		&r_p \cdot \paren{1- (\kappars)^{p-1}} \cdot n_{p-1} + \textnormal{Number of vertices of $\tG$ not touching $\Mrs{\istar}$} \tag{from each of the $r_p$ many $H_{\istar, j}$}\\ 
		&= n \cdot \paren{1- (\kappars)^{p-1}} \cdot \paren{\frac{1-2\drs}{3+2\drs}} + \textnormal{Number of vertices of $\tG$ not touching $\Mrs{\istar}$} \tag{by \Cref{eq:params-match1}, \Cref{eq:parameters-match2}}\\
		&=   n \cdot \paren{1- (\kappars)^{p-1}} \cdot \paren{\frac{1-2\drs}{3+2\drs}} + 2(b_p-r_p) \cdot n_{p-1} \\
		&= n \cdot \paren{1-(\kappars)^{p-1}} \cdot \paren{\frac{1-2\drs}{3+2\drs}}  + \paren{\frac{2+4\drs}{3+2\drs}} \cdot n \tag{again by \Cref{eq:params-match1}, \Cref{eq:parameters-match2}}  \\
		&= n \cdot \paren{1-(\kappars)^{p-1}} \cdot \kappars  + (1-\kappars) \cdot n \\
		&= n \cdot \paren{1- (\kappars)^p},
	\end{align*}
	which completes the proof.
\end{proof}

\begin{property}\label[property]{prop:approx-specialM}
	For positive $\eps \leq (\kappars)^p$, any $(1 - \eps)$-approximate matching of $G \sim \GG_p$ contains at least an $\paren{1-(\eps/(\kappars)^p)}$ fraction of the edges of the ultimate matching $\Mstar$.
\end{property}
\begin{proof}
	Let $M$ be a $(1 - \eps)$-approximate maximum matching of $G$.
	By \Cref{prop:perf-matching}, $G$ has a perfect matching of size $n$ which implies that $M$ has at least $(1-\eps)\cdot n$ edges.

	The size of the largest matching that does not use any edges of $\Mstar$ is at most the vertex cover of $\Gbar$ from \Cref{prop:perfGnotM} by Konig's theorem. Therefore, the size of such a matching is at most, 
	$$
	n \cdot \paren{1- (\kappars)^p}.
	$$
	Hence, $M$ must contain edges of $\Mstar$. In particular, we have that $M$ contains at least,
	\begin{align*}
		&\geq \paren{1 - \eps} \cdot n- \paren{1 - (\kappars)^p} \cdot n\\
		&= \paren{(\kappars)^p- \eps} \cdot n
	\end{align*}
	many edges of the ultimate matching $\Mstar$. 
	Finally, by \Cref{prop:s-ind-match}, the number of edges in $\Mstar$ is $(\kappars)^p\cdot n$.
	The matching $M$ contains at least $\paren{1-\eps \cdot \paren{1/(\kappars)^p}}$ fraction of the edges of the matching $\Mstar$. 
\end{proof}

These properties tell us that any protocol for $\beta$-approximate maximum matching with $\beta = (1 -  \eps)$ with its input distributed according to $\GG_p$ necessarily learns at least $(1-\eps/(\kappars)^p)$ edges of $\Mstar$.
Now, suppose the goal was to use the output of the protocol and determine the edges of a uniform random gadget from $\cMstar$. Learning any one of the two edges of the gadget would reveal the other edge also (be it a parallel gadget or a cross gadget). 
Thus, any $(1-\eps)$-approximation reveals at least $(1-\eps/(\kappars)^p)/2$ ultimate gadgets. 
We formalize this intuition by using search sequences and predicates to point to ultimate gadgets. Large matchings will be used to solve these predicates.

Let $J := (j_p, j_{p-1}, \dots, j_0 )$ where $j_i \in [r_i]$ for $0 \leq i \leq p$ be a \textbf{search sequence} that identifies a \textbf{ultimate gadget} $\Mstar_J := \Mstar_{j_p, j_{p-1}, \dots, j_0} \in \cMstar$ of an input graph $G \sim \GG_p$. We define the \textbf{search predicate} as 
\begin{align*}
	\PP_p (G, J) := 
	\begin{cases}
		0 & \text{~if the special gadget $\Mstar_J$ is a parallel gadget and }\\
		1 &\text{~if it is a cross gadget.} 	
	\end{cases}
\end{align*}

For any $(p+1)$-player protocol $\pi$ on input $G$ sampled from $\GG_p$, we say that $\pi$ \textbf{solves the search predicate} $\PP_p$ for the input graph $G$ if, given only the transcript of messages communicated in $\pi$ and a uniform random search sequence $J = (j_p, \dots, j_0)$ (hidden from the protocol and with no further access to $G$), $\PP_p(G,J)$ can be deterministically determined. 
It is important to note that $J$ is {\em not} part of the input to the protocol and is probabilistically independent from the graph $G$.

\begin{lemma} \label{lem:mbmtopred}
	For every $0 \leq \tau \leq 1$ and constant $0 < \eps \leq(\kappars)^p$, any $p$-round $(p+1)$-player protocol $\pi$ for $(1 - \eps)$-approximate maximum matching on input $G \sim \GG_p$ that succeeds with probability at least $\tau$ and with communication cost $s$ can be altered and then used to solve the search predicate $\PP_p$ for a uniform random search sequence $J$ with probability at least 
	$$ \frac12 \cdot (1+  \tau \cdot (1-\eps/\kappars^p))$$
	with communication cost $s + O\paren{n \log n}$.
\end{lemma}
\begin{proof}
	Let $M$ be the $(1 - \eps)$-approximate maximum matching of $G$ that is output by the protocol $\pi$ and can thus be determined from its transcript.
	We use the output to solve $\PP_p$ as follows:
	If $M$ contains at least one edge of the ultimate gadget $\Mstar_J \in \cMstar$, we can correctly determine $\PP_p(G,J)$. Otherwise, we output $\PP_p(G,J) = 0$, which is essentially a random guess for $G \sim \GG_p$. 
	
	When the protocol outputs a matching, we need to explicitly identify what edges from the ultimate gadgets and the ultimate induced matching have been sent in the transcript. We do this by recording the edges of $\Maux$ explicitly at every level:  in the first round, player $\player{p+1}$ additionally communicates the edges $\Maux$ goven to them ($\Maux$ from  the definition of $\GG_p$) that identifies the $r_p$ many special subgraphs sampled from $\GG_{p-1}$. In the next round, $\player{p}$ outputs the $\Maux$ matchings from \emph{all} $r_p$-many special subgraphs sampled from $\GG_{p-2}$. 
	Continuing this way, we know that the edges in $M$ which have not been revealed as a part of $\Maux $ at some level, are from the ultimate induced matching. 
	
	To output these auxiliary edges explicity, one edge takes $O(\log n)$ communication and there are totally only $n$ edges, as they form a matching. Hence, this adds at most $ O(n\log n)$ factor to the communication. 
	
	Let us analyze the probability of success of this strategy. 
	We have three cases:
	\begin{enumerate}[label=$(\roman*)$]
		\item When $\pi$ succeeds, which happens with probability $\tau$, and the transcript identifies the gadget that $J$ points to, the protocol can solve the predicate exactly. By \Cref{prop:approx-specialM}, we know that the transcript contains at least $(1-\eps/(\kappars)^p)$-fraction of the edges of $\Mstar$. The probability that it contains at least one edge from a randomly chosen $J$ is also at least $(1-\eps/(\kappars)^p)$.
		\item When $\pi$ succeeds, but it does not identify the gadget that $J$ points to, which happens with probability at most $\tau \cdot (\eps /(\kappars)^p)$, we output a random guess. This succeeds with probability $1/2$.
			\item When the algorithm fails, which happens with probability $(1-\tau)$, the algorithm may not output a large matching, and the required edges which identify the gadget $J$ may not exist. In this case also, we output a random guess that succeeds with probability $1/2$.\footnote{An algorithm that fails may also output invalid edges that do not exist in the graph, leading us to make mistakes. However, in many existing works on lower bounds for streaming matchings \cite{GoelKK12,Kapralov13,Konrad15}, there is an implicit assumption that the algorithm, even when it fails, outputs valid edges which may form a smaller matching. We use the same assumption in our paper.}
	\end{enumerate}
	Hence, the total probability of success is at least:
	\[
	\tau \cdot (1-\eps/(\kappars)^p) + \tau (\eps /(\kappars)^p) \cdot (1/2) + (1-\tau) \cdot (1/2) = \frac12 + \tau \cdot (1-\eps/\kappars^p) \cdot \frac12. \qedhere
	\]
\end{proof}

With that, to prove \Cref{thm:cc-multi-MBM}, it is enough to only consider protocols that solve the search predicate. We prove the following lower bound on the communication cost of a protocol that solves a search predicate with large enough advantage.

\begin{lemma} \label{lem:main-cc-matchinsert}
	For every constant integer $p \geq 1$ any $p$-round protocol $\pi$ that, given input $G \sim \GG_p(n_p)$ and a uniform random search sequence, solves predicate $\PP_p$ with probability at least $$\delta \geq \frac12 + \frac{p}{12(p+1)}$$
	has communication cost 
	$$\cc{\pi} \geq  \paren{\frac1{(144\raoconst(p+1)^2)}}^p  n_p^{ 1+1 / (2 \log \log n_p)^{3p}} $$
	where $\raoconst$ is the constant from \Cref{eq:msg-compress-final} and $\crs$ is the constant from \Cref{prop:RS-gkk12}.
\end{lemma}

The proof of this lemma follows using a round-elimination argument where, by starting with a `very good' $p$-round $(p+1)$-player protocol, we can eliminate the first round and final player to obtain a still `very good' $(p-1)$-round $p$-party protocol. 
Repeatedly doing this eventually leads to a `too good to be true' protocol for the base case $p=0$, which trivially can not exist,  a contradiction. This strategy relies on the following lemma:

\begin{lemma}\label{lem:roundelim}
	
		For every $p \geq 1$, $0 < \eta, \delta \leq 1$ and integer $s \geq 1$, suppose that there exists a deterministic $p$-round $(p+1)$-player protocol $\pi_p$ that, given input $G \sim \GG_p$ and a uniform random search sequence $J$, solves $\PP_p(G, J)$ with probability at least $\delta$ and communication cost $\cc{\pi_p} \leq s$.
		
	Then, for constant $\raoconst$ (in \Cref{eq:msg-compress-final}), there exists a deterministic $(p-1)$-round $p$-player protocol $\pi_{p-1}$ that, given input $G' \sim \GG_{p-1}$ and a uniform random search sequence $J'$, solves $\PP_{p-1}(G',J')$ with probability at least 
	$$\delta - \eta - \sqrt{\frac{s}{2 \cdot r_p \cdot t_p}}$$
	and communication cost
	$$\cc{\pi_{p-1}} \leq \frac{\raoconst}{\eta} \cdot \paren{\frac{s}{r_p} + p^2}$$
	where $\raoconst \geq 1$ is an absolute constant defined in~\Cref{eq:msg-compress-final} (derived from \Cref{prop:msg-compress}). 
\end{lemma}

We can use \Cref{lem:roundelim} to prove the lower bound in \Cref{lem:main-cc-matchinsert}.
\begin{proof}[Proof of \Cref{lem:main-cc-matchinsert}]
	We prove the lower bound for deterministic protocols, and use the easy direction of Yao's minimax principle to extend the bound to all randomized protocols.
	
	Before proving the result by induction, we give the following useful bounds for every constant $p \geq 1$: Using the parameters in \Cref{eq:params-match1}, \Cref{eq:parameters-match2} and since $b_p = n_{p-1}^{(\log \log n_{p-1})^3}$, we have that 
	\begin{align*}
		\begin{split}
			n_p &:= \paren{\frac{3}{2} + \drs} \cdot b_p \cdot n_{p-1} \\
			&= \paren{\frac{3}{2} + \drs} \cdot (n_{p-1})^{1+( \log \log n_{p-1})^3}
		\end{split}
		, \quad
		\begin{split}
			r_p &:= \paren{\frac{1}{2} - \drs} \cdot b_p \\
			&= \paren{\frac{1}{2} - \drs} \cdot (n_{p-1})^{(\log \log n_{p-1})^3}
		\end{split}
	\end{align*}
	\begin{align}\label{eq:t-ub-np-1}
		t_p &= b_p^{\crs/\log \log b_p} \\
		&\geq \paren{n_{p-1}^{( \log \log n_{p-1} )^3}}^{1/ (2 \cdot \log \log n_{p-1})}
		\tag{since $\log \log b_p \leq 2 \cdot \log \log n_{p-1}$} \\
		&\geq n_{p-1}^4. \tag{for large enough $n_{p-1}$}
	\end{align}
	Let $s_p := \paren{\frac1{(144\raoconst(p+1)^2)}}^p  \cdot n_p^{1 + \paren{1 / (2 \log \log n_p)}^{3p}}$ (as in the statement). It then follows that
	\begin{align}\label{eq:sprpp2-ub}
		\frac{s_p}{r_p} + p^2 
		&=  \paren{\frac1{(144\raoconst(p+1)^2)}}^p  \cdot n_p^{ 1/(2\log \log n_p)^{3p}} \cdot \frac{n_p}{r_p} + p^2 \\
		&= \paren{\frac1{(144\raoconst(p+1)^2)}}^p  \cdot n_p^{ 1/(2\log \log n_p)^{3p}} \cdot \paren{\frac{3/2 + \drs}{1/2 - \drs}} \cdot n_{p-1} + p^2 \tag{from the parameters above}\\
		&\leq \paren{\frac1{(144\raoconst(p+1)^2)}}^p  \cdot n_p^{ 1/(2\log \log n_p)^{3p}} \cdot (3 + 16\drs) \cdot n_{p-1} + p^2 \tag{for $\drs \leq 1/4$} \\
		&\leq \paren{\frac1{(144\raoconst(p+1)^2)}}^p  \cdot n_p^{ 1/(2\log \log n_{p-1})^{3p}} \cdot (3 + 16\drs) \cdot n_{p-1}+ p^2 \tag{since $\log \log n_p \geq \log \log n_{p-1}$} \\
		&\leq \paren{\frac1{(144\raoconst(p+1)^2)}}^p  \cdot n_{p-1}^{ 1/(2\log \log n_{p-1})^{3(p-1)}} \cdot (3 + 16\drs) \cdot n_{p-1} + p^2 \tag{since $\log n_p \leq 2(\log \log n_{p-1})^3 \cdot \log n_{p-1})$} \\
		&\leq \paren{\frac1{(144\raoconst(p+1)^2)}}^p  \cdot 12 \cdot n_{p-1}^{ 1/(2\log \log n_{p-1})^{3(p-1)}} \cdot n_{p-1}. \tag{since $p^2 = o(n_{p-1})$ and $\drs < 1$} 
	\end{align}
	We also have that
	\begin{align}\label{eq:sprptp-ub}
		\frac{s_p}{r_p \cdot t_p} &\leq \paren{\frac1{(144\raoconst(p+1)^2)}}^p   n_p^{ \paren{1 / (2 \log \log n_p)}^{3p}} \cdot \frac{n_p}{r_p} \cdot \frac1{t_p}\\
		&\leq  \paren{\frac1{(144\raoconst(p+1)^2)}}^p   n_p^{ \paren{1 / (2 \log \log n_p)}^{3p}} \cdot \frac{n_p}{r_p \cdot n_{p-1}} \cdot \frac1{n_{p-1}^3} \tag{as $t_p \geq n_{p-1}^4$ from \Cref{eq:t-ub-np-1}} \\
		&=  \paren{\frac1{(144\raoconst(p+1)^2)}}^p   n_p^{ \paren{1 / (2 \log \log n_p)}^{3p}} \cdot \paren{\frac{3+2\drs}{1-2\drs}} \cdot \frac1{(n_{p-1})^3}\tag{from \Cref{eq:params-match1} and \Cref{eq:parameters-match2}}\\
		&\leq  \paren{\frac1{(144\raoconst(p+1)^2)}}^p  \cdot (11) \cdot (n_{p-1})\cdot \frac1{(n_{p-1})^3} \tag{as $\drs<1$ and $\log n_p < 2(\log \log(n_{p-1})^3 \cdot \log n_{p-1}$} \\
		&\leq   \paren{\frac1{(144\raoconst(p+1)^2)}}^p  \cdot (11) \cdot (n_{p-1})^{-2} \\
		&< \paren{\frac1{13p(p+1)}}^2.\tag{for large enough $n_{p-1}$} 
	\end{align}
	With these bounds, we are now ready to proceed with the inductive argument.
	
	\subparagraph*{Base Case (1-round).} Towards a contradiction, let $\pi_1$ be a $1$-round deterministic protocol that, given input $G \sim \GG_1(n_1)$ and a uniform random search sequence, solves $\PP_1$ with probability at least $13/24 $, as in the statement, and communication cost $\cc{\pi_1} < s_1$, which is strictly less than the statement. 
	We use \Cref{lem:roundelim} with $\eta = 1/(72)$ to get a protocol with probability of success at least:
	\[
	\frac{13}{24} - \frac1{72} - \frac1{26 \sqrt{2}} > \frac12,
	\]
	by our bound on $(s_1/(t_1 r_1))$ above.
	This protocol is deterministic, and uses $0$-rounds to solve predicate $\PP_0$ on input graph $G_0 \sim \GG_0$ and a uniform random search sequence.
	Regardless of the communication cost, this is a contradiction for $0$-round protocols as they cannot do better than random guessing. Therefore, $\pi_1$ must have communication cost $\cc{\pi_1} \geq s_1$.
	
	\subparagraph*{Inductive Step ($p$-round).} 
	Now, suppose that the statement is true for $(p-1)$-round deterministic protocols. 
	Let $\pi_p$ be a $p$-round deterministic protocol that, given input $G \sim \GG_p$ and a uniform random search sequence, solves $\PP_p$ with probability at least $ \frac12 + \frac{p}{12(p+1)}$, as in the statement, and communication cost $\cc{\pi_p} < s_p$. 
	By \Cref{lem:roundelim} with $\eta=1/(36p(p+1))$ and the above bounds, there exists a $(p-1)$-round deterministic protocol that, given input $G' \sim \GG_{p-1}(n_{p-1})$ and a uniform random search sequence, solves $\PP_{p-1}$ with probability at least 
	\begin{align*}
		&\frac12 + \frac{p}{12(p+1)} -   \frac{1}{36p(p+1)} - \sqrt{s_p/(2 \cdot r_p \cdot t_p)} \\
		&\geq \frac12 + \frac{p}{12(p+1)} -   \frac{p}{36p(p+1)}-  \frac{1}{13\sqrt{2}p(p+1)} \tag{from \Cref{eq:sprptp-ub}}\\
		&\geq \frac12 + \frac{(p-1)}{12p},
	\end{align*} 
	and communication cost 
	\begin{align*}
		\cc{\pi_{p-1}} &< \frac{\raoconst}{\eta} \cdot (\frac{s_p}{r_p} + p^2)  \\
		&\leq \raoconst \cdot (36p(p+1))\cdot  \paren{\frac1{(144\raoconst(p+1)^2)}}^p  \cdot 12 \cdot n_{p-1}^{ 1+1/(2\log \log n_{p-1})^{3(p-1)}} \tag{from \Cref{eq:sprpp2-ub}} \\
		&\leq  \paren{\frac1{(144\raoconst \cdot p^2)}}^{p-1} \cdot n_{p-1}^{ 1+1/(2\log \log n_{p-1})^{3(p-1)}} = s_{p-1}.
	\end{align*}
	This contradicts the induction hypothesis and, hence, $\pi_p$ must have communication cost $\cc{\pi_p}$ at least $s_p$.
\end{proof}

We now explicitly complete the proof of \Cref{thm:cc-multi-MBM}.
\begin{proof}[Proof of \Cref{thm:cc-multi-MBM}]
	Let $\pi$ be any $p$-round $(p+1)$-protocol  for $(1-\eps)$-approximate maximum matching with
	\begin{equation*}
		\eps <  (\kappars)^p, \qquad \kappars = \frac{1-2\drs}{3+2\drs},
	\end{equation*}
	probability of success at least $ (2/3)$ and communication cost $s$ (as in the statement).
	
	By \Cref{lem:mbmtopred}, protocol $\pi$ on input $G \sim \GG_p(n)$ can be made into a protocol $\pi'$ with communication cost $s + O(n \log n)$ that, given input $G$ and a uniform random search sequence, solves the search predicate $\PP_p$ with probability at least 
	\[
\frac12 \cdot (1+(2/3) \cdot (1-\eps/\kappars^p)),
	\]
	for $\eps <(\kappars)^p$. Let $\eps = (\kappars)^p - \rho$ for some positive constant $\rho$. Then, the probability of success of protocol $\pi$ is at least,
	\[
	\frac12 \cdot (1+ \frac23 \cdot \rho \cdot (1/\kappars)^p)  = \frac12 + c_{p,\rho},
	\]
	for some constant $c_{p,\rho} > 0$ that depends only on $p$ and $\rho$. We can run protocol $O(1)$ times in parallel to boost this success probability to $7/12$, and taking a majority of the answers. This incurs a cost of $O(1)$ multiplicative factor in the communication. The new probability of success is at least,
	\[
 \frac{7}{12} > \frac12 + \frac{p}{12(p+1)}.
	\]
	Therefore, by \Cref{lem:main-cc-matchinsert}, protocol $\pi'$ has communication cost 
	$$\cc{\pi'} = O(s) + O(n \log n) \geq   \paren{\frac1{(144\raoconst(p+1)^2)}}^p \cdot n^{1 + 1 / (2\log \log n)^{3p}} = n^{1+\Omega(1/\log \log n)^{3p}},$$
	for $p = o(\log \log n/\log \log \log n)$, as $\raoconst$ is a constant. 
	Finally, by rearranging, we have that
	$\cc{\pi} = s = n^{1 + \Omega\paren{1/(\log \log n)^{3p}}}$, which proves the statement.
\end{proof}

The only remaining component in the proof is \Cref{lem:roundelim}. This proof can be found in \Cref{app:roundelim}.

% !TeX root = main.tex 
%!TEX root = main.tex
\section{Round Elimination of Hierarchical Embedding}\label{app:roundelim}

This section contains the proof of \Cref{lem:roundelim}. We restate the lemma here for convenience. 

\begin{lemma*}[Restatement of \Cref{lem:roundelim}]
	For every $p \geq 1$, $0 < \eta, \delta \leq 1$ and integer $s \geq 1$, suppose that there exists a deterministic $p$-round $(p+1)$-player protocol $\pi_p$ that, given input $G \sim \GG_p$ and a uniform random search sequence $J$, solves $\PP_p(G, J)$ with probability at least $\delta$ and communication cost $\cc{\pi_p} \leq s$.
	
	Then, for constant $\raoconst$ (in \Cref{eq:msg-compress-final}), there exists a deterministic $(p-1)$-round $p$-player protocol $\pi_{p-1}$ that, given input $G_{p-1} \sim \GG_{p-1}$ and a uniform random search sequence $J'$, solves $\PP_{p-1}(G_{p-1},J')$ with probability at least 
	$$\delta - \eta - \sqrt{\frac{s}{2 \cdot r_p \cdot t_p}}$$
	and communication cost
	$$\cc{\pi_{p-1}} \leq \frac{\raoconst}{\eta} \cdot \paren{\frac{s}{r_p} + p^2}$$
	where $\raoconst \geq 1$ is an absolute constant defined in~\Cref{eq:msg-compress-final} (derived from \Cref{prop:msg-compress}). 
\end{lemma*}

We use the notation and tools from
\Cref{sec:info-theory,subsec:info-cost-compression}.

Fix $p \geq 1$ and suppose there exists a deterministic $p$-round $(p+1)$-player protocol $\pi_p$ that, on input $G \sim \GG_p$ and a uniform random search sequence $J$, solves $\PP_p(G, J)$ with probability at least $\delta$ and communication cost $\cc{\pi_p} \leq s$ (as in the statement).

Given an input $G_{p-1} \sim \GG_{p-1}$, we want to use $\pi_p$ to answer the search predicate $\PP_{p-1}(G_{p-1},J')$ for a uniform random search sequence $J' = (j_{p-1}, j_{p-2}, \dots, j_0)$, i.e., solve $\PP_{p-1}(G_{p-1}, J')$.

The proof consists of the following two broad steps:

\begin{enumerate}[label=$(\roman*)$]
	\item A low communication $p$-round $p$-party protocol that solves predicate $\PP_{p-1}$, dubbed the player elimination step. 
	\item A low communciation $(p-1)$-round $p$-party protocol that solves predicate $\PP_{p-1}$, called the round elimination step. 
\end{enumerate}

\subsubsection*{Step 1: Player Elimination}

We first construct a low information cost protocol with one fewer player.

\begin{ourbox}
	\textbf{A $p$-round $p$-party protocol $\sigma$ for $\PP_{p-1}(G_{p-1}, J')$ on input graph $G_{p-1} \sim \GG_{p-1}$ and search sequence $J'$}: 
	\begin{enumerate}
		\item Sample indices $\istar \in [t_p]$ and $\jstar \in [r_p]$ uniformly at random using \underline{public randomness}. 
		\item For every $a \in [p]$, each player $\player{a}$ samples 
		$
		\Paren{G(a) \sim \GG_p \mid \ristar=\istar , \rH_{\istar,\jstar}(a) = G_{p-1}(a)}
		$
		using \underline{private randomness}. 
		Let $G=(G(1), G(2),\ldots,G(p+1))$ be the 
		input of all players in $\GG_p$ where the input of the simulated player $\player{p+1}$ is fixed by $\ristar = \istar$.
		\item The players run $\pi_p$ on input $G$ by \underline{simulating} the messages of player $\player{p+1}$ in $\pi_p$ given that $\istar$ is public knowledge. 
		\item At the end, for the search sequence $J = (j_{p-1}, \ldots, j_0)$, to answer the search predicate $\PP_{p-1}(G_{p-1},J)$, return the answer of $\pi_p$ for $\PP_p(G, (\jstar, j_{p-1},j_{p-2},\ldots,j_0))$.  
	\end{enumerate}
\end{ourbox}

Define the \textbf{distribution $\reald$} as the distribution of the instance $G$ obtained in the protocol $\sigma$ when $G_{p-1} \sim \GG_{p-1}$. We claim that this is the ``right'' distribution. 

\begin{observation}\label[observation]{obs:reald-is-real}
	Distribution $\reald$ is the same as $\GG_p$. 
\end{observation}
\begin{proof}
	We know from \Cref{prop:product-dist-match} that $\GG_p$ is a product distribution of $\GG_{p}(1) \times \ldots \GG_{p}(p+1)$. Moreover, each $\GG_{p}(i)$ for $i \in [p]$ is
	a collection of \emph{independent} instances $H_{i,j}(a) \sim \GG_{p-1}(a)$ and $\istar \in [t_p]$ is chosen uniformly at random in $\GG_p$ to define the input of player $\player{p+1}$. This matches 
	exactly the distribution of $\reald$ in $\sigma$, given that $H_{p-1}(\ell)$ used by player $\player{\ell}$ is sampled from $\GG_{p-1}(\ell)$ by definition for $\ell \in [p]$. 
\end{proof}

Now, we can show that the information cost of $\sigma$ is much lower than that of $\prot_p$.

\begin{claim} \label{clm:prot-low-info}
	Protocol $\sigma$ solves $\PP_{p-1}$ with probability at least $\delta$ and information cost
	$$ \ic{\sigma}{\GG_{p-1}} \leq \frac{1}{r_p} \cdot \ic{\pi_p}{\GG_p}.$$
\end{claim}
\begin{proof}
	We know from \Cref{obs:reald-is-real} that the distribution of $G$ in $\sigma$ is the same as that of $\GG_p$.
We further have that any message communicated by the simulated player $\player{p+1}$ is determined by the transcript of messages on the blackboard at that point and the player's input, which is fixed by $\istar$ and publicly known. 
Thus, the transcript of messages communicated by protocol $\sigma$ on input $G_{p-1} \sim \GG_{p-1}$ is identical to that of protocol $\pi_p$ on input $G \sim \GG_p$.
	
	Since $J = (j_{p-1}, \ldots, j_0)$ is a uniform random search sequence for $\PP_{p-1}$ and $\jstar \in [r_p]$ is a uniformly chosen, we also have that $(\jstar, j_{p-1}, \ldots, j_0)$ is a uniform random search sequence for $\PP_p$. 
	Recall that $H_{\istar,\jstar} = G_{p-1}$ and, by definition, $\jstar$ in the search sequence points to the special subgraph $H_{\istar,\jstar}$. 
	Hence, it follows that
	$\PP_p(G,(\jstar, j_{p-1}, \ldots, j_0)) = \PP_{p-1}(H_{\istar,\jstar}, J)$
	and $\sigma$ returns the correct answer when $\pi_p$ succeeds, i.e., with probability at least $\delta$.
	
	Finally,	
	we let $\rPi, \rSigma$ denote the random variables of the transcript of messages communicated in $\pi_p$ and $\sigma$ respectively.
	Recall that $\rSigma$ is $\rPi$ without the messages communicated by the simulated player $\player{p+1}$.
	Furthermore, we have that the random variables $\rH_{1,1}, \ldots, \rH_{t_p,r_p}, \ristar, \rjstar$ are mutually independent of each other.
	Therefore, we have that
	\begin{align*}
		\ic{\sigma}{\GG_{p-1}} 
		&= \mi{\rG_{p-1}}{\rSigma \mid \ristar, \rjstar} \tag{by \Cref{def:ext-info}} \\
		&= \frac{1}{r_p} \cdot \sum_{\jstar=1}^{r_p} \mi{\rG_{p-1}}{\rSigma \mid \ristar, \rjstar = \jstar} \tag{since $\rjstar$ is a uniform distribution} \\
		&= \frac{1}{r_p} \cdot \sum_{\jstar=1}^{r_p} \mi{\rH_{\ristar,\jstar}}{\rSigma \mid \ristar, \rjstar = \jstar}
		\tag{since $H = H_{\istar, \jstar}$ for any choice of $\istar$} \\
		&= \frac{1}{r_p} \cdot \sum_{\jstar=1}^{r_p} \mi{\rH_{\ristar,\jstar}}{\rSigma \mid \ristar}
		\tag{as event $\rjstar = \jstar \perp $ all $\rH_{i,j}, \rSigma, \ristar$} \\
		&\leq \frac{1}{r_p} \cdot \sum_{\jstar=1}^{r_p} \mi{\rH_{\ristar,\jstar}}{\rSigma \mid \ristar, \rH_{\ristar, <\jstar}}
		\tag{by \Cref{prop:info-increase} and as $\rH_{\ristar, <\jstar} \perp \rH_{\ristar, \jstar} \mid \ristar$} \\
		&= \frac{1}{r_p} \cdot \mi{\rH_{\ristar,*}}{\rSigma \mid \ristar}
		\tag{by the chain rule (\itfacts{chain-rule})} \\
		&\leq \frac{1}{r_p} \cdot \big(\mi{\rH_{\ristar,*}}{\rSigma \mid \ristar} + \mi{\rH_{-\ristar, *}}{\rSigma \mid \rH_{\ristar, *}, \ristar} + \mi{\ristar}{\rSigma}\big)
		\tag{since mutual information is non-negative} \\
		&= \frac{1}{r_p} \cdot \mi{\set{\rH_{i,j}}_{i \in [t_p], j \in [r_p]}, \ristar}{\rSigma} 
		\tag{by the chain rule (\itfacts{chain-rule})} \\
		&\leq \frac{1}{r_p} \cdot \mi{\rG}{\rPi} = \frac{1}{r_p} \cdot \ic{\pi_p}{\GG_p}
	\end{align*}
	where the final inequality holds by the data processing inequality (\itfacts{data-processing}) since we have that $(\{\rH_{i,j}\}_{i \in [t_p], j \in [r_p]}, \ristar)$ is a (deterministic) function of $\rG$ and $\rSigma$ is a function of $\rPi$.
\end{proof}

Although having a low information cost protocol is a significant improvement, this only holds since $\sigma$ faithfully runs $\pi_p$ on inputs $G$ distributed according to $\GG_p$.
We will, however, later require guarantees for inputs distributed `almost' according to $\GG_p$. 
Hence, we require a low communication cost variant of the protocol, which is independent of the input distribution.

We design a protocol $\sigmae$ by compressing the transcript of messages communicated in $\sigma$ down to its information cost. 
To avoid ambiguity, here, for every $a \in [p]$ and $k \in [p]$, 
we use $\Sigma_{k}(a)$ to denote the message sent by player $\player{a}$ of $\sigma$ in round $k$. We also use $\rR_{\sigma}$ to denote the public randomness of $\sigma$. 
We define $\Sigmae_{k}(a)$ and $\rR_{\sigmae}$ analogously for $\sigmae$. Protocol $\sigmae$ is as follows: 

\begin{ourbox}
	\textbf{A \underline{communication efficient} implementation of $\sigma$ as a protocol $\sigmae$}: 
	\begin{itemize}
		\item For $k=1$ to $p$ rounds and $a=1$ to $p$ players in this order: 
		\begin{enumerate}
			\item Suppose the players at this point \emph{all} know $R_{\sigma}$ as well as
			\[
			\Sigma_{<k} := \Sigma_{1}(1), \Sigma_{1}(2), \ldots,\Sigma_{k-1}(p), \quad \text{and} \quad \Sigma_{k}(<a) := \Sigma_{k}(1),\ldots,\Sigma_{k}(a-1),
			\]
			i.e., the public randomness of $\sigma$ and the messages that ``should have been'' communicated by $\sigma$ in the first $k-1$ rounds plus the ones by $\player{1},\ldots,\player{a-1}$ in round $k$. 
			\item Define the distributions 
			\[
			\mathcal{A}_{k,a} := \rSigma_{k}(a) \mid \Sigma_{<k},\Sigma_{k}(<a), R_{\sigma}, G_{p-1}(a), \quad \text{and} \quad \mathcal{B}_{k,a} := \rSigma_{k}(a) \mid \Sigma_{<k},\Sigma_{k}(<a), R_{\sigma}; 
			\]
			notice that the difference is only that in $\mathcal{A}_{k,a}$ we also condition on the input of player $\player{a}$ in $G_{p-1}$, i.e., $G_{p-1}(a)$, but in $\mathcal{B}_{k,a}$, we do not. 
			\item The player $\player{a}$ plays the role of Alice in~\Cref{prop:msg-compress}, who knows both $\mathcal{A}_{k,a}$ and $\mathcal{B}_{k,a}$ and all other players play
			the role of Bob who only knows $\mathcal{B}_{k,a}$. The players use public randomness and Alice (i.e., player $\player{a}$) writes the message of~\Cref{prop:msg-compress}, denoted by $\sigmae_{k,a}$, on the blackboard
			which allows \emph{all} players to sample a \emph{single} message
			\[
			\Sigma_{k}(a) \sim \mathcal{A}_{k,a} = \rSigma_{k}(a) \mid \Sigma_{<k},\Sigma_{k}(<a), G_{p-1}(a),R_{\sigma}. 
			\] 
			This allows the players to obtain the message $\Sigma_k(a)$ and continue this for-loop. 
		\end{enumerate}
	\end{itemize}
\end{ourbox}

\begin{claim} \label{lem:prot-low-comm-expec}
	Protocol $\sigmae$ solves $\PP_{p-1}$ with probability at least $\delta$ and \underline{expected} length of messages
	$${\raoconst} \cdot \Paren{\frac{s}{r_p} + p^2}$$
	where $\raoconst \geq 1$ is an absolute constant defined in~\Cref{eq:msg-compress-final} (derived from \Cref{prop:msg-compress}). 
\end{claim}
\begin{proof}
	The fact that $\sigmae$ is {faithfully} simulating the protocol $\sigma$ follows immediately from \Cref{prop:msg-compress} as the message compression approach 
	has no error. 
	We now bound the length of messages communicated by $\sigmae$ \emph{in expectation}. 
	
	With a slight abuse of notation, we use $\card{\Sigmae}$ to denote the length of the messages communicated by $\sigmae$ (on a particular input and realization of randomness) 
	and $\card{\Sigmae_{k}(a)}$ for $k \in [p]$ and $a \in [p]$ as the length of the message $\Sigmae_{k}(a)$. Thus, by linearity of expectation, 
	\[
	\Exp\card{\Sigmae} = \sum_{k=1}^{p} \sum_{a=1}^{p} \Exp\card{\Sigmae_{k}(a)}. 
	\]
	Moreover, by~\Cref{prop:msg-compress}, for every $k \in [p]$ and $a \in [p]$, we have, 
	\begin{align*}
		\Exp\card{\Sigmae_{k}(a)} &\leq \raoconst \cdot \Paren{\Exp_{\Sigma_{<k},\Sigma_{k}(<a), G_{p-1}(a),R_{\sigma}} \kl{\mathcal{A}_{k,a}}{\mathcal{B}_{k,a}} + 1} 
		\tag{as the distributions of~\Cref{prop:msg-compress} are $\mathcal{A}_{k,a},\mathcal{B}_{k,a}$} \\
		&\leq \raoconst \cdot \Paren{\Exp_{\substack{\Sigma_{<k},\Sigma_{k}(<a), \\ G_{p-1}(a),R_{\sigma}}} \kl{\rSigma_{k}(a) \mid \Sigma_{<k},\Sigma_{k}(<a), R_{\sigma}, G_{p-1}(a)}{\rSigma_{k}(a) \mid \Sigma_{<k},\Sigma_{k}(<a), R_{\sigma}} + 1}
		\tag{by the definition of $\mathcal{A}_{k,a},\mathcal{B}_{k,a}$} \\
		&= \raoconst \cdot \Paren{\mi{\rSigma_{k}(a)}{\rG_{p-1}(a) \mid \rSigma_{<k},\rSigma_{k}(<a), \rR_{\sigma}} + 1}. \tag{by~\Cref{fact:kl-info}}
	\end{align*}
	Plugging these bounds in the expectation-term above implies 
	\begin{align*}
		\Exp\card{\Sigmae} &\leq \sum_{k=1}^{p} \sum_{a=1}^{p}\raoconst \cdot \Paren{\mi{\rSigma_{k}(a)}{\rG_{p-1}(a) \mid \rSigma_{<k},\rSigma_{k}(<a), \rR_{\sigma}} + 1} \\
		&\leq \sum_{k=1}^{p} \sum_{a=1}^{p}\raoconst \cdot \Paren{\mi{\rSigma_{k}(a)}{\rG_{p-1} \mid \rSigma_{<k},\rSigma_{k}(<a), \rR_{\sigma}} + 1} \tag{by~\itfacts{data-processing}  as $\rG_{p-1}(a)$ is a deterministic function of $\rG_{p-1}$} \\
		&=  \raoconst \cdot \Paren{\mi{\Sigma}{\rG_{p-1} \mid \rR_{\sigma}} + p^2} \tag{by the chain rule (\itfacts{chain-rule}) where $\Sigma = \set{\rSigma_{k}(a)}_{k \in [p], a \in [p]}$} \\
		&= \raoconst \cdot \Paren{\ic{\sigma}{\GG_{p-1}} + p^2} \tag{by \Cref{def:ext-info}} \\
		&\leq \raoconst \cdot \Paren{\frac{1}{r_p} \cdot \ic{\pi_p}{\GG_{p}} + p^2} \tag{by \Cref{clm:prot-low-info}} \\
		&\leq \raoconst \cdot \Paren{\frac{s}{r_p}+ p^2}
	\end{align*}
	where the last inequality holds since the communication cost of a protocol is always larger than its information cost (\Cref{prop:ic-cc}).
\end{proof}

We can now easily construct a protocol $\pi'$ that is identical to protocol $\sigmae$ except with message truncation: 
if, at any point, a player is going to send a message that results in the length of the messages communicated by the protocol to become more than the desired communication cost, then
terminate the protocol and return `fail'. We only provide the full protocol here for completeness. 

\begin{ourbox}
	\textbf{Protocol $\prot'$: the truncated version of the protocol $\sigmae$ for the given parameter $\eta > 0$ from \Cref{lem:roundelim}}
	\begin{itemize}
		\item Run the protocol $\sigmae$ as is. If at any point, a player is going to send a message that results in the communication 
		cost of the protocol to become more than 
		\[
		\frac{\raoconst}{\eta} \cdot \Paren{\frac{s}{r_p} + p^2}, 
		\]
		terminate the protocol and return `fail'. 
	\end{itemize}
\end{ourbox}

\begin{corollary} \label{cor:prot-low-comm}
	Protocol $\pi'$ solves $\PP_{p-1}$ with probability at least $\delta - \eta$ and communication cost
	$$\cc{\pi'} \leq \frac{\raoconst}{\eta} \cdot \Paren{\frac{s}{r_p} + p^2}$$
	where $\eta > 0$ is the parameter given in the statement of \Cref{lem:roundelim}. 
\end{corollary}
\begin{proof}
	By Markov's Inequality, the truncation 
	of the messages to the desired communication cost as in the statement	
	incurs a failure probability of $\eta$ in the simulation of $\sigmae$.
	Therefore, by \Cref{lem:prot-low-comm-expec}, the success probability is at least $\delta - \eta$ and, by the truncation, the communication cost is as desired.
\end{proof}

We have constructed a low communication protocol with one fewer player but with the same number of rounds. The next step reduces the number of rounds by one.

\subsubsection*{Step 2: Round Elimination}

We are now ready to use the low communication cost protocol $\pi'$ with one fewer player to construct a low communication cost protocol $\pi_{p-1}$ with one fewer player and, additionally, one fewer round. 
We accomplish this by sampling the messages communicated in the first round in order to eliminate a round.
Note that $\pi_{p-1}$ will simulate $\pi_p$ in exactly the same way as protocol $\pi'$; however, the distribution of the underlying input $G$ will necessarily be sampled differently.

\begin{ourbox}
	\textbf{A $(p-1)$-round $p$-party protocol $\pi_{p-1}$ for $\PP_{p-1}$ on input graph $G_{p-1} \sim \GG_{p-1}$}: 
	\begin{enumerate}
		\item Sample first-round messages $\rPi_{1}(1),\ldots,\rPi_{1}(p)$ of $\pi_{p}$ from $\GG_p$ and indices $\istar \in [t_p]$ and $\jstar \in [r_p]$ uniformly at random and independently using \underline{public randomness}\footnote{We emphasize
			that we are \emph{not} sampling the message $\rPi_{1}(p+1)$ in this step (as it is determined by $\istar$).}.   
		\item For every $a \in [p]$, each player $\player{a}$ sets $H_{\istar, \jstar}(a) = G_{p-1}(a)$ and samples
		\[
		G(a) \sim \GG_p \mid \rPi_{1}(1),\ldots,\rPi_{1}(p), \ristar=\istar , \rH_{\istar,\jstar}(a)=G_{p-1}(a)
		\]
		using \underline{private randomness}. Let $G=(G(1),G(2),\ldots,G(p+1))$ be the 
		input of all players where the input of the simulated player $\player{p+1}$ is fixed by $\ristar = \istar$.
		\item The players run $\pi_p$ on input $G$ \underline{from its second round} onwards, following the protocol $\prot'$. 
		\item At the end, for a uniform search sequence $J=(j_{p-1},\ldots,j_0)$, to answer the search predicate $\PP_{p-1}(G_{p-1},J)$, return the answer of $\pi_p$ for $\PP_p(G, (\jstar, j_{p-1},\ldots,j_0))$\footnote{This is possible since
			$\prot'$ generates the messages of $\pi_p$ on input $G$ implicitly.}.  
	\end{enumerate}
\end{ourbox}

As previously mentioned, both $\pi'$ and $\pi_{p-1}$ are simulating $\pi_p$ on an input $G$ in exactly the same way, but with different underlying distributions.
In protocol $\pi'$, $G$ is distributed according to $\GG_p$ and is exactly what the protocol $\pi_p$ would expect as an input. 
Let $\dist(G)$ represent this distribution of $G$.
On the other hand, in protocol $\pi_{p-1}$, $G$ is distributed differently.
Let $\dist(G \mid \pi_{p-1})$ represent this distribution.
We show that these distributions are statistically close to each other.

\begin{claim} \label{lem:tvd-diff}
	$$\tvd{\dist(G)}{\dist(G \mid \pi_{p-1})} \leq \sqrt{\frac{s}{2 \cdot r_p \cdot t_p}}.$$
\end{claim}
\begin{proof}
	Recall that we have the product distribution 
	$\GG_p = \prod_{a \in [p+1]} \GG_p(a)$.
	Furthermore, for any $a \in [p]$, we have the product distribution 
	$\GG_p(a) = \prod_{i \in [t_p], j \in [r_p]} \rH_{i,j}(a)$ and $\GG_{p}(p+1) = \ristar$.
	Let $\rPi_1(\leq p) = (\rPi_1(1), \ldots, \rPi_1(p))$.
	By the rectangle property (for multi-player protocols), we then have that 
	\begin{align}
		\paren{\GG_p \mid \rPi_1(\leq p), \ristar, \rjstar, \rH_{\ristar,\rjstar}} 
		 = \prod_{a \in [p]} \paren{\GG_p(a) \mid \rPi_1(\leq p), \ristar, \rH_{\ristar,\rjstar}(a)} \label{eq:rec-prod-dist}
	\end{align}
	where the first equality holds since $\rjstar$ is entirely independent of $\GG_p,\rPi$.
	
	With that, if $\pi_{p-1}$ were a faithful simulation of $\pi'$, we would expect $G$ to be distributed as follows:
	\begin{align*}
		\dist(G) 
		&= (\rPi_1(\leq p), \ristar, \rjstar) \times \paren{\rH_{\ristar, \rjstar} \mid \rPi_1(\leq p), \ristar, \rjstar} \times \paren{\GG_p \mid \rPi_1(\leq p), \ristar, \rjstar, \rH_{\ristar,\rjstar}} \\
		&= (\rPi_1(\leq p), \ristar, \rjstar) \times \paren{\rH_{\ristar, \rjstar} \mid \rPi_1(\leq p), \ristar, \rjstar} \times \prod_{a \in [p]} \paren{\GG_p(a) \mid \rPi_1(\leq p), \ristar, \rH_{\ristar,\rjstar}(a)}. \tag{by \Cref{eq:rec-prod-dist}}
	\end{align*}
	This, however, is {\em not} possible as the second term implies that the given input $G_{p-1}$ would need to be distributed according to $\paren{\rH_{\ristar, \rjstar} \mid \rPi_1(\leq p), \ristar, \rjstar}$ where $\rH_{\ristar,\rjstar}$ without the conditioning is distributed according to $\GG_{p-1}$. 
	Since the given input $G_{p-1} \sim \GG_{p-1}$ in protocol $\pi'$, we have that $G$ is distributed as follows:
	\begin{align*}
		\dist(G \mid \pi_{p-1}) 
		&= (\rPi_1(\leq p), \ristar, \rjstar) \times \paren{\rH_{\ristar, \rjstar}} \times \prod_{a \in [p]} \paren{\GG_p(a) \mid \rPi_1(\leq p), \ristar, \rH_{\ristar,\rjstar}(a)}.
	\end{align*}
	Therefore, the only difference between the two distributions is the second term.
	
	Using this in conjunction with \Cref{fact:tvd-chain-rule}, followed by a standard application of Pinsker's Inequality (\Cref{fact:pinskers}), we have that
	\begin{align}
		\tvd{\dist(G)}{\dist(G \mid \pi_{p-1})} 
		&\leq \Exp_{\rPi_1(\leq p), \ristar,\rjstar} \tvd{\paren{\rH_{\istar, \jstar} \mid \Pi_1(\leq p), \istar, \jstar}}{\rH_{\istar, \jstar}} \nonumber \tag{by \Cref{fact:tvd-chain-rule}} \\
		&\leq \Exp_{\rPi_1(\leq p), \ristar,\rjstar} \left[\sqrt{1/2 \cdot \kl{\rH_{\istar, \jstar} \mid \Pi_1(\leq p), \istar, \jstar}{\rH_{\istar, \jstar}}}\right] \nonumber \tag{by \Cref{fact:pinskers}} \\
		&\leq \sqrt{1/2 \cdot \Exp_{\rPi_1(\leq p), \ristar,\rjstar} \big[\kl{\rH_{\istar, \jstar} \mid \Pi_1(\leq p), \istar, \jstar}{\rH_{\istar, \jstar}}\big]} \label{eq:tvd-diff}
	\end{align}
	where the last inequality holds by Jensen's Inequality since $\sqrt{\cdot}$ is concave.
	
	Recall that the random variables $\rH_{1,1}, \ldots, \rH_{t_p,r_p}, \ristar, \rjstar$ are mutually independent of each other.
	We then have that
	\begin{align*}
		\Exp_{\rPi_1(\leq p), \ristar,\rjstar} & \big[\kl{\rH_{\istar, \jstar} \mid \Pi_1(\leq p), \istar, \jstar}{\rH_{\istar, \jstar}}\big] \\
		&= \Exp_{\rPi_1(\leq p), \ristar,\rjstar} \big[\kl{\rH_{\istar, \jstar} \mid \Pi_1(\leq p), \istar, \jstar}{\rH_{\istar, \jstar} \mid \istar, \jstar}\big] \tag{since $\ristar,\rjstar \ \bot \ \rH_{\ristar,\rjstar}$} \\
		&= \mi{\rH_{\ristar, \rjstar}}{\rPi_1(\leq p) \mid \ristar, \rjstar} \tag{by \Cref{fact:kl-info}} \\
		&\leq \mi{\rH_{\ristar, \rjstar}}{\rPi_1(\leq p) \mid \rH_{<\ristar,*}, \rH_{\ristar, <\rjstar}, \ristar, \rjstar} \tag{by \Cref{prop:info-increase} as $\rH_{<\ristar, *}, \rH_{\istar, <\jstar} \perp \rH_{\ristar, \rjstar} \mid \ristar, \rjstar$} \\
		&= \frac{1}{r_p \cdot t_p} \cdot  \sum_{\istar \in [t_p], \jstar \in [r_p]} \mi{\rH_{\istar, \jstar}}{\rPi_1(\leq p) \mid \rH_{<\istar,*}, \rH_{\istar, <\jstar}, \ristar = \istar, \rjstar = \jstar} \tag{since $\ristar$ and $\rjstar$ are uniform random variables}\\
		&= \frac{1}{r_p \cdot t_p} \cdot  \sum_{\istar \in [t_p], \jstar \in [r_p]} \mi{\rH_{\istar, \jstar}}{\rPi_1(\leq p) \mid \rH_{<\istar,*}, \rH_{\istar, <\jstar}} \tag{as event $(\ristar,\rjstar)=(\istar, \jstar) \perp$ all $\rH_{i,j}, \rPi_1(\leq p)$} \\
		&\leq \frac{1}{r_p \cdot t_p} \cdot  \mi{\set{\rH_{i,j}}_{i \in [t_p], j \in [r_p]}}{\rPi_1(\leq p)} \tag{by the chain rule (\itfacts{chain-rule})} \\
		&\leq \frac{1}{r_p \cdot t_p} \cdot  \mi{\rG}{\rPi} \tag{by \itfacts{data-processing} since $\set{\rH_{i,j}}_{i \in [t_p], j \in [r_p]}$ and $\rPi_1(\leq p)$ are functions of $\rG$ and $\rPi$, respectively} \\
		&= \frac{1}{r_p \cdot t_p} \cdot \ic{\pi_p}{\GG_p}
		\leq  \frac{1}{r_p \cdot t_p} \cdot \cc{\pi_p} = \frac{s}{r_p \cdot t_p} .
		\tag{by \Cref{def:ext-info}, then \Cref{prop:ic-cc}, and since $\cc{\pi_r} = s$}
	\end{align*}
	Combining this with \Cref{eq:tvd-diff} completes the proof.
\end{proof}

We can now argue the guarantees of $\pi_{p-1}$.

\begin{claim}
	Protocol $\pi_{p-1}$ solves $\PP_{p-1}$ with probability at least 
	$\delta - \eta - \sqrt{s/(2 \cdot r_p \cdot t_p)}$ 
	and communication cost
	$$\cc{\pi_{p-1}} \leq \frac{\raoconst}{\eta} \cdot \Paren{\frac{s}{r_p} + p^2}.$$
\end{claim}
\begin{proof}
	Recall that, given an input $G_{p-1} \sim \GG_{p-1}$, protocols $\pi'$ and $\pi_{p-1}$ both simulate $\pi_p$ on input $G$, which is a function of $G_{p-1}$, in exactly the same way to solve $\PP_{p-1}$.
	However, by sampling the first round messages,
	$G$ is distributed according to $\dist(G \mid \pi_{p-1})$ in $\pi_{p-1}$ instead of $\dist(G)$ in $\pi'$.
	Therefore, we have that
	\begin{align*}
		\Pr_{\substack{G_{p-1} \sim \GG_{p-1} \\ G \sim \dist(G \mid \pi_{p-1})}} \paren{\pi_{p-1} \text{~solves~} \PP_{p-1} } &\geq
		\Pr_{\substack{G_{p-1} \sim \GG_{p-1} \\ G \sim \dist(G)}} \paren{\pi' \text{~solves~} \PP_{p-1} } - \tvd{\dist(G)}{\dist(G \mid \pi_{p-1})} \tag{by \Cref{fact:tvd-small}} \\
		&\geq \delta - \eta - \tvd{\dist(G)}{\dist(G \mid \pi_{p-1})} \tag{by \Cref{cor:prot-low-comm}} \\
		&\geq \delta - \eta - \sqrt{\frac{s}{2 \cdot r_p \cdot t_p}}. \tag{by \Cref{lem:tvd-diff}}
	\end{align*}
	
	Since $\pi_{p-1}$ follows protocol $\pi'$ for its communication, except for the sampled first round messages, its communication cost is at most
	$$\cc{\pi_{p-1}} \leq \cc{\pi'} \leq \frac{\raoconst}{\eta} \cdot \Paren{\frac{s}{r_p} + p^2}$$
	by \Cref{cor:prot-low-comm}.

It remains to construct a deterministic protocol with the same guarantees as the public coin protocol $\pi_{p-1}$.
Since communication cost is a worst-case measure, we can fix the public random string of $\pi_{p-1}$ to its best choice by an averaging argument in order to obtain a deterministic protocol with the same guarantees.
\end{proof}
This completes the proof of \Cref{lem:roundelim}.

% !TEX root = main.tex

\providecommand{\HS}{\textnormal{\textsc{HiddenRestrictions}}}
\providecommand{\DHR}{\ensuremath{\cD_{\mathrm{HR}}}}
\providecommand{\cD}{\ensuremath{\mathcal{D}}}
\providecommand{\cF}{\ensuremath{\mathcal{F}}}
\providecommand{\cG}{\ensuremath{\mathcal{G}}}
\providecommand{\cL}{\ensuremath{\mathcal{L}}}
\providecommand{\cM}{\ensuremath{\mathcal{M}}}
\providecommand{\cS}{\ensuremath{\mathcal{S}}}
\providecommand{\cU}{\ensuremath{\mathcal{U}}}
\providecommand{\E}{\mathbb{E}}
\providecommand{\pre}{\mathrm{pre}}
\providecommand{\post}{\mathrm{post}}
\providecommand{\rPi}{\ensuremath{\rv{\Pi}}}
\providecommand{\Unif}{\ensuremath{\mathcal{U}}}

\section{Second Proof: A Direct Communication Lower Bound}\label{sec:communication}

We now begin the second proof.  This section proves the communication lower
bound directly, by induction on the levels.  In \Cref{sec:reduction}, we reduce
this communication lower bound to the streaming lower bound.

We use a one-way \((p+1)\)-party communication game, called \(\HS\), as the
main communication problem.

Let $P_1, \dots, P_{p+1}$ be the $p+1$ parties and let $t_1, r_1, t_2, r_2, \dots, t_p, r_p$ be positive integers, with \(t_i\ge 2\) for all \(i\). 
The game has the following functions: 

\begin{itemize}
    \item $f_{p+1}: \{*\} \to [t_1]$
    \item $f_{p}: [t_1] \times [r_1] \to [t_2]$
    \item $f_{p-1}: [t_1] \times [r_1] \times [t_2] \times [r_2] \to [t_3]$
    \item \dots 
    \item $f_1: [t_1] \times [r_1] \times \dots \times [t_p] \times [r_p] \to \{0,1\}$. 
\end{itemize}

For each \((y_1,\dots,y_p)\in[r_1]\times\cdots\times[r_p]\), define its
selected path by
\begin{itemize}
    \item $i_1 := f_{p+1}(*)$,
    \item $i_2(y_1) := f_p(i_1,y_1)$,
    \item $i_3(y_1,y_2) := f_{p-1}(i_1,y_1,i_2(y_1),y_2)$,
    \item \(\dots\)
    \item $i_{p+1}(y_1,\ldots,y_p) :=
    f_1(i_1,y_1,i_2(y_1),y_2,\ldots,i_p(y_1,\ldots,y_{p-1}),y_p)\in\{0,1\}$.
\end{itemize}
The \(y\)'s are the output coordinates.  The \(i\)'s are hidden
indices chosen by the players' functions along the way.
Player $P_i$ receives \(f_i\).
Each pass consists of the players speaking in order from $P_1$ to $P_{p+1}$.
For an input, let
\[
X=\bigl(i_{p+1}(y_1,\ldots,y_p)\bigr)_{(y_1,\ldots,y_p)\in [r_1]\times\cdots\times[r_p]}
\in\{0,1\}^{[r_1]\times\cdots\times[r_p]}
\]
be the output vector.  At the end of the \(p\)-th pass, the goal is for
$P_{p+1}$ to output this bit for every \((y_1,\ldots,y_p)\in
[r_1]\times\cdots\times[r_p]\), where each bit is evaluated along its selected
path \(i_1,i_2(y_1),\ldots,i_p(y_1,\ldots,y_{p-1})\).

The proof follows these selected paths.  Once the earlier hidden indices are
fixed, the next function selects a subtable of the next player's input.  The
lower bound tracks how much information the transcript can reveal about this
selected subtable.

The main result of this section is the following communication lower bound.
\begin{theorem} \label{thm:HS}
Fix \(p\) and parameters \(t_i\ge 2\), \(r_i\ge 1\). Let \(\mathcal P\) be any randomized \(p\)-pass one-way
protocol for \HS{}, and write \(b=\cc{\mathcal P}\) for its communication cost.
Suppose that, on every input, the final player outputs
\(\widehat X\in\{0,1\}^{[r_1]\times\cdots\times[r_p]}\) and
\[
    \E_{\mathrm{coins}}\bigl[d_{\rm H}(X,\widehat X)\bigr]
    \le
    \left(\frac12-\eta\right)r_1\cdots r_p
\]
for some fixed \(0<\eta<1/2\). Then, with \(t_{p+1}=2\),
\[
    b\ge
    \widetilde{\Omega}_{p,\eta}\left(
    \min_{k\in[p]} r_1\cdots r_k\,t_k
    \right).
\]
Here \(d_{\rm H}\) denotes Hamming distance.
\end{theorem}

For the proof, we use the following product distribution.

\noindent\fbox{%
\begin{minipage}{\dimexpr\linewidth-2\fboxsep-2\fboxrule-6pt\relax}
Distribution \(\DHR\):
\begin{enumerate}
    \item Sample \(f_{p+1}(*)\) uniformly from \([t_1]\).
    \item For \(\ell=p,p-1,\dots,1\), sample every value of
\[
f_\ell : [t_1]\times[r_1]\times\cdots\times[t_{p+1-\ell}]\times[r_{p+1-\ell}]
\;\to\;
\begin{cases}
[t_{p+2-\ell}], & \ell\ge 2,\\
\{0,1\}, & \ell=1,
\end{cases}
\]
independently and uniformly from its output alphabet.
\end{enumerate}
\end{minipage}%
}

\subsection{Information Bounds for Restrictions}
We first prove a few information bounds.  For a fixed level \(k\), it will be useful to think of \(F_k\)
as a table whose rows are output-coordinate prefixes, such as
\((y_1,\ldots,y_{k-1})\), and whose columns are possible values of the next
hidden index.  The selector \(G_k\), from another player's input,
chooses one column in each row.  The selected table is
\[
        F_{k,G_k}[a] = F_k[a,G_k(a)] .
\]
We will bound how much information the transcript can have about
\(F_{k,G_k}\). Information about the full table \(F_k\) is spread over the
\(t_k\) possible columns.  Information about the selector \(G_k\) is separated, and we bound it after giving the information bounds.

The first lemma breaks the KL divergence of a table into contributions from its
entries. 

\begin{lemma}
\label{lem:KL-mass-table}
Let $\cS$ be a nonempty finite set, and let $\cA,\cB$ be finite sets. Let $\cF$ be the uniform distribution on
$\cS^{\cA\times\cB}$, and let $\cF'$ be any distribution on $\cS^{\cA\times\cB}$.
Fix a total order on $\cA\times\cB$. Let  $\rv{F} \sim \cF$ and $\rv{F}' \sim \cF'$.
For each $(a,b)\in\cA\times\cB$ define
\[
t_{a,b} = 
\E\left[
  \kl{\cL(\rv{F}'[a,b]\mid \rv{F}'_{<(a,b)})}{\Unif(\cS)}
\right].
\]
Then:
\begin{enumerate}
\item 
$\displaystyle \kl{\cF'}{\cF}=\sum_{(a,b)\in\cA\times\cB} t_{a,b}.$
\item 
$\displaystyle 0\le t_{a,b}\le \log|\cS|$ for all $(a,b)\in\cA\times\cB.$
\item 
For each deterministic $g:\cA\to\cB$ and each $f\in\cS^{\cA\times\cB}$ define the restriction
$f_g\in\cS^{\cA}$ by $f_g[a]:=f[a,g(a)]$ for all $a\in\cA$.
Let $\cF_g$ denote the law of $(\rv{F})_g$
and let $\cF'_g$ denote the law of $(\rv{F}')_g$. Then
\[
\kl{\cF'_g}{\cF_g} \leq \sum_{a\in\cA} t_{a,g(a)}.
\]
\end{enumerate}
\end{lemma}
\begin{proof}
    \textbf{(1)} This follows from the KL chain rule applied to
$\cF'$ and $\cF$, together with the fact that under $\cF$ the coordinates are
i.i.d.\ uniform on $\cS$, so for every $(a,b)$ and every prefix value $y$,
$\cL(\rv F[a,b]\mid \rv F_{<(a,b)}=y)=\Unif(\cS).$

\textbf{(2)} For any distribution $\cP$ on $\cS$,
\[
\kl{\cP}{\Unif(\cS)}
=
\sum_{s\in\cS}\cP(s)\log\frac{\cP(s)}{1/|\cS|}
=
\log|\cS|-\mathrm{H}(\cP)
\le
\log|\cS|.
\]
Also \(\kl{\cP}{\Unif(\cS)}\geq 0\). Applying this to the definition of
\(t_{a,b}\) proves (2).

\textbf{(3)} Fix a deterministic function $g:\cA\to\cB$.
Let $I_g := \{(a,g(a)) : a\in\cA\} \subseteq \cA\times\cB,$
and order the elements of $I_g$ according to the fixed total order on
$\cA\times\cB$, writing
\[
i_1 \prec i_2 \prec \cdots \prec i_{|\cA|},
\qquad i_k=(a_k,g(a_k)).
\]

Under $\cF$, the random vector $(\rv F)_g$ is uniform on $\cS^{\cA}$, so
$\cF_g$ is the uniform distribution on $\cS^{\cA}$. Applying the KL chain rule gives 
\[
\kl{\cF'_g}{\cF_g}
=
\sum_{k=1}^{|\cA|}
\E\!\left[
  \kl{
    \cL(\rv F'[i_k]\mid \rv F'[i_{<k}])
  }{
    \Unif(\cS)
  }
\right],
\]
where $\rv F'[i_{<k}]:=(\rv F'[i_1],\dots,\rv F'[i_{k-1}])$.

For each $k$, the random vector $\rv F'[i_{<k}]$ is a
deterministic function of the full prefix
$\rv F'_{< i_k}$ (with respect to the fixed total order on
$\cA\times\cB$).  
Applying the monotonicity of KL divergence under additional conditioning, we have 
\[
\E\left[
  \kl{
    \cL(\rv F'[i_k]\mid \rv F'[i_{<k}])
  }{
    \Unif(\cS)
  }
\right]
\leq 
\E\left[
  \kl{
    \cL(\rv F'[i_k]\mid \rv F'_{< i_k})
  }{
    \Unif(\cS)
  }
\right].
\]

Summing over $k=1,\dots,|\cA|$ proves (3).
\end{proof}

We also need the following which says that if a distribution is close to
uniform, then it cannot put much mass on points whose
probability is much more than uniform.

\begin{lemma} \label{lem:Bfromuniform}
    Let $\cB$ be a nonempty finite set. Let $\cU_\cB = \Unif(\cB)$ and let $\cD$ be any distribution over $\cB$. Then
    \[
    \sum_{b \in \cB} \max \left\{ 0, \cD(b) - \frac{2}{|\cB|} \right\} \leq 2 \cdot \kl{\cD}{\cU_{\cB}}. 
    \]
\end{lemma}
\begin{proof}
It is enough to prove the following inequality for each \(b\in\cB\):
\begin{equation} \label{eq:heavy-element-pointwise}
    \max\left\{0, \cD(b) - \frac{2}{|\cB|}\right\} \le 2\left(\cD(b)\cdot  \ln( \cD(b) |\cB|) + \frac{1}{|\cB|} - \cD(b) \right).
\end{equation}
After summing over all \(b\in\cB\), the right-hand side becomes
\begin{align*}
2 \sum_{b \in \cB}\left(\cD(b)\cdot  \ln( \cD(b) |\cB|) + \frac{1}{|\cB|} - \cD(b) \right) &=  2 \sum_{b \in \cB}\left(\cD(b)\cdot  \ln( \cD(b) |\cB|) \right) \\ 
& \leq 2 \cdot \kl{\cD}{\cU_{\cB}}. 
\end{align*}
Fix \(b\in\cB\).

First consider the case where $\cD(b) \leq \frac{2}{|\cB|}$. Over the interval
\(\cD(b) \in [0,1/|\cB|]\), the right-hand side of
(\ref{eq:heavy-element-pointwise}) strictly
decreases. Over the interval
\(\cD(b) \in [1/|\cB|,2/|\cB|]\), it strictly increases. Thus the minimum is
attained at \(\cD(b)=1/|\cB|\), where the right-hand side is \(0\).

Now consider the case where  $\cD(b) > \frac{2}{|\cB|}$. Let
\(\cD(b)=\alpha/|\cB|\) for \(\alpha\ge2\). Then
(\ref{eq:heavy-element-pointwise}) reduces to
\[
\frac{\alpha - 2}{|\cB|} \leq 2 \left (\frac{\alpha \ln \alpha + 1 -\alpha}{|\cB|} \right) \enspace \leftrightarrow \enspace 3\alpha -4 \leq 2\alpha \ln \alpha 
\]
which is true for $\alpha \geq 2$: it holds at \(\alpha=2\), and the derivative
of the right-hand side minus the left-hand side is \(2\ln\alpha - 1>0\) on this
range.
\end{proof}

We apply this bound row by row.

\begin{lemma}
\label{lem:selector-averaging}
Let $\cA$ be a finite set, and let $\cB$ be a nonempty finite set. Let $\cG$ be the uniform distribution on
$\cB^{\cA}$, and let $\cG'$ be any distribution on $\cB^{\cA}$.
Let $\{t_{a,b}\}_{(a,b)\in\cA\times\cB}$ be any array satisfying
$0\le t_{a,b}\le K$ for all $(a,b)$. 
Then
\[
\E_{\rv{g}\sim\cG'}\left[\sum_{a\in\cA} t_{a,\rv{g}(a)}\right] \leq 
\frac{2}{|\cB|}\sum_{(a,b)\in\cA\times\cB} t_{a,b} + 
2K\cdot \kl{\cG'}{\cG}.
\]
\end{lemma}
\begin{proof}
Let $U_{\cB} = \Unif(\cB)$. For an $a \in \cA$, define the marginal distribution $P_a$ on $\cB$ by 
\[
P_a(b)\ :=\ \Pr_{\rv g\sim\cG'}[\rv g(a)=b],
\qquad b\in\cB.
\]
Then
\begin{align*}
\E_{\rv g\sim\cG'} \left[\sum_a t_{a,\rv g(a)}\right]
&=
\sum_{a, b} P_a(b)\,t_{a,b}
\\
& \leq \frac{2}{|\cB|} \sum_{a,b}t_{a,b} + \sum_{a,b} \max \Big\{0,P_a(b) - \frac{2}{|\cB|}\Big\} \cdot t_{a,b} \\
& \leq \frac{2}{|\cB|} \sum_{a,b}t_{a,b} + K \cdot \sum_{a,b} \max \Big\{0,P_a(b) - \frac{2}{|\cB|}\Big\}. \tag{since $t_{a,b} \leq K$}
\end{align*}

We bound the last sum by \(2\kl{\cG'}{\cG}\).  By
\Cref{lem:Bfromuniform},
\[
\sum_{a,b} \max \left\{ 0, P_a(b) - \frac{2}{|\cB|} \right\} \leq 2\sum_a \kl{P_a}{U_{\cB}}.
\]

To bound \(\sum_{a\in\cA}\kl{P_a}{U_{\cB}}\), let
\(\cM\) denote the distribution obtained by sampling $\rv{g}(a)\sim P_a$
independently for each $a\in\cA$.
Then for all $g\in\cB^{\cA}$,
\[
\log\frac{\cG'(g)}{\cG(g)}
=
\log\frac{\cG'(g)}{\cM(g)}+\log\frac{\cM(g)}{\cG(g)}.
\]
Taking expectation over $\rv{g}\sim\cG'$ gives
\[
\kl{\cG'}{\cG}
=
\kl{\cG'}{\cM}+\E_{\rv{g}\sim\cG'}\left[\log\frac{\cM(\rv{g})}{\cG(\rv{g})}\right].
\]
Moreover,
\[
\log\frac{\cM(g)}{\cG(g)}
=\sum_{a\in\cA}\log\frac{P_a(g(a))}{U_{\cB}(g(a))},
\]
so its expectation under $\cG'$ equals
\[
\sum_{a\in\cA}\E_{b\sim P_a}\!\left[\log\frac{P_a(b)}{U_{\cB}(b)}\right]
=
\sum_{a\in\cA}\kl{P_a}{U_{\cB}}, 
\]
which gives the result. 
\end{proof}

Finally, we break up the KL divergence of the true uniform distribution from an arbitrary probability distribution into its contributions. This will then allow us to bound these contributions separately and conclude our proof. 
\begin{lemma}
\label{lem:core_lemma}
Let $\cS,\cB$ be nonempty finite sets and let $\cA$ be a finite set.  Let $\cF$ be the uniform distribution on $\cS^{\cA\times\cB}$ and let $\cG$ be the uniform distribution on $\cB^{\cA}$. Let $\cF'$ be any probability distribution on 
$\cS^{\cA\times\cB}$ and let
$\cG'$ be any probability distribution on $\cB^{\cA}$.
For $f\in\cS^{\cA\times\cB}$ and $g\in\cB^{\cA}$, define the
restriction $f_g\in\cS^{\cA}$ by
\[
  f_g[a] \;:=\; f[a,g(a)] \qquad\text{for all }a\in\cA.
\]
Let $\cF_g$ and $\cF'_g$ denote the induced laws of $(\rv F)_g$ when
$\rv{F}\sim\cF$ and $\rv{F}\sim\cF'$, respectively. Then
\[
  \E_{\rv{g}\sim\cG'}[\kl{\cF'_\rv{g}}{\cF_\rv{g}}]\leq
  \frac{2}{|\cB|}\kl{\cF'}{\cF} + 
  2\log|\cS|\cdot \kl{\cG'}{\cG}.
\]
\end{lemma}
\begin{proof}
Fix an arbitrary total order on \(\cA\times\cB\), and define the array $\{t_{a,b}\}_{a\in\cA,b\in\cB}$ from $\cF'$ as in
Lemma~\ref{lem:KL-mass-table}.  
By Lemma~\ref{lem:KL-mass-table}, for every deterministic selector
$g:\cA\to\cB$,
\begin{equation} \label{eq:Fg}
\kl{\cF'_g}{\cF_g} \leq 
\sum_{a\in\cA} t_{a,g(a)}.
\end{equation}

Averaging over $\rv{g}\sim\cG'$,
\[
\E_{\rv{g}\sim\cG'}\big[\kl{\cF'_{\rv{g}}}{\cF_{\rv{g}}}\big] \leq
\E_{\rv{g}\sim\cG'}\!\left[\sum_{a\in\cA} t_{a,\rv{g}(a)}\right].
\]

Lemma~\ref{lem:KL-mass-table} also gives
\(0 \leq t_{a,b} \leq \log|\cS|\).  Applying
Lemma~\ref{lem:selector-averaging} with \(K=\log|\cS|\) gives
\[
\E_{\rv{g}\sim\cG'}\!\left[\sum_{a\in\cA} t_{a,\rv{g}(a)}\right]
\;\le\;
\frac{2}{|\cB|}\sum_{(a,b)\in\cA \times \cB}t_{a,b}
\;+\;
2\log|\cS|\cdot \kl{\cG'}{\cG}.
\]

Finally, Lemma~\ref{lem:KL-mass-table} gives
\[
\sum_{(a,b)\in\cA \times \cB}t_{a,b}
\;=\;
\kl{\cF'}{\cF}, 
\]
which completes the proof.
\end{proof}

\subsection{Inductive Proof of Theorem~\ref{thm:HS}}

We prove the corresponding deterministic lower bound under \(\DHR\).
If a randomized protocol satisfies the worst-case Hamming guarantee in
Theorem~\ref{thm:HS}, then its expected Hamming loss, averaged over
\(\DHR\) and over all coins, is at most
\((1/2-\eta)r_1\cdots r_p\).  Hence some fixed choice of the coins gives a
deterministic protocol with the same communication bound and expected Hamming
loss at most \((1/2-\eta)r_1\cdots r_p\) under \(\DHR\).

If the protocol had \(p+1\) passes, the players could learn the hidden indices
one level at a time.  In the first pass, player
$P_{p+1}$ reveals index $i_1 \in [t_1]$; in the second
pass, player $P_p$ reveals the indices \(i_2(y_1)\) for each
\(y_1\in[r_1]\), and so on.  The proof measures the transcript at the points
where these hidden indices would become available.

\begin{definition}[Critical player and critical point]
For pass $k\in\{1,\dots,p\}$, the \emph{critical player} of pass $k$
is player
\[
P_{c(k)} \;:=\; P_{p+2-k}.
\]
The \emph{critical point} of pass $k$, denoted $C_k$, is the moment
immediately after the critical player $P_{c(k)}$ has sent their message in
pass $k$.
\end{definition}

\begin{assumption} \label{assump:reveal-input}
    In each pass $k$, we assume player $P_{c(k)}$ communicates their entire input. We do not charge this as communication cost. 
\end{assumption}

This assumption can only make the protocol stronger.  We refer to this
additional communication as ``free reveals". 

We also use the following about one-way protocols under a product
input distribution. Informally it says that once the protocol randomness is fixed, conditioning on a
transcript prefix and the free reveals only restricts each player's input
separately.  Thus the players' inputs remain independent after conditioning.

\begin{claim}\label{claim:rectangle}
Fix a deterministic protocol. Let \(X_i\) be the original input of player
\(P_i\), and let \(U\) be the transcript prefix, together with all free reveals
made so far, at any fixed point in the protocol. Then, for every
\(u\in\mathrm{supp}(U)\), there exist sets \(R_1(u),\ldots,R_{p+1}(u)\) such
that, as subsets of the full input space,
\[
\{U=u\}=\prod_{i=1}^{p+1} R_i(u).
\]
Consequently, under the product distribution \(\DHR\), the players' inputs
remain mutually independent after conditioning on \(U=u\).

In particular, fix \(i\neq j\). Let \(Z_i\) be determined by
\(X_i\) and \(u\), let \(W_j\) be determined by \(X_j\) and \(u\), and let
\(E_j\) be an event determined by \(X_j\) and \(u\), with
\(\Pr[E_j\mid U=u]>0\). Then, for all sets \(B,C\) of possible values,
\[
\Pr[Z_i\in B,\ W_j\in C\mid U=u,E_j]
=
\Pr[Z_i\in B\mid U=u]\Pr[W_j\in C\mid U=u,E_j].
\]
In particular,
\[
\mathcal L(Z_i\mid U=u,E_j)=\mathcal L(Z_i\mid U=u).
\]
\end{claim}

\begin{proof}
We reveal the fixed prefix one message or free reveal at a time. Initially the
consistent input set is the full product space. Suppose that, after some
partial prefix \(v\), the consistent input set is \(\prod_i R_i(v)\). Suppose
the following item in the prefix is produced by player \(P_i\). Once \(v\) is
fixed, determinism says that this item is a function only of \(X_i\).
Conditioning on its value therefore only restricts the \(i\)-th coordinate
\(R_i(v)\), and leaves every other coordinate unchanged. The same argument
applies to a free reveal, since it is also produced by one player as a
deterministic function of that player's input. Thus the consistent set remains a
product rectangle after each step, and hence \(\{U=u\}\) is a rectangle.

Conditioning a product distribution on a rectangle gives a product of the
conditional marginals, so the players' inputs remain independent given \(U=u\).
The displayed consequence follows because conditioning further on \(E_j\)
changes only the conditional marginal of \(X_j\).
\end{proof}

Below, \(b\) denotes the total communication of the protocol excluding the free reveals from Assumption~\ref{assump:reveal-input}.
By the reduction to deterministic protocols under \(\DHR\) described above, it
suffices to consider a deterministic protocol whose expected Hamming error under
\(\DHR\) is at most \((1/2-\eta)r_1\cdots r_p\).

\paragraph{Invariant.}
Fix any deterministic $p$-pass protocol $\cP$ for \HS{} with total charged
communication $b$ bits, under the product distribution $\DHR$. 
Let $\rPi$ be the random transcript produced by running $\cP$ on an input drawn from $\DHR$. 
For each $k\in\{1,\dots,p\}$ let $\rPi_k^{\pre}$ denote the transcript before the critical player
$P_{c(k)}$ speaks in pass $k$ and $\rPi_k^{\post}$ denote the transcript immediately after critical player $P_{c(k)}$ has spoken and the free reveal from Assumption~\ref{assump:reveal-input} has been made. 

Let $P_{u(k)}:= P_{p + 1-k}$ denote the player who speaks immediately before $P_{c(k)}$. 
Let $\rv F_k$ be the random table obtained from the input of $P_{u(k)}$ after
fixing all coordinates determined by indices
$i_1, i_2(\cdot), \dots, i_{k-1}(\cdot)$ revealed by the critical players in the
first $k-1$ passes, and let $\cF_k$ be its law. Recall that by
\Cref{assump:reveal-input}, the full inputs of the critical players in the
first $k-1$ passes have been freely revealed.
Similarly, let $\rv G_k$ be the random selector obtained from the input of
$P_{c(k)}$ after fixing the same earlier restrictions, and let $\cG_k$ be its
law. 

For $k=1$, take $\cA_1:=\{\ast\}$. For $k\in\{1,\dots,p\}$, define
\[
\cA_k := [r_1]\times\cdots\times[r_{k-1}],
\qquad
\cB_k := [t_k],
\qquad
\cS_k :=
\begin{cases}
[t_{k+1}]^{[r_k]}, & k<p,\\
\{0,1\}^{[r_p]}, & k=p.
\end{cases}
\]
Under $\DHR$, the random variable $\rv F_k$ is uniform on
$\cS_k^{\cA_k\times\cB_k}$ and the random variable $\rv G_k$ is uniform on
$\cB_k^{\cA_k}$. Here
\(\rv F_k\) is the table held by the player immediately before the critical
player.  A row \(a\in\cA_k\) is an output-coordinate prefix, a column \(b\in\cB_k\) is a
possible next hidden index, and \(\rv G_k\) chooses the actual column in each
row.

For $f\in\cS_k^{\cA_k\times\cB_k}$ and $g\in\cB_k^{\cA_k}$, define the
restriction $f_g\in\cS_k^{\cA_k}$ by
\[
    f_g[a] := f[a,g(a)]
    \qquad\text{for all }a\in\cA_k.
\]
For deterministic \(g\), write \(\rv F_{k,g}:=(\rv F_k)_g\).  Accordingly,
define the random restricted table
\[
    \rv F_{k,\rv G_k}\in\cS_k^{\cA_k}
    \qquad\text{by}\qquad
    \rv F_{k,\rv G_k}[a] := \rv F_k[a,\rv G_k(a)]
    \quad\text{for all }a\in\cA_k.
\]

The invariant is the following.

\begin{lemma} \label{lem:invariant}
    Define $\Psi_k = \II(\rv F_{k, \rv G_k};\rPi_k^{\post}).$
    For all $k\in\{1,\dots,p\}$,
\[
\Psi_k \;\le\; \Phi_k,
\qquad\text{where}\qquad
\Phi_1:=\frac{2b}{t_1},
\quad
\Phi_k:= 3r_k\Phi_{k-1}\log(t_{k+1})+\frac{2b}{t_k}\ \ \ (k\ge2).
\]
Here, we use the convention that $t_{p+1} = 2$. 
\end{lemma}

We first show how the invariant implies the theorem, and then prove the
invariant.

\begin{proof}[Concluding Theorem~\ref{thm:HS}]
    Let
\[
X =\bigl(i_{p+1}(y_1,\dots,y_p)\bigr)_{(y_1,\dots,y_p)\in[r_1]\times\cdots\times[r_p]}
\in\{0,1\}^{r_1\cdots r_p}
\]
be the output vector, where each bit is evaluated along its selected hidden
path.  Under \(\DHR\), the hidden paths are determined by
\(f_{p+1},\ldots,f_2\), independently of \(f_1\).  Once these
paths are fixed, the \(r_1\cdots r_p\) output bits are independent bits
from \(f_1\).  Hence \(X\) is uniform on \(\{0,1\}^{r_1\cdots r_p}\), so
\[
\HH(X) = r_1\cdots r_p.
\]

Let \(\widehat X\) be the output of \(P_{p+1}\) produced from \(\rPi\).  Recall that we have
\begin{equation}
\label{eq:theorem2-hamming-error}
    \E_{\DHR}\bigl[d_{\rm H}(\widehat X,X)\bigr]
    \le
    \left(\frac12-\eta\right)r_1\cdots r_p
\end{equation}
for a fixed constant \(0<\eta<1/2\). Let \(m:=r_1\cdots r_p\).  After ordering
the \(m\) output coordinates arbitrarily, write
\(E_i:=X_i\oplus \widehat X_i\).  Since \(\widehat X\) is a function of \(\rPi\),
\[
\II(X;\rPi)\ge \II(X;\widehat X)=\HH(X)-\HH(X\mid \widehat X).
\]
Moreover,
\[
\HH(X\mid \widehat X)=\HH(E\mid \widehat X)\le \HH(E)
\le \sum_{i=1}^m \HH(E_i)
\le m\cdot h\!\left(\frac12-\eta\right),
\]
where the last step uses concavity of binary entropy and the expected Hamming
error bound \eqref{eq:theorem2-hamming-error}.  Hence
\[
    \II(X;\rPi)
    \ge
    \left(1-h\!\left(\frac12-\eta\right)\right)m
    =
    \Omega_\eta(r_1\cdots r_p).
\]

We compare the full transcript \(\rPi\) to the transcript at the last
critical point. Write
\[
\rPi=(\rPi_p^{\post},W),
\]
where \(W\) is the suffix of messages sent after the point represented by
\(\rPi_p^{\post}\).  At this point, player \(P_1\) will not speak again.
Moreover, by Assumption~\ref{assump:reveal-input}, the full inputs of
\(P_2,\ldots,P_{p+1}\) are contained in \(\rPi_p^{\post}\).  Since the protocol
is deterministic, the suffix \(W\) is therefore a deterministic function of
\(\rPi_p^{\post}\). Thus
\[
X \to \rPi_p^{\post} \to \rPi
\]
is a Markov chain. By \itfacts{data-processing},
\[
\II(X;\rPi)\le \II(X;\rPi_p^{\post}).
\]

At the last step, $\rv F_{p,\rv G_p}$ is the full output vector $X$:
\[
    \rv F_{p,\rv G_p}\in
    \bigl(\{0,1\}^{[r_p]}\bigr)^{\cA_p}
    =
    \bigl(\{0,1\}^{[r_p]}\bigr)^{[r_1]\times\cdots\times[r_{p-1}]},
\]
which is the same as a bit vector indexed by
$[r_1]\times\cdots\times[r_p]$. Hence
\[
    \II(X;\rPi_p^{\post})
    =
    \II(\rv F_{p,\rv G_p};\rPi_p^{\post})
    =
    \Psi_p.
\]
By \Cref{lem:invariant}, $\Psi_p\le \Phi_p$. Combining the previous estimates,
\[
    \Omega_\eta(r_1\cdots r_p)
    \le
    \II(X;\rPi)
    \le
    \II(X;\rPi_p^{\post})
    =
    \Psi_p
    \le
    \Phi_p.
\]

Unwinding $\Phi_k = 3r_k\Phi_{k-1}\log(t_{k+1}) + 2b/t_k$ with $\Phi_1=2b/t_1$ gives
\[
\Phi_p
=
2b\left(
\frac{\prod_{j=2}^p 3r_j\log(t_{j+1})}{t_1}
+\frac{\prod_{j=3}^p 3r_j\log(t_{j+1})}{t_2}
+\cdots
+\frac{1}{t_p}
\right)
\]
where empty products are $1$.  It follows that
\[
r_1\cdots r_p
\leq
O_{p,\eta}(1)\cdot
b\left(
\frac{\prod_{j=2}^p r_j\log(t_{j+1})}{t_1}
+\frac{\prod_{j=3}^p r_j\log(t_{j+1})}{t_2}
+\cdots
+\frac{1}{t_p}
\right).
\]
Equivalently, there exists some $k\in\{1,\dots,p\}$ such that
\[
r_1\cdots r_p
\leq
O_{p,\eta}(1)\cdot b\cdot \frac{\prod_{j=k+1}^p r_j \log(t_{j+1})}{t_k},
\]
which rearranges to
\[
b \ge \Omega_{p,\eta}\left(\frac{r_1\cdots r_k\,t_k}{\prod_{j = k+1}^p\log(t_{j+1})}\right).
\]
Taking the minimum over $k$ gives
\[
b \ge
\Omega_{p,\eta}\left(
\min_{k\in[p]}
\frac{r_1\cdots r_k\,t_k}
{\prod_{j=k+1}^p \max\{1,\log(t_{j+1})\}}
\right),
\]
where \(t_{p+1}=2\).  Suppressing logarithmic factors gives
\[
b \ge \widetilde{\Omega}_{p,\eta}\left(\min_{k\in[p]} r_1\cdots r_k\,t_k\right).
\]
\end{proof}

We now prove the invariant.  Before applying Lemma~\ref{lem:core_lemma}, we
need to connect it to
\(\Psi_k=\II(\rv F_{k,\rv G_k};\rPi_k^{\post})\).  The quantity \(\Psi_k\) is
measured after the critical player's message and free reveal, while
Lemma~\ref{lem:core_lemma} works with the conditional laws just before that
point.  The next claim uses Claim~\ref{claim:rectangle} to say that, once
\(\rPi_k^{\pre}\) is fixed, the non-selector part of the critical player's input
does not give extra information about the selected table.  Thus the reveal
matters only through \(\rv G_k\), and we get the KL expression needed below.
\begin{claim}\label{cl:postmixing}
    Fix \(k \in \{1, \dots, p\}\). For each transcript value
\(\pi\) in the support of \(\rPi_k^{\pre}\), let
\[
\cF'_{k,\pi} := \cL(\rv F_k\mid \rPi_k^{\pre}=\pi),
\qquad
\cG'_{k,\pi} := \cL(\rv G_k\mid \rPi_k^{\pre}=\pi).
\]
For each \(g\in \cB_k^{\cA_k}\), let \(\cF_{k,g}\) be the law of
\(\rv F_{k,g}\) under \(\cF_k\), and let \(\cF'_{k,\pi,g}\) be the law of
\((\rv F)_g\) for \(\rv F\sim\cF'_{k,\pi}\). Then
\[
\II(\rv F_{k,\rv G_k};\rPi_k^{\post})
\le
\E_{\pi\sim \rPi_k^{\pre}}
\E_{g\sim \cG'_{k,\pi}}
\left[
\kl{\cF'_{k,\pi,g}}{\cF_{k,g}}
\right].
\]
\end{claim}
\begin{proof}
Let \(T:=\rPi_k^{\pre}\), let \(Z:=\rv F_{k,\rv G_k}\), and let \(R_k\) be the
part of the critical player's input remaining after \(\rv G_k\) is specified.
Since the protocol is deterministic and the critical player's input is freely
revealed at \(C_k\), \(\rPi_k^{\post}\) is a deterministic function of
\((T,\rv G_k,R_k)\).  Hence, by data processing and the chain rule,
\[
\II(Z;\rPi_k^{\post})
\le \II(Z;T,\rv G_k,R_k)
=
\II(Z;T,\rv G_k)+\II(Z;R_k\mid T,\rv G_k).
\]

Fix \(T=\pi\) and \(\rv G_k=g\).  The variable \(Z=\rv F_{k,g}\) is determined
by \(P_{u(k)}\)'s input, while \(R_k\) is determined by the remaining input of
the critical player.  Here \(g\) is fixed, and the event \(\{\rv G_k=g\}\)
depends only on the critical player's input.  Applying
Claim~\ref{claim:rectangle} with \(E_j=\{\rv G_k=g\}\) gives the needed
independence, so \(\II(Z;R_k\mid T,\rv G_k)=0\).  Also, before
conditioning on \(T\), the selected table \(Z\) is independent of \(\rv G_k\).
Therefore
\[
\II(Z;\rPi_k^{\post})
\le \II(Z;T,\rv G_k)
=\II(Z;T\mid \rv G_k).
\]
Writing the last term as expected KL divergence,
\[
\II(Z;T\mid \rv G_k)
=
\sum_{g,\pi}
\Pr[\rv G_k=g,T=\pi]\,
D\!\left(
\cL(Z\mid T=\pi,\rv G_k=g)
\middle\|
\cL(Z\mid \rv G_k=g)
\right).
\]
For fixed \(g\), \(\cL(Z\mid \rv G_k=g)=\cF_{k,g}\).  By
Claim~\ref{claim:rectangle}, \(\rv F_k\) and \(\rv G_k\) are independent
conditioned on \(T=\pi\), so
\(\cL(Z\mid T=\pi,\rv G_k=g)=\cF'_{k,\pi,g}\).  Since
\(\cG'_{k,\pi}=\cL(\rv G_k\mid T=\pi)\), the preceding display becomes
\[
\II(\rv F_{k,\rv G_k};\rPi_k^{\post})
\le
\E_{\pi\sim \rPi_k^{\pre}}
\E_{g\sim \cG'_{k,\pi}}
\left[
\kl{\cF'_{k,\pi,g}}{\cF_{k,g}}
\right],
\]
as claimed.
\end{proof}

\paragraph{Base case: $C_1$.}
We use the $k=1$ instantiation
\[
\cA_1=\{\ast\},\qquad
\cB_1=[t_1],\qquad
\cS_1=
\begin{cases}
[t_2]^{[r_1]}, & p>1,\\
\{0,1\}^{[r_1]}, & p=1.
\end{cases}
\]
Let $\cF_1$ be the uniform distribution on $\cS_1^{\cA_1\times \cB_1}$ and let
$\cG_1$ be the uniform distribution on $\cB_1^{\cA_1}$. For $f \in \cS_1^{\cA_1\times \cB_1}$ and $g \in \cB_1^{\cA_1}$, recall the definition of restriction $f_g \in \cS_1^{\cA_1}$ by 
\[
f_g[\ast] := f[\ast, g(\ast)]. 
\]
For each transcript value $\pi_1$ in the support of random variable $\rPi_1^{\pre}$, let 
\[
\cF'_{1,\pi_1} := \cL(\rv F_1\mid \rPi_1^\pre=\pi_1),
\qquad
\cG'_{1,\pi_1} := \cL(\rv G_1\mid \rPi_1^\pre=\pi_1)
\]
where the laws are taken under $\DHR$ and the fixed deterministic protocol. For each fixed \(\pi_1\) and each deterministic \(g\), let
\(\cF_{1,g}\) be the law of \(\rv F_{1,g}\), and let
\(\cF'_{1,\pi_1,g}\) be the law of \((\rv F)_g\) for
\(\rv F\sim\cF'_{1,\pi_1}\).
The second KL term below is zero. Before \(C_1\), the critical player
\(P_{c(1)}=P_{p+1}\), whose input determines \(G_1\), has not spoken and no
free reveal has involved this player. The rectangle argument of
Claim~\ref{claim:rectangle} therefore leaves this player's input unrestricted
when we condition on \(\rPi_1^{\pre}=\pi_1\). Hence conditioning on
\(\rPi_1^{\pre}=\pi_1\) does not change the law of \(G_1\), so
\[
\cG'_{1,\pi_1}=\cG_1
\]
and so
\[
\kl{\cG'_{1,\pi_1}}{\cG_1}=0.
\]
Now apply Lemma~\ref{lem:core_lemma} with $(\cS,\cA,\cB)=(\cS_1,\cA_1,\cB_1)$,
$\cF=\cF_1$, $\cG=\cG_1$, $\cF'=\cF'_{1,\pi_1}$, and $\cG'=\cG'_{1,\pi_1}$.
For each fixed $\pi_1$ this gives
\begin{align*}\E_{\rv g\sim \cG'_{1,\pi_1}}\left[
  \kl{\cF'_{1, \pi_1, \rv g}}{\cF_{1,\rv g}}
\right]
&\leq
\frac{2}{t_1}\,\kl{\cF'_{1,\pi_1}}{\cF_1}
+
2\log|\cS_1|\cdot \kl{\cG'_{1,\pi_1}}{\cG_1} \\
&=
\frac{2}{t_1}\,\kl{\cF'_{1,\pi_1}}{\cF_1}.
\end{align*}
Taking expectation over $\rPi_1^\pre$ gives
\begin{equation} \label{eq:expectationpi}
\E_{\pi_1\sim\rPi_1^\pre}\left[
  \E_{\rv{g}\sim \cG'_{1,\pi_1}}\left[
    \kl{\cF'_{1, \pi_1, \rv g}}{\cF_{1,\rv g}}
  \right]
\right]
\leq
\frac{2}{t_1}\,
\E_{\pi_1\sim\rPi_1^\pre}\left[
  \kl{\cF'_{1,\pi_1}}{\cF_1}
\right].
\end{equation}
Therefore,
\begin{align*}
    \Psi_1 & = \II(\rv{F}_{1,\rv{G}_1}; \rPi_1^\post)  
    \\ &\leq \E_{\pi_1\sim\rPi_1^\pre}
        \E_{\rv g\sim\cG'_{1,\pi_1}}\left[
        \kl{\cF'_{1, \pi_1, \rv g}}{\cF_{1,\rv g}} \right] \tag{using \Cref{cl:postmixing}}
    \\ & \leq \frac{2}{t_1}\,
\E_{\pi_1\sim\rPi_1^\pre}\left[
  \kl{\cF'_{1,\pi_1}}{\cF_1}
\right] \tag{using (\ref{eq:expectationpi})}
	\\ &= \frac{2}{t_1}\II(\rv F_1;\rPi_1^\pre)
	\\ &\leq \frac{2}{t_1} \HH(\rPi_1^\pre)
	\\& \leq \frac{2b}{t_1} 
		. \tag{a $b$-bit transcript has entropy at most $b$}
	\end{align*}

\paragraph{Induction step.}

Assume Lemma~\ref{lem:invariant} holds after critical point \(C_{k-1}\), so
\(\Psi_{k-1}\le \Phi_{k-1}\).  We prove \(\Psi_k\le \Phi_k\).

As in the base case, for each transcript value $\pi_k$ in the support of random variable $\rPi_k^{\pre}$, let
\[
\cF'_{k,\pi_k} := \cL(\rv F_k\mid \rPi_k^\pre=\pi_k),
\qquad
\cG'_{k,\pi_k} := \cL(\rv G_k\mid \rPi_k^\pre=\pi_k)
\]
where the laws are taken under $\DHR$ and the fixed deterministic protocol. For each fixed \(\pi_k\) and deterministic \(g\), let \(\cF_{k,g}\) be the law
of \(\rv F_{k,g}\), and let
\(\cF'_{k,\pi_k,g}\) be the law of \((\rv F)_g\) for
\(\rv F\sim\cF'_{k,\pi_k}\).

We apply \Cref{lem:core_lemma} with $(\cS, \cA, \cB) = (\cS_k, \cA_k, \cB_k)$, $\cF = \cF_k$, $\cG = \cG_k$, $\cF' = \cF'_{k, \pi_k}$, and $\cG' = \cG'_{k, \pi_k}$. For each fixed $\pi_k$ this gives
\begin{align*}\E_{\rv g\sim \cG'_{k,\pi_k}}\left[
  \kl{\cF'_{k, \pi_k, \rv g}}{\cF_{k,\rv g}}
\right]
&\leq
\frac{2}{t_k}\,\kl{\cF'_{k,\pi_k}}{\cF_k}
+
2\log|\cS_k|\cdot \kl{\cG'_{k,\pi_k}}{\cG_k}.
\end{align*}

Taking expectation over $\rPi_k^\pre$ gives
\begin{align}
\E_{\pi_k\sim\rPi_k^\pre}\left[
 \E_{\rv g\sim \cG'_{k,\pi_k}}\left[
  \kl{\cF'_{k, \pi_k, \rv g}}{\cF_{k,\rv g}}
\right]
\right]
\leq
\frac{2}{t_k}\,
\E_{\pi_k\sim\rPi_k^\pre}\left[
  \kl{\cF'_{k,\pi_k}}{\cF_k}
\right]
\notag\\
\qquad
{}+ 2\log|\cS_k|\cdot
\E_{\pi_k \sim \rPi_k^\pre}\left[
  \kl{\cG'_{k,\pi_k}}{\cG_k}
\right].
\label{eq:inductexpect}
\end{align}
Therefore,
\begin{align*}
   \Psi_k & =  \II(\rv F_{k, \rv G_k}; \rPi_k^\post) 
   \\ & \leq
   \E_{\pi_k\sim\rPi_k^\pre}
   \E_{\rv g\sim\cG'_{k,\pi_k}}
\left[ \kl{\cF'_{k, \pi_k, \rv g}}{\cF_{k,\rv g}}\right] \tag{using \Cref{cl:postmixing}} 
   \\ & \leq \frac{2}{t_k}\,
\E_{\pi_k\sim\rPi_k^\pre}\left[
  \kl{\cF'_{k,\pi_k}}{\cF_k}
 \right] + 2\log|\cS_k|\cdot  \II(\rv G_k; \rPi_k^\pre).\tag{using (\ref{eq:inductexpect})}
\end{align*}
Here we used
\[
        \E_{\pi_k\sim\rPi_k^\pre}
        \kl{\cG'_{k,\pi_k}}{\cG_k}
        =
        \II(\rv G_k;\rPi_k^\pre).
\]
The remaining expectation is \(\II(\rv F_k;\rPi_k^\pre)\).  We bound the two
terms. Let \(R_{<k}\) be the free reveals before \(C_k\), and let
\(M_k\) be the transcript before \(C_k\). Then
\[
\rPi_k^{\pre}=(R_{<k},M_k),
\qquad
\HH(M_k\mid R_{<k})\le b.
\]
Since logarithms are base \(2\), a message transcript of at most
\(b\) bits has entropy at most \(b\). For every fixed value of
\(R_{<k}\), the free reveals have only fixed which coordinates of
\(P_{u(k)}\)'s input make up \(\rv F_k\). Those coordinates have not been freely
revealed, so they remain independent uniform entries.  Thus \(\rv F_k\) has law
\(\cF_k\) and is independent of \(R_{<k}\).  By the chain rule
(\itfacts{chain-rule}),
\[
\II(\rv F_k;\rPi_k^{\pre})
=
\II(\rv F_k;R_{<k},M_k)
=
\II(\rv F_k;M_k\mid R_{<k})
\le
\HH(M_k\mid R_{<k})
\le b.
\]

Next, between the point represented by \(\rPi_{k-1}^{\post}\) and the point just
before \(P_{c(k)}\) speaks, the player \(P_{c(k)}\), whose input determines
\(\rv G_k\), does not speak. Write
\[
\rPi_k^{\pre}=(\rPi_{k-1}^{\post},W_k),
\]
where \(W_k\) is the message suffix in this interval. Conditioned on
\(\rPi_{k-1}^{\post}=\rho\), Claim~\ref{claim:rectangle} implies that the input
of \(P_{c(k)}\) remains independent of the inputs of the players who speak during
this suffix. Since \(P_{c(k)}\) does not speak in this interval, \(W_k\) is a
function of those other inputs and the fixed prefix \(\rho\).  Thus \(W_k\) is
conditionally independent of \(\rv G_k\) given \(\rPi_{k-1}^{\post}\).  Hence
\[
\rv G_k\to \rPi_{k-1}^{\post}\to \rPi_k^{\pre}
\]
is a Markov chain. The data processing inequality (\itfacts{data-processing}) gives
\[
\II(\rv G_k;\rPi_k^{\pre})
\le
\II(\rv G_k;\rPi_{k-1}^{\post}).
\]
Since \(\rv G_k\) is a deterministic function of
\(\rv F_{k-1,\rv G_{k-1}}\), another application of
\itfacts{data-processing} gives
\[
\II(\rv G_k;\rPi_{k-1}^{\post})
\le
\II(\rv F_{k-1,\rv G_{k-1}};\rPi_{k-1}^{\post})
=
\Psi_{k-1}
\le
\Phi_{k-1}.
\]
Finally, \(\log|\cS_k|=r_k\log(t_{k+1})\), with \(t_{p+1}=2\).  Therefore
\[
\Psi_k
\le
\frac{2b}{t_k}
+
2r_k\log(t_{k+1})\Phi_{k-1}
\le
\frac{2b}{t_k}
+
3r_k\log(t_{k+1})\Phi_{k-1}
=
\Phi_k.
\]

% !TEX root = main.tex
\section{From Communication to Streaming}\label{sec:reduction}

We now prove Theorem~\ref{thm:main-lb}. 

We reduce \HS{} to \(p\)-pass approximate maximum bipartite matching. Given an \HS{} input, we build a bipartite graph in which a
large matching gives enough information to guess the hidden output bits.  The
construction uses the RS graphs from
Proposition~\ref{prop:RS-gkk12}.

\subsection{The Hierarchical Graph}

Fix an \HS{} input.  For each tuple \((y_1,\ldots,y_p)\), the functions
determine the hidden indices \(i_1,\ldots,i_p\) in sequence: first
\(i_1=f_{p+1}(\ast)\), and then each \(i_\ell\) is determined by the earlier
pairs \((i_1,y_1),\ldots,(i_{\ell-1},y_{\ell-1})\).  This gives a path
\((i_1,y_1,\ldots,i_p,y_p)\), whose output bit is
\(f_1(i_1,y_1,\ldots,i_p,y_p)\).  As the tuple \((y_1,\ldots,y_p)\) varies,
these paths form a depth-\(p\) tree that branches over the possible values of
\(y_\ell\) at level \(\ell\).

Intuitively, the graph is built by repeating the same RS-graph replacement at
each level.  At level \(\ell\), we use an \((r_\ell,t_\ell)\)-RS graph
\(Q_\ell\), whose edges are partitioned into \(t_\ell\) induced matchings of
\(r_\ell\) edges each.  The hidden index \(i_\ell\in[t_\ell]\) selects the
important matching, and the values \(y_\ell\in[r_\ell]\) label its \(r_\ell\)
edges.  On every edge of \(Q_\ell\), we embed the corresponding
level-\((\ell+1)\) block between the two endpoint groups.  Repeating this
through level \(p\) leaves one possible edge at each leaf.  The randomized
deletion rule below correlates the presence of this edge with a masked version
of \(f_1(i_1,y_1,\ldots,i_p,y_p)\).  Auxiliary edges match off the vertices
outside the selected matching. For \(p=2\), this corresponds to the formulation of \cite{KonradN24}. Figure~\ref{fig:rs-hierarchy}
shows one step of the replacement.

\begin{figure}[t]
\centering
\begin{tikzpicture}[
    font=\scriptsize,
    endpoint/.style={
        draw=black!55,
        fill=black!4,
        rounded corners=1.2pt,
        minimum width=1.28cm,
        minimum height=0.42cm,
        inner sep=2pt
    },
    block/.style={
        draw=DarkBlue!70,
        fill=DarkBlue!6,
        rounded corners=1.4pt,
        line width=0.65pt,
        minimum width=2.72cm,
        minimum height=0.5cm,
        inner sep=2pt
    },
    badge/.style={
        draw=DarkRed!70,
        fill=DarkRed!7,
        rounded corners=1.1pt,
        inner xsep=3pt,
        inner ysep=1pt,
        font=\scriptsize
    },
    selected/.style={DarkRed,line width=0.78pt},
    guide/.style={black!25,line width=0.45pt,densely dashed},
    lab/.style={font=\scriptsize},
    header/.style={font=\scriptsize\bfseries}
]
    \node[header,align=center] at (-3.1,2.05) {left vertex\\groups};
    \node[header,align=center] at (0,2.05) {next-level\\blocks};
    \node[header,align=center] at (3.1,2.05) {right vertex\\groups};
    \node[lab,DarkRed] at (0,2.45)
        {the hidden index \(i_\ell\) selects this matching};

    \node[endpoint] (u1) at (-3.1,1.25) {group \(1\)};
    \node[endpoint] (v1) at (3.1,1.25) {group \(1\)};
    \node[block] (g1) at (0,1.25) {level-\((\ell+1)\) block};

    \node[endpoint] (u2) at (-3.1,0.45) {group \(2\)};
    \node[endpoint] (v2) at (3.1,0.45) {group \(2\)};
    \node[block] (g2) at (0,0.45) {level-\((\ell+1)\) block};

    \node[lab] at (-3.1,-0.18) {\(\vdots\)};
    \node[lab] at (0,-0.18) {\(\vdots\)};
    \node[lab] at (3.1,-0.18) {\(\vdots\)};

    \node[endpoint] (ur) at (-3.1,-0.92) {group \(r_\ell\)};
    \node[endpoint] (vr) at (3.1,-0.92) {group \(r_\ell\)};
    \node[block] (gr) at (0,-0.92) {level-\((\ell+1)\) block};

    \begin{scope}[on background layer]
        \node[
            draw=DarkRed!20,
            fill=DarkRed!2,
            rounded corners=3pt,
            fit=(u1) (v1) (ur) (vr),
            inner xsep=8pt,
            inner ysep=8pt
        ] {};
        \draw[guide] (u1.east) -- (v1.west);
        \draw[guide] (u2.east) -- (v2.west);
        \draw[guide] (ur.east) -- (vr.west);
    \end{scope}

    \draw[selected] (u1.east) -- (g1.west);
    \draw[selected] (g1.east) -- (v1.west);
    \draw[selected] (u2.east) -- (g2.west);
    \draw[selected] (g2.east) -- (v2.west);
    \draw[selected] (ur.east) -- (gr.west);
    \draw[selected] (gr.east) -- (vr.west);

    \node[badge,anchor=east] at (-3.92,1.25) {\(y_\ell=1\)};
    \node[badge,anchor=east] at (-3.92,0.45) {\(y_\ell=2\)};
    \node[badge,anchor=east] at (-3.92,-0.92) {\(y_\ell=r_\ell\)};

\end{tikzpicture}
\caption{One step of the recursive replacement; only the selected matching is
shown.  At level \(\ell\), the hidden index \(i_\ell\) selects an induced
matching with \(r_\ell\) edges.  Each \(y_\ell\in[r_\ell]\) labels one of these
edges, on which a level-\((\ell+1)\) block is embedded.  For \(\ell<p\), this
child block is built from an
\((r_{\ell+1},t_{\ell+1})\)-RS graph.  At the bottom level \(\ell=p\), the
recursion ends: each child block has one vertex on each side and one possible
edge between them, whose presence is determined by the randomized rule below.}
\label{fig:rs-hierarchy}
\end{figure}
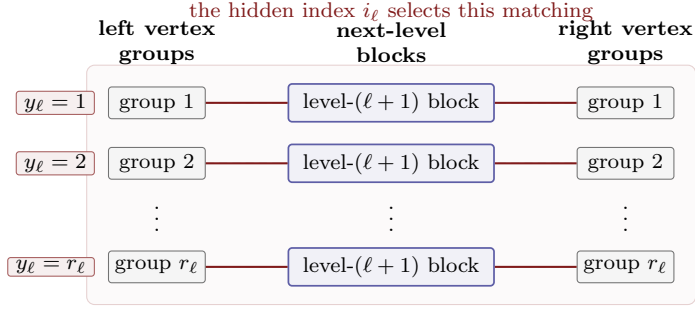

\paragraph{Active prefixes.}
Set \(\mathsf{Pref}_0=\{\emptyset\}\).  For \(\ell\in[p]\), let
\[
        \mathsf{Pref}_\ell=
        [t_1]\times[r_1]\times\cdots\times[t_\ell]\times[r_\ell]
\]
be the set of all prefixes through level \(\ell\).  For \(\ell\in[p]\), given
a parent prefix \(\alpha\in\mathsf{Pref}_{\ell-1}\), \(i\in[t_\ell]\), and
\(y\in[r_\ell]\), write \(\alpha\circ(i,y)\) for the level-\(\ell\) prefix
obtained by appending \((i,y)\).  For every \(\ell\in[p]\), the input defines
a selector
\[
\operatorname{sel}_\ell:\mathsf{Pref}_{\ell-1}\to[t_\ell],
\qquad
\operatorname{sel}_\ell(\alpha)=
\begin{cases}
f_{p+1}(\ast), & \ell=1,\\
f_{p+2-\ell}(\alpha), & \ell\geq2.
\end{cases}
\]
Starting with \(\mathsf{Act}_0=\{\emptyset\}\), define, for \(\ell\in[p]\),
\[
        \mathsf{Act}_\ell=
        \{\alpha\circ(\operatorname{sel}_\ell(\alpha),y):
        \alpha\in\mathsf{Act}_{\ell-1},\ y\in[r_\ell]\}.
\]
The nodes in \(\mathsf{Act}_0,\ldots,\mathsf{Act}_p\), connected by the
parent-child relation above, form the active tree.  For each \(\ell\in[p]\),
every node in \(\mathsf{Act}_{\ell-1}\) has exactly \(r_\ell\) active children:
the selector fixes \(i_\ell\), while \(y_\ell\) ranges over \([r_\ell]\).
Consequently, there is exactly one active leaf for each
tuple \((y_1,\ldots,y_p)\in[r_1]\times\cdots\times[r_p]\).

For each level \(\ell\in[p]\), fix a balanced \((r_\ell,t_\ell)\)-RS graph
\(Q_\ell\) with bipartition \(U_\ell\cup V_\ell\), where
\(|U_\ell|=|V_\ell|=:b_\ell\), and write
\[
        M_{\ell,i}=\{e_{\ell,i,y}:y\in[r_\ell]\},
        \qquad e_{\ell,i,y}=(u_{\ell,i,y},v_{\ell,i,y}),
        \qquad i\in[t_\ell].
\]
Let \(U_{\ell,i}\) and \(V_{\ell,i}\) denote the two endpoint sets of
\(M_{\ell,i}\).

\paragraph{Bottom edges.}
At the bottom of the construction, every full prefix
\(\alpha\in\mathsf{Pref}_p\) has a two-vertex block
\(\Gamma_{p+1}(\alpha)\).  All such blocks use the same local labels
\(L_{p+1}=\{\lambda\}\) and \(R_{p+1}=\{\rho\}\), with \(N_{p+1}=1\), and
have one possible edge \(e_\alpha=(\lambda,\rho)\).  We also use \(e_\alpha\)
for the image of this edge after the block is embedded into the final graph.

The bit attached to this block is \(X_\alpha=f_1(\alpha)\); on active
leaves, these are exactly the output bits of \HS{}.  Independently for every
full prefix, sample a uniform public mask \(Y_\alpha\in\{0,1\}\) and a deletion
coin \(D_\alpha\), private to \(P_1\), with \(\Pr[D_\alpha=1]=q\), where
\(q\in(0,1)\) is a sufficiently small constant depending only on \(p\) and
\(\eps\).  All these random bits are mutually independent.  Set
\[
        X'_\alpha=X_\alpha\oplus Y_\alpha.
\]
The edge \(e_\alpha\) is absent only when \(X'_\alpha=0\) and
\(D_\alpha=1\); otherwise it is present.  Thus an edge with
\(X'_\alpha=1\) is always present, while one with \(X'_\alpha=0\) is present
with probability \(1-q\).  Write
\(X'=(X'_\alpha)_{\alpha\in\mathsf{Pref}_p}\).  Because the masks are
independent fair bits, for every fixed \HS{} input the vector \(X'\) is uniform
on \(\{0,1\}^{\mathsf{Pref}_p}\).  The masks are public to the communication
protocol but are not shown to the streaming algorithm.

\paragraph{Building upward.}
We construct the remaining blocks for \(\ell=p,p-1,\ldots,1\).  Suppose all
level-\((\ell+1)\) blocks have the common labeled sides \(L_{\ell+1}\) and
\(R_{\ell+1}\), each of size \(N_{\ell+1}\).  For a parent prefix
\(\alpha\in\mathsf{Pref}_{\ell-1}\) and every edge \(e_{\ell,i,y}\) of
\(Q_\ell\), embed the child block
\(\Gamma_{\ell+1}(\alpha\circ(i,y))\) between the groups
\(\{u_{\ell,i,y}\}\times L_{\ell+1}\) and
\(\{v_{\ell,i,y}\}\times R_{\ell+1}\), as in
Definition~\ref{def:embed}.  Blocks on RS edges that share an endpoint also
share the corresponding vertex group; blocks on the edges of any one induced
matching are vertex-disjoint.  The index
\(i^\star=\operatorname{sel}_\ell(\alpha)\) designates the selected matching.

Next add \((b_\ell-r_\ell)N_{\ell+1}\) auxiliary vertices on each side.  For
each possible \(i\in[t_\ell]\), fix in advance a perfect matching from
\((U_\ell\setminus U_{\ell,i})\times L_{\ell+1}\) to the right auxiliary
vertices and a perfect matching from the left auxiliary vertices to
\((V_\ell\setminus V_{\ell,i})\times R_{\ell+1}\).  For the selected value
\(i^\star\), add the two corresponding matchings.  Their
\(2(b_\ell-r_\ell)N_{\ell+1}\) edges are the level-\(\ell\) auxiliary edges.
The resulting graph is \(\Gamma_\ell(\alpha)\).  Its vertex set is independent
of \(\alpha\); denote its two sides by \(L_\ell\) and \(R_\ell\), and their
common size by \(N_\ell\).

After carrying out this step at every level, the final graph is
\(G=\Gamma_1(\emptyset)\).  We write
\(\Gamma_{\ell+1}^{(\alpha,i,y)}\) for the copy of the child block embedded on
\(e_{\ell,i,y}\) inside \(\Gamma_\ell(\alpha)\).

\paragraph{Edge ownership.}
We say that a player owns an edge if that player is responsible for listing
the edge in the simulated stream whenever it is present.  We assign the edges
of \(G\) as follows.
\begin{itemize}
    \item Player \(P_1\) owns every bottom edge \(e_\alpha\), for
    \(\alpha\in\mathsf{Pref}_p\).  From \(f_1\), \(Y_\alpha\), and
    \(D_\alpha\), it can determine whether each such edge is present.
    \item For every \(\ell\in[p]\), player \(P_{p+2-\ell}\) owns all
    level-\(\ell\) auxiliary edges.  Its input \(f_{p+2-\ell}\) determines
    \(\operatorname{sel}_\ell\) at every level-\((\ell-1)\) prefix, and hence
    determines these edges.
\end{itemize}
Thus \(P_{p+1}\) owns the level-1 auxiliary edges, \(P_p\) owns the level-2
auxiliary edges, and so on, down to \(P_2\), which owns the level-\(p\)
auxiliary edges.  Embedding a child block into its parent does not change who
lists the child's existing edges; the parent construction adds only the new
auxiliary edges assigned above.  Since every RS edge is replaced by a child
block, every edge of \(G\) is either a bottom edge or an auxiliary edge from
exactly one level, so this assigns every edge to exactly one player.

We fix an order on all possible edge labels, independent of the input, masks,
and deletion coins.  Each player lists its present owned edges in this order.

\paragraph{Active bottom edges.}
Recall that every full prefix \(\alpha\in\mathsf{Pref}_p\) has one possible
bottom edge \(e_\alpha\), owned by \(P_1\).  The leaves
\(\alpha\in\mathsf{Act}_p\) are the full prefixes reached by the active tree,
one for each output-coordinate tuple of \HS{}.  Define
\[
        M^{\star,0}=\{e_\alpha:\alpha\in\mathsf{Act}_p\},
        \qquad
        M^\star=M^{\star,0}\cap E(G).
\]
The set \(M^{\star,0}\) contains one possible edge label for each
output-coordinate tuple, whether or not that edge is present, and hence has
size \(r_1\cdots r_p\).  The set \(M^\star\) contains exactly those active
bottom edges that are actually present in \(G\).  These are the active leaf
edges referred to below.

\subsection{Matching Structure}

The selected RS matchings allow us to construct a matching that contains every
active leaf edge that is present.  We compare the size of this particular
matching with the maximum matching size after all possible active leaf edges
are removed.  The constructed matching is larger by exactly the number of
active leaf edges present.  In the next subsection, this comparison shows that
any sufficiently good approximate matching must contain many of them.

\begin{lemma}
\label{lem:selected-block-invariant-multipass}
Fix a level \(\ell\in[p]\) and a node \(\alpha\in \mathsf{Pref}_{\ell-1}\).  Let
\(i^\star=\operatorname{sel}_\ell(\alpha)\) be the selected hidden index.  In
\(\Gamma_\ell(\alpha)\), consider the vertex groups corresponding to the
endpoints of the selected matching \(M_{\ell,i^\star}\).  The subgraph induced
by these groups is exactly the disjoint union of the child blocks embedded on
the edges of that matching:
\[
        \Gamma_\ell(\alpha)\left[
        \big(U_{\ell,i^\star}\times L_{\ell+1}\big)
        \cup
        \big(V_{\ell,i^\star}\times R_{\ell+1}\big)
        \right]
        =
        \bigsqcup_{y\in[r_\ell]}
        \Gamma_{\ell+1}^{(\alpha,i^\star,y)}.
\]
In particular, no level-\(\ell\) auxiliary edge lies inside this
selected subgraph.
\end{lemma}

\begin{proof}
The matching \(M_{\ell,i^\star}\) is induced in \(Q_\ell\).  After a next-level
block is embedded on every RS edge, restricting to the vertex groups associated
with the endpoints of \(M_{\ell,i^\star}\) leaves exactly the blocks embedded
on the edges of that matching: any other RS edge has at least one endpoint
outside these groups.  Since the selected RS edges form a matching, their child
blocks are vertex-disjoint.  Every level-\(\ell\) auxiliary edge has one regular
endpoint in
\((U_\ell\setminus U_{\ell,i^\star})\times L_{\ell+1}\) or
\((V_\ell\setminus V_{\ell,i^\star})\times R_{\ell+1}\), and its other endpoint
is an auxiliary vertex.  Hence no auxiliary edge lies in the selected induced
subgraph.
\end{proof}

\begin{corollary}
\label{cor:ultimate-induced-matching-multipass}
The set $M^\star$ is an induced matching in $G$.  Moreover,
\(|M^{\star,0}|=r_1r_2\cdots r_p\), and \(M^\star\subseteq M^{\star,0}\)
is exactly the set of active leaf edges that are present after the random rule
for leaf edges.
\end{corollary}

\begin{proof}
Iterate Lemma~\ref{lem:selected-block-invariant-multipass} from the root down to
the leaves, restricting inside each active block to the endpoint groups of its
selected matching.  The final induced subgraph is the disjoint union of the
active bottom blocks.  Each such block has at most one edge, so the present
active bottom edges form an induced matching.  The active tree has
\(r_1r_2\cdots r_p\) leaves.
\end{proof}

\begin{lemma}
\label{lem:size-one-copy-multipass}
For every $\ell\in[p]$,
\[
        N_\ell=(2b_\ell-r_\ell)N_{\ell+1}.
\]
Consequently, \(N_1=\prod_{\ell=1}^p(2b_\ell-r_\ell)\), the graph \(G\) has
\(2N_1\) vertices, and
\[
        \frac{|M^{\star,0}|}{N_1}
        =
        \prod_{\ell=1}^p
        \frac{r_\ell}{2b_\ell-r_\ell}.
\]
In particular, if
\[
        \left(\frac12-\gamma\right)b_\ell
        \leq r_\ell\leq\frac{b_\ell}{2}
        \qquad\text{for every }\ell\in[p],
\]
then
\[
        \left(\frac{1/2-\gamma}{3/2+\gamma}\right)^p
        \leq \frac{|M^{\star,0}|}{N_1}
        \leq 3^{-p},
\]
and hence \(|M^{\star,0}|/N_1=3^{-p}-O_p(\gamma)\).
\end{lemma}

\begin{proof}
At level \(\ell\), each left vertex \(u\in U_\ell\) contributes the
\(N_{\ell+1}\) vertices \((u,x)\) with \(x\in L_{\ell+1}\), and similarly on the
right side.  The block construction then adds
\((b_\ell-r_\ell)N_{\ell+1}\) auxiliary vertices on each side.  Thus
\[
        N_\ell=b_\ell N_{\ell+1}+(b_\ell-r_\ell)N_{\ell+1}
        =(2b_\ell-r_\ell)N_{\ell+1}.
\]
The remaining claims follow by multiplying over levels and using
\(|M^{\star,0}|=r_1r_2\cdots r_p\).
\end{proof}

For the matching-size calculation, set
\(K_\ell:=\prod_{h=\ell}^p r_h\), with \(K_{p+1}:=1\).  This is the number of
possible active leaf edges below any active node at depth \(\ell-1\).  In
particular, \(K_1=|M^{\star,0}|\).
Here and below, deleting possible active leaf edges means deleting those edges,
not their endpoints.

\begin{lemma}
\label{lem:matching-from-active-leaves}
Fix a level \(\ell\in[p+1]\) and an active node \(\alpha\in \mathsf{Act}_{\ell-1}\).  Let
\(m(\alpha)\) be the number of present active leaf edges below \(\alpha\).  Then
\(\Gamma_\ell(\alpha)\) contains a matching of size $N_\ell-\bigl(K_\ell-m(\alpha)\bigr).$
\end{lemma}

\begin{proof}
We prove the claim by induction from the leaves upward.  At level \(p+1\), the
block is a single possible edge.  The claimed size is \(1-(1-m(\alpha))\), which
is \(1\) if the edge is present and \(0\) otherwise.

Now fix $\ell\leq p$ and assume the claim for $\ell+1$.  Let
\(i^\star=\operatorname{sel}_\ell(\alpha)\).  The level-$\ell$
auxiliary edges form a matching of size \(2(b_\ell-r_\ell)N_{\ell+1}\).
They saturate all regular vertices in
\((U_\ell\setminus U_{\ell,i^\star})\times L_{\ell+1}\) and
\((V_\ell\setminus V_{\ell,i^\star})\times R_{\ell+1}\), as well as all
level-\(\ell\) auxiliary vertices.  The unmatched vertices are therefore the
endpoint groups of the selected matching, which induce the vertex-disjoint
blocks \(\Gamma_{\ell+1}^{(\alpha,i^\star,y)}\) by
Lemma~\ref{lem:selected-block-invariant-multipass}.  If \(m_y\) denotes the
number of present active leaf edges below \(\alpha\circ(i^\star,y)\), then
\(\sum_y m_y=m(\alpha)\).  By the induction hypothesis, the selected blocks
contribute $\sum_{y\in[r_\ell]}
        \bigl(N_{\ell+1}-(K_{\ell+1}-m_y)\bigr)$
matching edges.  Together with the auxiliary matching, this has size
\[
        2(b_\ell-r_\ell)N_{\ell+1}
        +r_\ell N_{\ell+1}
        -r_\ell K_{\ell+1}
        +m(\alpha)
        =
        N_\ell-(K_\ell-m(\alpha)).
\]
This proves the claim.
\end{proof}

\begin{lemma}
\label{lem:residual-matching-size}
Fix a level \(\ell\in[p+1]\) and an active node \(\alpha\in \mathsf{Act}_{\ell-1}\).  If all
possible active leaf edges below \(\alpha\) are deleted, then the remaining
graph inside \(\Gamma_\ell(\alpha)\) has maximum matching size $N_\ell-K_\ell.$
\end{lemma}

\begin{proof}
The proof is again by induction from the leaves upward.  At level \(p+1\),
deleting the possible leaf edge leaves a graph with no edges, and
\(N_{p+1}-K_{p+1}=0\).

Now fix $\ell\leq p$ and assume the claim for $\ell+1$.  Let
\(i^\star=\operatorname{sel}_\ell(\alpha)\).  We cover the vertices outside the
selected matching, and then use the induction hypothesis inside the selected
next-level blocks.

After deleting all possible active leaf edges below \(\alpha\), put into the
cover all regular vertices whose level-\(\ell\) RS vertex is not an endpoint of
the selected matching:
\[
        \bigl(U_\ell\setminus U_{\ell,i^\star}\bigr)\times L_{\ell+1}
        \quad\text{and}\quad
        \bigl(V_\ell\setminus V_{\ell,i^\star}\bigr)\times R_{\ell+1}
\]
on the left and right sides.  Because the selected RS matching is induced, these
vertices cover every embedded edge not lying inside a selected next-level block.
They also cover all level-\(\ell\) auxiliary edges.  The only edges
left to cover are inside the selected blocks, which are vertex-disjoint by
Lemma~\ref{lem:selected-block-invariant-multipass}.  By induction, each selected
block has maximum matching size \(N_{\ell+1}-K_{\ell+1}\) after the active leaf
edges below it are deleted.  Since each block is bipartite,
Fact~\ref{fact:konig} gives a vertex cover of the same size in each block.  The
union of these covers with the regular vertices above has size
\[
        2(b_\ell-r_\ell)N_{\ell+1}
        +
        r_\ell(N_{\ell+1}-K_{\ell+1})
        =
        N_\ell-K_\ell.
\]
Thus the maximum matching size is at most \(N_\ell-K_\ell\).  For the reverse
inequality, take the level-\(\ell\) auxiliary matching together with a matching
of size \(N_{\ell+1}-K_{\ell+1}\) in each selected block, whose existence follows
from the induction hypothesis.  These matchings are vertex-disjoint and have
total size \(N_\ell-K_\ell\).  Hence the maximum matching size is exactly
\(N_\ell-K_\ell\).
\end{proof}

At the root, \(m(\emptyset)=|M^\star|\) and
\(K_1=|M^{\star,0}|\).  Lemmas~\ref{lem:matching-from-active-leaves}
and~\ref{lem:residual-matching-size} therefore give
\[
    \mu(G)\ge N_1-\bigl(|M^{\star,0}|-|M^\star|\bigr),
    \qquad
    \mu\bigl(G\setminus M^{\star,0}\bigr)=N_1-|M^{\star,0}|.
\]
Thus the matching constructed in \(G\) is larger than the maximum matching in
\(G\setminus M^{\star,0}\) by exactly \(|M^\star|\), as claimed in the opening
paragraph.

\subsection{Decoding from a Large Matching}

We now show how to turn a sufficiently large matching in \(G\) into a
prediction of the \HS{} output vector.  Each \((y_1,\ldots,y_p)\) corresponds
to one possible active bottom edge in \(M^{\star,0}\).  The argument
has two parts.  First, the matching-size comparison from the previous section,
together with concentration below, forces a sufficiently good approximate
matching to contain many edges from \(M^{\star,0}\).  Second, each such returned
edge gives a slight statistical advantage for guessing its corresponding
output bit.

An active leaf edge returned in the matching does not determine its bit
directly, so we first work with the masked output
\[
        X'_\alpha=X_\alpha\oplus Y_\alpha,
        \qquad \alpha\in \mathsf{Pref}_p.
\]
For every fixed \HS{} input, the masks make the bits of \(X'\) independent and
unbiased.  When \(X'_\alpha=1\), the edge \(e_\alpha\) is always present; when
\(X'_\alpha=0\), it is present with probability \(1-q\).  Thus a valid output
matching that contains \(e_\alpha\) provides evidence that
\(X'_\alpha=1\), although it does not determine the bit.  After showing that a
good matching contains many active leaf edges, we define a randomized decoder:
it guesses \(X'_\alpha=1\) on a fixed-size random subset of those returned
edges, makes independent fair guesses on the remaining active leaves, and then
uses the public masks to obtain guesses for \(X\).  The coupling argument below
proves that this bias remains valid even though the returned matching may
depend on the entire edge stream.

Set
\[
        R_\star:=|M^{\star,0}|=r_1r_2\cdots r_p,
        \qquad
        \beta:=R_\star/N_1.
\]
Thus \(R_\star\) is the number of output coordinates, while \(\beta\) is this
number normalized by the size \(N_1\) of one side of \(G\).
By Lemma~\ref{lem:size-one-copy-multipass}, if
\((1/2-\gamma)b_\ell\leq r_\ell\leq b_\ell/2\) for every \(\ell\), then
\(\beta=3^{-p}-O_p(\gamma)\).

We begin by showing that almost all of these possible active bottom edges are
present.  Each one is deleted with probability \(q/2\), independently over the
active leaves.

\begin{lemma}
\label{lem:present-ultimate-size-multipass}
For every fixed \HS{} input, with probability at least
\(1-1/\operatorname{poly}(R_\star)\) over the bits \(Y_\alpha\) and \(D_\alpha\),
\[
        \left|
        |M^\star|-\left(1-\frac q2\right)R_\star
        \right|
        \le R_\star^{2/3}.
\]
\end{lemma}

\begin{proof}
Once the \HS{} input is fixed, the active leaves are fixed.  For each such leaf,
the public mask \(Y_\alpha\) makes \(X'_\alpha\) uniform.
The edge is deleted exactly when \(X'_\alpha=0\) and \(D_\alpha=1\), which has
probability \(q/2\).  These events are independent over active leaves, so
\(|M^\star|\) is binomial with parameters \(R_\star\) and \(1-q/2\).
Hoeffding's inequality gives
\[
    \Pr\!\left[
    \left||M^\star|-\left(1-\frac q2\right)R_\star\right|
    >R_\star^{2/3}
    \right]
    \le 2\exp(-2R_\star^{1/3}),
\]
which is at most \(1/\operatorname{poly}(R_\star)\) for all sufficiently large
\(R_\star\).
\end{proof}

The concentration bound, together with
Lemma~\ref{lem:matching-from-active-leaves}, gives a lower bound on
\(\mu(G)\).  Let \(G_{\rm res}:=G\setminus M^{\star,0}\) be the graph obtained
by deleting all possible active leaf edges.  Lemma~\ref{lem:residual-matching-size} gives
\(\mu(G_{\rm res})=N_1-R_\star\).

\begin{lemma}
\label{lem:large-opt-and-residual-multipass}
For every fixed \HS{} input, with probability at least
\(1-1/\operatorname{poly}(R_\star)\) over the bits \(Y_\alpha\) and \(D_\alpha\),
\[
        \mu(G)
        \geq
        N_1-\frac q2R_\star-R_\star^{2/3}
        =
        \left(1-\frac q2\beta-\beta R_\star^{-1/3}\right)N_1.
\]
Moreover, for every outcome of these bits,
\(\mu(G_{\rm res})=N_1-R_\star\).
\end{lemma}

\begin{proof}
Lemma~\ref{lem:matching-from-active-leaves} applied at the root gives a
matching of size \(N_1-(R_\star-|M^\star|)\).
Lemma~\ref{lem:present-ultimate-size-multipass}
then gives
\[
        N_1-\bigl(R_\star-|M^\star|\bigr)
        \ge
        N_1-\frac q2R_\star-R_\star^{2/3}.
\]
The equality for \(G_{\rm res}\) follows from
Lemma~\ref{lem:residual-matching-size} at the root after deleting all possible
active leaf edges.
\end{proof}

Any sufficiently large approximate matching in \(G\) must include many active
leaf edges.  Otherwise its other edges lie in \(G_{\rm res}\), whose maximum
matching is too small.

\begin{lemma}
\label{lem:approx-finds-ultimate-edges-multipass}
Suppose
\[
        \left(\frac12-\gamma\right)b_\ell
        \leq r_\ell\leq\frac{b_\ell}{2}
        \qquad\text{for every }\ell\in[p],
\]
where $\gamma>0$ and $q>0$ are sufficiently small constants depending only on
$p$ and $\eps$.  Then, for every fixed \HS{} input and all sufficiently large
\(R_\star\), with probability at least
\(1-1/\operatorname{poly}(R_\star)\) over the bits \(Y_\alpha\) and
\(D_\alpha\), every
\((1-3^{-p}+\eps)\)-approximate matching \(M\) in \(G\) contains at least
\(c_{p,\eps}N_1\) edges from
\(M^{\star,0}\), for some constant \(c_{p,\eps}>0\).
\end{lemma}

\begin{proof}
All edges of \(M\) outside \(M^{\star,0}\) are edges of \(G_{\rm res}\), so $|M|\leq |M\cap M^{\star,0}|+\mu(G_{\rm res}).$
Using Lemma~\ref{lem:large-opt-and-residual-multipass} and
\(\beta=R_\star/N_1\), every \((1-3^{-p}+\eps)\)-approximate matching
satisfies
\[
\begin{aligned}
        |M\cap M^{\star,0}|
        &\geq
        (1-3^{-p}+\eps)
        \left(N_1-\frac q2R_\star-R_\star^{2/3}\right)
        -(N_1-R_\star) \\
        &=
        \left[
        \beta-3^{-p}+\eps
        -
        \frac q2\beta(1-3^{-p}+\eps)
        -\beta(1-3^{-p}+\eps)R_\star^{-1/3}
        \right]N_1.
\end{aligned}
\]
The hypothesis and Lemma~\ref{lem:size-one-copy-multipass} give
\[
    \left(\frac{1/2-\gamma}{3/2+\gamma}\right)^p
    \leq \beta\leq 3^{-p},
\]
so \(\beta=3^{-p}-O_p(\gamma)\).  First choosing
\(\gamma,q\ll_p\eps\), and then taking \(R_\star\) sufficiently large, makes
the bracketed quantity at least a positive constant depending only on \(p\)
and \(\eps\).  Choose \(c_{p,\eps}\) below that constant, and also below one
tenth of \(\bigl((1/2-\gamma)/(3/2+\gamma)\bigr)^p\).
\end{proof}

Set \(s_\star=\lfloor c_{p,\eps}N_1\rfloor\), decreasing
\(c_{p,\eps}\) if needed so
Lemma~\ref{lem:approx-finds-ultimate-edges-multipass} guarantees at least
\(s_\star\) active leaf edges whenever the lemma applies.

Given a matching \(M\) and the selector functions
\(f_2,\ldots,f_{p+1}\), the decoder can reconstruct the active leaves and
identify \(M\cap M^{\star,0}\).  If this set has size at least \(s_\star\), choose a
uniform subset \(S\subseteq M\cap M^{\star,0}\) of size \(s_\star\); otherwise
set \(S=\emptyset\).  On leaves whose edges are in \(S\), set
\(\widehat X'_\alpha=1\).  On all other active leaves, use independent fair
guesses.  The decoder chooses \(S\) using only \(M\cap M^{\star,0}\) and its
own randomness; it does not directly inspect the public mask values when
choosing \(S\).  The masks are used only afterward to convert guesses for
\(X'\) into guesses for \(X\).

The validity assumption in the next lemma is used only to rule out absent edges
in the output matching.

\begin{lemma}
\label{lem:biased-decoding-multipass}
Fix an arbitrary \HS{} input.  Let \(M\) be the output of a randomized
procedure that is given only the resulting edge stream, listed in the fixed
order on edge labels, together with independent internal randomness.  Assume
that this procedure always outputs a valid matching in the input graph.  Let \(S\) and
\(\widehat X'\in\{0,1\}^{\mathsf{Act}_p}\) be produced by the decoder above,
with \(S\) determined only by \(M\cap M^{\star,0}\) and independent decoder
randomness.  Then, over the bits \(Y_\alpha,D_\alpha\), the procedure's
internal randomness, and decoder randomness,
\[
        \E\!\left[d_{\rm H}(\widehat X',X'|_{\mathsf{Act}_p})\right]
        \le
        \frac{R_\star}{2}
        -
        \frac{q}{2(2-q)}\,\E[|S|].
\]
In particular, if \(\mathcal E\) is any event on which
\(|M\cap M^{\star,0}|\ge s_\star\), then
\[
        \E\!\left[d_{\rm H}(\widehat X',X'|_{\mathsf{Act}_p})\right]
        \le
        \frac{R_\star}{2}
        -
        \frac{q}{2(2-q)}\,s_\star\Pr[\mathcal E].
\]
\end{lemma}

\begin{proof}
Fix an active leaf \(\alpha\), and let \(C_\alpha\) be the event that its edge is
chosen into \(S\).  Since \(M\) is a valid matching in the input graph,
\(C_\alpha\) can occur only when the leaf edge is present.  Fix all randomness
other than \(Y_\alpha\) and \(D_\alpha\), including the procedure's internal
randomness and the decoder randomness.  In each of the three cases
\[
    (X'_\alpha,D_\alpha)\in\{(1,0),(1,1),(0,0)\},
\]
the edge \(e_\alpha\) is present, and the resulting edge stream is identical.
Indeed, \(Y_\alpha\) and \(D_\alpha\) affect no other edge, and the fixed edge
order is independent of these bits.  The procedure therefore returns the same
matching, and the decoder chooses the same set \(S\), in all three cases.  In
the remaining case \((X'_\alpha,D_\alpha)=(0,1)\), the edge is absent and
validity of \(M\) makes \(C_\alpha\) impossible.  Averaging over the randomness
used in this comparison gives
\[
        \Pr[C_\alpha\mid X'_\alpha=0]
        =
        (1-q)\Pr[C_\alpha\mid X'_\alpha=1].
\]
Since \(X'_\alpha\) is unbiased, whenever \(\Pr[C_\alpha]>0\),
\[
        \Pr[X'_\alpha=1\mid C_\alpha]
        =
        \frac{1}{2-q}.
\]
Conditioned on \(e_\alpha\in S\), the estimate \(\widehat X'_\alpha=1\) is
correct with probability \(1/(2-q)\), a gain over a fair guess of
\[
        \frac{1}{2-q}-\frac12=\frac{q}{2(2-q)}.
\]
If \(e_\alpha\notin S\), the decoder uses an independent fair bit, whose error
probability is \(1/2\) regardless of the conditional distribution of
\(X'_\alpha\).  Consequently,
\[
    \Pr[\widehat X'_\alpha\ne X'_\alpha]
    =
    \frac12-\frac{q}{2(2-q)}\Pr[C_\alpha].
\]
Summing this identity over all active leaves and using
\(\sum_{\alpha\in\mathsf{Act}_p}\Pr[C_\alpha]=\E[|S|]\) gives
\[
        \E\!\left[d_{\rm H}(\widehat X',X'|_{\mathsf{Act}_p})\right]
        \le
        \frac{R_\star}{2}
        -
        \frac{q}{2(2-q)}\,\E[|S|].
\]
On \(\mathcal E\), we have \(|M\cap M^{\star,0}|\ge s_\star\).  Therefore
\(|S|=s_\star\) on \(\mathcal E\), and
\(\E[|S|]\ge s_\star\Pr[\mathcal E]\).
\end{proof}

\subsection{Simulating the Streaming Algorithm}
Let \(\alg\) be a \(p\)-pass streaming algorithm for approximate maximum
bipartite matching, using \(s_{\alg}\) bits of space.  We simulate \(\alg\) on
the graph \(G\).  Public coins determine the masks \(Y_\alpha\) and the random
bits used by \(\alg\), while \(P_1\) privately samples the deletion coins
\(D_\alpha\).  These choices are independent and fixed for all passes.  

For each $a \in [p+1]$, let $\sigma_a$ be the sequence of present edges owned by $P_a$,
listed in the fixed order on edge labels defined above. Thus the stream is partitioned as
$\sigma=\sigma_1\circ\cdots\circ\sigma_{p+1}$, and each player can generate its portion of
the stream locally from its input and randomness. By Proposition~A.8, the standard
streaming-to-communication simulation gives a $p$-round $(p+1)$-party protocol in which
the charged communication from passing the streaming memory state is at most
$p(p+1)s_A=O_p(s_A)$ bits. No edge list needs to be communicated.

We also use the free reveals allowed by
Assumption~\ref{assump:reveal-input}.  After the critical player
\(P_{p+2-k}\) sends its memory state in pass \(k\), it reveals its entire input
\(f_{p+2-k}\), which specifies the selected level-\(k\) matching for every
level-\((k-1)\) prefix.  These are exactly the \(p\) reveals excluded from the
charged communication in the proof of Theorem~\ref{thm:HS}; no other
input-dependent communication is free.  In particular, neither \(f_1\) nor the
deletion coins are revealed.  The reveals are available to the protocol
players but not to \(\alg\), which receives only the edge stream and its random
bits.

The timing prevents the reveals from making the problem trivial.  In the final
pass, \(P_1\) sends its last memory state before \(P_2\) reveals \(f_2\), so
\(P_1\) never communicates after all active leaves are known.  The revealed
functions identify the relevant entries of \(f_1\), but their values can reach
the final player only through the charged memory states.

At the end of the last pass, \(P_{p+1}\) computes the matching \(M\) output by
\(\alg\).  From \(f_2,\ldots,f_{p+1}\), it reconstructs the active leaves and
\(M^{\star,0}\), and then applies the decoder defined above.  The decoder needs
no deletion coins: it uses the active leaf edges contained in the valid
matching \(M\) and makes fair guesses on the remaining active leaves.  Thus the
charged transcript consists only of streaming memory states, and the resulting
randomized \(p\)-pass protocol for \HS{} has cost \(O_p(s_{\alg})\).

\subsection{Parameter Choice and Proof}

\begin{proof}[Proof of Theorem~\ref{thm:main}]
We first show that a small-space streaming algorithm gives an \HS{} protocol
with charged communication \(O_p(s_{\alg})\) and nontrivial Hamming advantage.
We then choose the RS
parameters so that Theorem~\ref{thm:HS} forces the desired space lower bound.

Choose constants \(\gamma,q>0\), depending only on \(p\) and \(\eps\),
small enough for Lemma~\ref{lem:approx-finds-ultimate-edges-multipass}.
Let $\alg$ be a randomized constant-error \(p\)-pass
\((1-3^{-p}+\eps)\)-approximation streaming algorithm for maximum
bipartite matching that outputs a valid matching in the input graph on every
run.
Let \(s_{\alg}\) be its space usage.
A run is successful if this matching has the required approximation ratio.
Since \(\alg\) has constant error, there is a fixed constant
\(\zeta>0\) such that the success probability is at least \(1/2+\zeta\).

Run the communication simulation with \(\alg\).  Fix an arbitrary \HS{} input.
All probabilities below are
over the public masks, the deletion coins, the randomness of \(\alg\), and the
decoder randomness.  Let \(\mathsf{Succ}\) be the event that the simulated
streaming run succeeds, and let \(\mathsf{Good}\) be the event that the conclusion
of Lemma~\ref{lem:approx-finds-ultimate-edges-multipass} holds for the sampled
graph.  For every realized graph \(G\), the algorithm succeeds with probability
at least \(1/2+\zeta\).  Hence
\(\Pr[\mathsf{Succ}]\ge 1/2+\zeta\) after averaging over the random graph
construction.  The parameter choice below satisfies the hypotheses of
Lemma~\ref{lem:approx-finds-ultimate-edges-multipass} and has
\(R_\star\to\infty\), so \(\Pr[\mathsf{Good}]=1-o(1)\).  On
\(\mathsf{Succ}\cap\mathsf{Good}\), the output matching satisfies
\(|M\cap M^{\star,0}|\ge s_\star\), and
\[
        \Pr[\mathsf{Succ}\cap\mathsf{Good}]
        \ge \Pr[\mathsf{Succ}]+\Pr[\mathsf{Good}]-1
        \ge \frac12+\zeta-o(1).
\]
The final player applies the decoding procedure defined in the paragraph
immediately preceding Lemma~\ref{lem:biased-decoding-multipass}; the guarantee
for this procedure is proved in that lemma.  The bits \(Y_\alpha\) are public, so
it forms an output vector \(\widehat X\) by setting
\(\widehat X_\alpha=\widehat X'_\alpha\oplus Y_\alpha\) for
\(\alpha\in \mathsf{Act}_p\).  For each output-coordinate tuple
\((y_1,\ldots,y_p)\), the active-tree recursion determines exactly one
\(\alpha\in\mathsf{Act}_p\).  Hence these values specify one output bit for
every coordinate of the \HS{} instance.
For every outcome of the randomness,
\[
        d_{\rm H}(\widehat X,X|_{\mathsf{Act}_p})
        =
        d_{\rm H}(\widehat X',X'|_{\mathsf{Act}_p}),
\]
because both sides use the same public mask \(Y_\alpha\) at each active leaf.
Applying Lemma~\ref{lem:biased-decoding-multipass} with the event
\(\mathsf{Succ}\cap\mathsf{Good}\) gives
\[
        \E\!\left[
        d_{\rm H}(\widehat X,X|_{\mathsf{Act}_p})
        \right]
        \le
        \frac{R_\star}{2}
        -
        \frac{q}{2(2-q)}\,s_\star
        \Pr[\mathsf{Succ}\cap\mathsf{Good}].
\]
For all sufficiently large \(R_\star\), this probability is at least
\(1/2+\zeta/2\).  Also
\(s_\star/R_\star=\Theta_{p,\eps}(1)\): indeed
\(s_\star=\Theta_{p,\eps}(N_1)\) and
\(R_\star=(3^{-p}-O_p(\gamma))N_1\).  Thus the expected Hamming error is at
most
\[
        \left(\frac12-\eta_{p,\eps}\right)R_\star
\]
for some constant \(\eta_{p,\eps}>0\).
This holds for every fixed \HS{} input, over the masks, deletion coins,
streaming-algorithm randomness, and decoder randomness.  Thus the resulting
protocol has worst-case Hamming advantage \(\eta_{p,\eps}\).

The simulation in the preceding subsection has charged communication at most
\(p(p+1)s_{\alg}\).  Applying the charged-communication form proved under
Assumption~\ref{assump:reveal-input} in the proof of
Theorem~\ref{thm:HS} gives
\[
        s_{\alg}
        \ge
        \widetilde{\Omega}_{p,\eps}\left(
        \min_{k\in[p]} r_1r_2\cdots r_k\,t_k
        \right).
\]
We now choose parameters that make the right-hand side superlinear.  Fix a
sufficiently large integer \(n\), the desired number of vertices in the final
streaming instance.  We first construct a graph with at most \(n\) vertices and
then add isolated vertices to make its vertex count exactly \(n\).  After
fixing \(\gamma\), apply Proposition~\ref{prop:RS-gkk12} with parameter
\(\gamma/2\), and let \(c_{\rm RS}>0\) be the resulting constant.  Thus the RS
construction gives
\[
        t_\ell\ge b_\ell^{c_{\rm RS}/\log\log b_\ell}
\]
for all sufficiently large side sizes.  Choose a constant
\(0<\theta\leq\min\{1/4,c_{\rm RS}/16\}\), and set
\[
        N=\left\lfloor\frac{n}{2(3/2+\gamma)^p}\right\rfloor,
        \qquad
        L=\log\log N.
\]
Define
\[
        a_\ell=(\theta/L)^{\ell-1}\quad(\ell=2,\ldots,p),
        \qquad
        a_1=1-\sum_{\ell=2}^p a_\ell,
\]
so that \(a_1,\ldots,a_p>0\) and \(\sum_{\ell=1}^p a_\ell=1\).  Set
the target side size at level \(\ell\) to \(\lfloor N^{a_\ell}\rfloor\).
These targets all tend to infinity because \(p\) is fixed and
\(N^{a_p}\to\infty\).  The every-\(n\) form of
Proposition~\ref{prop:RS-gkk12} therefore lets us take RS graphs with exactly
\(b_\ell\) vertices in each part, where
\[
        b_\ell=\lfloor N^{a_\ell}\rfloor
        =\Theta_{p,\eps}(N^{a_\ell})
\]
and at least \(b_\ell^{c_{\rm RS}/\log\log b_\ell}\) induced matchings of
size \(\lfloor(1/2-\gamma/2)b_\ell\rfloor\).  Since every \(b_\ell\) tends to
infinity, for all sufficiently large \(n\),
\[
    \left\lfloor\left(\frac12-\frac\gamma2\right)b_\ell\right\rfloor
    \geq
    \left\lceil\left(\frac12-\gamma\right)b_\ell\right\rceil.
\]
We may therefore delete edges from each induced matching and take
\[
        r_\ell=\left\lceil\left(\frac12-\gamma\right)b_\ell\right\rceil,
        \qquad
        t_\ell\geq b_\ell^{c_{\rm RS}/\log\log b_\ell}.
\]
For all sufficiently large \(n\), these choices satisfy
\((1/2-\gamma)b_\ell\leq r_\ell\leq b_\ell/2\), as required by
Lemma~\ref{lem:approx-finds-ultimate-edges-multipass}.
The hierarchical graph \(G\) constructed above, before any isolated vertices
are added, has at most \(n\) vertices, since
\[
        2N_1
        =2\prod_{\ell=1}^p(2b_\ell-r_\ell)
        \leq2(3/2+\gamma)^p\prod_{\ell=1}^p b_\ell
        \leq 2(3/2+\gamma)^pN
        \leq n.
\]
We then add isolated vertices to \(G\) until its vertex count is exactly \(n\).
This changes neither the edge stream nor any matching quantity, so the
simulation and decoding argument above apply unchanged to the padded graph.

Fix \(k\in[p]\).  Since
\(a_p\log N=\theta^{p-1}\log N/L^{p-1}\) grows faster than \(L\), the
choice \(b_k=\lfloor N^{a_k}\rfloor\) gives, for all sufficiently large \(N\),
\[
        \log b_k\geq\frac{a_k}{2}\log N,
        \qquad
        \log\log b_k\leq L,
\]
and hence
\[
        t_k\geq N^{c_{\rm RS}a_k/(2L)}.
\]
The terms \(a_2,\ldots,a_p\) form a geometric sequence.  Our choice of
\(\theta\) ensures, for every \(k\in[p]\), that
\[
        \sum_{j=k+1}^p a_j\leq\frac{c_{\rm RS}a_k}{4L}.
\]
Indeed, for \(k\geq2\) the left-hand side is at most
\(2\theta a_k/L\); for \(k=1\), it is at most \(2\theta/L\) and
\(a_1\geq1/2\).  Since \(r_j=\Theta(b_j)\), it follows that
\[
\begin{aligned}
        r_1\cdots r_k\,t_k
        &\geq
        \widetilde{\Omega}_{p,\eps}\!\left(
        N^{\sum_{j=1}^k a_j+c_{\rm RS}a_k/(2L)}
        \right) \\
        &\geq
        \widetilde{\Omega}_{p,\eps}\!\left(
        N^{1+c_{\rm RS}a_k/(4L)}
        \right).
\end{aligned}
\]
The smallest \(a_k/L\) is
\(a_p/L=\theta^{p-1}/L^p=\Theta_{p,\eps}(L^{-p})\).  Moreover,
\(N=\Theta_{p,\eps}(n)\) and \(L=\Theta(\log\log n)\).  Since the induced
matchings are edge-disjoint, \(t_\ell r_\ell\leq b_\ell^2\); together with
\(r_\ell=\Theta_{p,\eps}(b_\ell)\), this shows that every \(r_\ell\) and
\(t_\ell\) is \(O_{p,\eps}(n)\).  Thus all logarithmic factors suppressed above
form a fixed polylogarithmic factor in \(n\).  Because
\(\log N/L^p=\omega(\log\log N)\), dividing by any fixed polylogarithmic factor
still leaves a bound of the form \(N^{1+\Omega_{p,\eps}(L^{-p})}\).  Combining
this with the communication lower bound above gives
\[
        s_{\alg}\geq
        n^{1+\Omega_{p,\eps}((\log\log n)^{-p})}
\]
bits of space.
\end{proof}

\section*{Acknowledgements and AI Disclosure}

The authors thank Sepehr Assadi and Christian Konrad for many useful discussions about the problem and literature.

The authors used ChatGPT 5.5 Pro to assist in preparing portions of this
manuscript, in particular, the second proof presented for our main theorem. Specifically, it was used to draft substantial portions of Section~7 from an
author-provided outline and technical direction, after which the authors
extensively revised and verified the resulting material. It was also used to
reorganize parts of Section~6 and to assist with the exposition and formulation
of several lemmas near the beginning of that section, again followed by
extensive revision and verification by the authors. Its remaining use was
limited to occasional rewording, proofreading, and minor edits to the
introduction. The authors take full responsibility for all content of the
paper.

\bibliographystyle{alpha}
\bibliography{general}

\appendix

% !TeX root = main.tex 
%!TEX root = main.tex

\section{Additional Preliminaries}\label{app:additional-prelims}

\subsection{Information Theory}\label{sec:info-theory}\label{app:info}

We collect the information-theoretic notation and facts used throughout the
paper. All random variables in this subsection are finite-valued. We refer the
reader to Cover and Thomas~\cite{CoverT06} for background and proofs of the
standard facts below.

For a random variable \(\rA\), the \emph{Shannon entropy} of \(\rA\) is
\begin{align}
    \en{\rA}
    :=
    \sum_{A\in\supp{\rA}}
    \Pr\paren{\rA=A}\cdot
    \log\paren{1/\Pr\paren{\rA=A}}.
    \label{eq:entropy}
\end{align}
The \emph{conditional entropy} of \(\rA\) conditioned on \(\rB\) is
\begin{align}
    \en{\rA\mid\rB}
    :=
    \Ex_{B\sim\distribution{\rB}}\bracket{\en{\rA\mid \rB=B}},
    \label{eq:cond-entropy}
\end{align}
where \(\en{\rA\mid \rB=B}\) is defined using the distribution of \(\rA\)
conditioned on the event \(\rB=B\).

The \emph{mutual information} of two random variables \(\rA\) and \(\rB\) is
denoted by \(\mi{\rA}{\rB}\) and defined as
\begin{align}
    \mi{\rA}{\rB}
    :=
    \en{\rA}-\en{\rA\mid\rB}
    =
    \en{\rB}-\en{\rB\mid\rA}.
    \label{eq:mi}
\end{align}
The \emph{conditional mutual information} is
\[
    \mi{\rA}{\rB\mid\rC}
    :=
    \en{\rA\mid\rC}-\en{\rA\mid\rB,\rC}.
\]

\subsubsection{Entropy and Mutual Information}
\label{sec:prop-en-mi}

The following basic properties of entropy and mutual information are used
throughout.

\begin{fact}\label{fact:it-facts}
Let \(\rA,\rB,\rC,\rD\) be four possibly correlated random variables.
\begin{enumerate}
    \item \label{part:uniform}
    \(0\leq \en{\rA}\leq \log\card{\supp{\rA}}\). The right equality holds
    if and only if \(\distribution{\rA}\) is uniform.
    \item \label{part:info-zero}
    \(\mi{\rA}{\rB\mid\rC}\geq0\). Equality holds if and only if \(\rA\) and
    \(\rB\) are independent conditioned on \(\rC\).
    \item \label{part:cond-reduce}
    Conditioning on a random variable reduces entropy:
    \[
        \en{\rA\mid\rB,\rC}\leq \en{\rA\mid\rB}.
    \]
    Equality holds if and only if \(\rA\perp \rC\mid\rB\).
    \item \label{part:chain-rule}
    Chain rule for mutual information:
    \[
        \mi{\rA,\rB}{\rC\mid\rD}
        =
        \mi{\rA}{\rC\mid\rD}
        +
        \mi{\rB}{\rC\mid\rA,\rD}.
    \]
    \item \label{part:data-processing}
    Data processing inequality: for a function \(f(\rA)\) of \(\rA\),
    \[
        \mi{f(\rA)}{\rB\mid\rC}
        \leq
        \mi{\rA}{\rB\mid\rC}.
    \]
\end{enumerate}
\end{fact}

We also use the following standard proposition on the effect of conditioning on
mutual information.

\begin{proposition}\label{prop:info-increase}
For random variables \(\rA,\rB,\rC,\rD\), if \(\rA\perp \rD\mid\rC\), then
\[
    \mi{\rA}{\rB\mid\rC}
    \leq
    \mi{\rA}{\rB\mid\rC,\rD}.
\]
\end{proposition}

\begin{proof}
Since \(\rA\) and \(\rD\) are independent conditioned on \(\rC\), by
\itfacts{cond-reduce},
\(\HH(\rA\mid\rC)=\HH(\rA\mid\rC,\rD)\). Also
\(\HH(\rA\mid\rC,\rB)\geq\HH(\rA\mid\rC,\rB,\rD)\). Thus
\begin{align*}
    \mi{\rA}{\rB\mid\rC}
    &=
    \HH(\rA\mid\rC)-\HH(\rA\mid\rC,\rB) \\
    &\leq
    \HH(\rA\mid\rC,\rD)-\HH(\rA\mid\rC,\rB,\rD)
    =
    \mi{\rA}{\rB\mid\rC,\rD}.
\end{align*}
\end{proof}

\subsubsection{Distances Between Distributions}
\label{sec:prob-distance}

For two finite distributions \(\mu\) and \(\nu\) over the same probability
space, the \emph{Kullback--Leibler divergence} of \(\mu\) relative to \(\nu\)
is denoted by \(\kl{\mu}{\nu}\) and defined as
\[
    \kl{\mu}{\nu}
    :=
    \Ex_{a\sim\mu}\Bracket{\log\frac{\mu(a)}{\nu(a)}}.
\]
We use the standard convention that terms with \(\mu(a)=0\) contribute \(0\),
and that \(\kl{\mu}{\nu}=\infty\) if \(\mu(a)>0\) and \(\nu(a)=0\) for some
\(a\).

\begin{fact}\label{fact:kl-info}
For random variables \(\rA,\rB,\rC\),
\[
    \mi{\rA}{\rB\mid\rC}
    =
    \Ex_{(B,C)\sim\distribution{(\rB,\rC)}}
    \Bracket{
        \kl{\distribution{\rA\mid\rB=B,\rC=C}}
            {\distribution{\rA\mid\rC=C}}
    }.
\]
\end{fact}

For two distributions \(\mu\) and \(\nu\) on the same finite set \(\Omega\), the
\emph{total variation distance} between them is
\[
    \tvd{\mu}{\nu}
    :=
    \max_{\Omega'\subseteq\Omega}\card{\mu(\Omega')-\nu(\Omega')}
    =
    \frac{1}{2}\sum_{x\in\Omega}\card{\mu(x)-\nu(x)}.
\]

\begin{fact}\label{fact:tvd-small}
Suppose \(\mu\) and \(\nu\) are two distributions on a probability space
containing an event \(\event\).
Then
\[
    \mu(\event)\leq \nu(\event)+\tvd{\mu}{\nu}.
\]
\end{fact}

\begin{fact}\label{fact:tvd-chain-rule}
For any distributions \(\mu\) and \(\nu\) on \(n\)-tuples
\((X_1,\ldots,X_n)\),
\[
    \tvd{\mu}{\nu}
    \leq
    \sum_{i=1}^{n}
    \Exp_{X_{<i}\sim\mu}
    \tvd{\mu(X_i\mid X_{<i})}{\nu(X_i\mid X_{<i})}.
\]
Here conditional distributions under \(\nu\) may be chosen arbitrarily on
prefixes \(X_{<i}\) to which \(\nu\) assigns probability zero.
\end{fact}

\begin{fact}[Pinsker's inequality]\label{fact:pinskers}
For any distributions \(\mu\) and \(\nu\),
\[
    \tvd{\mu}{\nu}
    \leq
    \sqrt{\frac{1}{2}\cdot \kl{\mu}{\nu}}.
\]
\end{fact}

\subsection{Streaming and Communication}

We work in the insertion-only edge-arrival streaming model. A \(p\)-pass
streaming algorithm sees the same edge stream \(p\) times and outputs a
matching at the end of the final pass. We assume the algorithm always outputs an edge set that is a valid matching in the input graph. An \(\alpha\)-approximation algorithm succeeds on \(G\) if this
matching has size at least \(\alpha\mu(G)\). For randomized streaming
algorithms, constant error means that the algorithm succeeds with probability
at least \(1/2+\zeta\), for some fixed constant \(\zeta>0\), on every input.

We also use the standard multi-party number-in-hand communication model. There
are \(k\geq1\) players \(P_1,\ldots,P_k\), and in each round the players speak
in the order \(P_1,\ldots,P_k\). Messages are written on a public blackboard
visible to all players. Protocols may use public and private randomness.

\begin{definition}\label{def:cc}
For any protocol \(\prot\), the \textbf{communication cost} of \(\prot\),
denoted by \(\cc{\prot}\), is the worst-case total number of bits communicated by all players in all rounds, over all inputs and all outcomes of the protocol randomness. 
We assume, by padding if needed, that all charged transcripts of \(\prot\) have
the same worst-case length.
\end{definition}

\begin{proposition}[cf.~\cite{AlonMS96}]\label{prop:cc-stream}
For any \(p\geq1\), \(s\geq1\), and \(\delta\in(0,1)\), suppose there is a
\(p\)-pass \(s\)-space streaming algorithm \(\alg\) for a problem whose input
stream is partitioned as \(\sigma=\sigma_1\circ\cdots\circ\sigma_k\), and
suppose \(\alg\) succeeds with probability at least \(\delta\). Then there is a
\(k\)-party protocol \(\prot\) with \(p\) rounds, communication cost
\(\cc{\prot}\leq pks\), and success probability at least \(\delta\) for the
same problem, where the final output is produced by the last player and is not
charged as communication. If the output itself is required to be written on the
blackboard, its encoding length must be added to the communication cost.
\end{proposition}

\begin{proof}
Consider a stream \(\sigma=\sigma_1\circ\dots\circ\sigma_k\), where
\(\sigma_i\) is the input held by player \(P_i\), in an arbitrary order.
The players use public randomness to sample the random tape of \(\alg\).
Player \(P_1\) runs \(\alg\) on \(\sigma_1\) and writes the memory state on the
blackboard; player \(P_2\) continues from that state on \(\sigma_2\), and so
on. This simulates one pass using at most \(ks\) bits of communication.
Repeating this for \(p\) passes gives communication at most \(pks\), and the
last player outputs the answer of \(\alg\). The protocol has the same success
probability as the streaming algorithm.
\end{proof}

\Cref{prop:cc-stream} allows us to translate communication lower bounds into
streaming lower bounds.

\subsection{Information Cost and Message Compression}
\label{subsec:info-cost-compression}

We also use the notion of \emph{information cost} of protocols, introduced
in~\cite{ChakrabartiSWY01}; see, e.g., \cite{Weinstein15} for a survey. The
definition below is the natural multi-party version of external information
cost~\cite{BarakBCR10}.

\begin{definition}\label{def:ext-info}
For any multi-party protocol \(\prot\) whose inputs are distributed according
to a distribution \(\mu\), the \textbf{external information cost} is
\[
    \ic{\prot}{\mu}:=\mi{\rX}{\rProt\mid\rR},
\]
where \(\rX\) denotes the random input sampled from \(\mu\), \(\rProt\)
denotes the charged transcript written on the blackboard, and \(\rR\) is the
public randomness.
\end{definition}

In words, multi-party external information cost measures the information
revealed by the charged transcript about the entire input to an external
observer.
The following result follows from~\cite{BarakBCR10}.

\begin{proposition}[cf.~\cite{BarakBCR10}]\label{prop:ic-cc}
For any multi-party protocol \(\prot\) on any input distribution \(\mu\),
\[
    \ic{\prot}{\mu}\leq \cc{\prot}.
\]
\end{proposition}

\begin{proof}
We have
\begin{align*}
    \ic{\prot}{\mu}
    \Eq{(1)}
    \mi{\rX}{\rProt\mid\rR}
    \Eq{(2)}
    \en{\rProt\mid\rR}-\en{\rProt\mid\rX,\rR}
    \Leq{(3)}
    \en{\rProt\mid\rR}
    \Leq{(4)}
    \en{\rProt}
    \Leq{(5)}
    \log\card{\supp{\rProt}}
    \Leq{(6)}
    \cc{\prot}.
\end{align*}
Here, \((1)\) is the definition of information cost, \((2)\) is the definition
of mutual information, \((3)\) is non-negativity of entropy
(\itfacts{uniform}), \((4)\) is because conditioning can only reduce entropy
(\itfacts{cond-reduce}), \((5)\) is because the uniform distribution maximizes
entropy (\itfacts{uniform}), and \((6)\) is the definition of worst-case
communication cost after padding charged transcripts to length
\(\cc{\prot}\).
\end{proof}

\paragraph{Message compression.}
We use message compression to reduce communication cost close to information
cost in certain settings. The following result is due to~\cite{HarshaJMR07}
and was strengthened slightly in~\cite{BravermanG14}; we use the textbook
presentation from~\cite{RaoY20}.

\begin{proposition}[cf.~{\cite[Theorem 7.6]{RaoY20}}]\label{prop:msg-compress}
Suppose Alice knows two distributions \(\cA,\cB\) over the same set \(U\), and
Bob knows only \(\cB\). There is a protocol for Alice and Bob to sample an
element according to \(\cA\), with no error, using public randomness and a
single message from Alice of expected length
\begin{align}
    \kl{\cA}{\cB}+2\log(1+\kl{\cA}{\cB})+O(1)
    \label{eq:kl-msg-compress}
\end{align}
bits.
\end{proposition}

Somewhat weaker bounds on the KL-divergence in
\Cref{eq:kl-msg-compress} already suffice for our purposes. Thus, to simplify
the exposition, we use
\begin{align}
    \kl{\cA}{\cB}+2\log(1+\kl{\cA}{\cB})+O(1)
    &\leq
    \kl{\cA}{\cB}+10\kl{\cA}{\cB}+O(1)
    \tag{as \(\log(1+x)\leq5x\) for every \(x\geq0\)} \\
    &\leq
    \raoconst\cdot\paren{\kl{\cA}{\cB}+1},
    \label{eq:msg-compress-final}
\end{align}
for some absolute constant \(\raoconst\geq1\), which we use throughout.

\subsection{Graph Facts}

\begin{fact}[K\H{o}nig's theorem]\label{fact:konig}
In every bipartite graph, the size of a maximum matching equals the size of a
minimum vertex cover.
\end{fact}

\end{document}